\documentclass[a4paper,12pt]{article}

\usepackage[T1]{fontenc}
\usepackage[utf8]{inputenc}
\usepackage[english]{babel}
\usepackage{enumitem}
\usepackage{lmodern}
\usepackage{amsmath}
\usepackage{amsthm}
\usepackage{amssymb}
\usepackage{mathtools}
\usepackage{scalerel}
\usepackage{graphicx}
\usepackage{tikz}
\usepackage{float}
\usepackage{caption}
\usepackage{subcaption}
\usepackage{xspace}
\usepackage[ruled, vlined, linesnumbered]{algorithm2e}
\usepackage[hidelinks]{hyperref}
\usepackage{clipboard}
\usepackage{thmtools,thm-restate} % repeat lemma in appendix

\allowdisplaybreaks[1]

\AtBeginDocument{%
	\setlength{\abovedisplayskip}{3pt plus 1pt minus 2pt}%
	\setlength{\belowdisplayskip}{3pt plus 1pt minus 2pt}%
}

\newtheorem{mytheorem}{Theorem}
\newtheorem{mycorollary}[mytheorem]{Corollary}
\newtheorem{mylemma}[mytheorem]{Lemma}
\newtheorem{mydefinition}[mytheorem]{Definition}

\usetikzlibrary{shapes}
\tikzset{
	thick/.style=      {line width=0.8pt},
	very thick/.style= {line width=3pt},
	ultra thick/.style={line width=1.6pt}
}

\newcommand{\cala}{\mathcal{A}}
\newcommand{\calb}{\mathcal{B}}
\newcommand{\cali}{\mathcal{I}}
\newcommand{\call}{\mathcal{L}}
\newcommand{\calo}{\mathcal{O}}
\newcommand{\calp}{\mathcal{P}}
\newcommand{\calw}{\mathcal{W}}
\newcommand{\calx}{\mathcal{X}}
\newcommand{\caly}{\mathcal{Y}}
\newcommand{\calz}{\mathcal{Z}}

\newcommand{\kpcst}{\mbox{$k$-PCST}\xspace}
\newcommand{\pcst}{\mbox{PCST}\xspace}
\newcommand{\kmst}{\mbox{$k$-MST}\xspace}

\newcommand{\etal}{\mbox{et~al.}\xspace}
\newcommand{\opt}{\ensuremath{\mathrm{opt}}\xspace}
\newcommand{\lp}{\ensuremath{\mathrm{LP}}\xspace}
\newcommand{\np}{\ensuremath{\mathrm{NP}}\xspace}
\newcommand{\nphard}{\mbox{\np-hard}\xspace}

\newcommand{\ecost}{c}
\newcommand{\vpen}{\pi}
\newcommand{\vweight}{w}
\newcommand{\dfunc}{\epsilon}

\colorlet{CorPotencial}{red!50!black}
\colorlet{CorLista}{blue!50!black}
\newcommand{\pot}{{{\color{CorPotencial}\lambda}}}
\newcommand{\tblist}{{{\color{CorLista}\tau}}}
\newcommand{\sublist}{{{\color{CorLista}\widetilde{\tau}}}}

\newcommand{\obj}{\sigma}
\newcommand{\actObj}{\mathcal{Q}}

\newcommand{\subminus}{{\scalebox{.4}{$\mathbf{-}$}}}
\newcommand{\subplus}{{\scalebox{.5}{$\mathbf{+}$}}}

\newcommand\mywidehat[1]{\hstretch{2}{\hat{\hstretch{.5}{#1}}}}
\newcommand{\myhat}[2][0pt]{%
	\smash{\mathrlap{\hspace*{#1}\mywidehat{\phantom{#2}}}}%
	#2%
}
\newcommand{\hatt}{\myhat[0.5pt]{T}}
\newcommand{\hath}{\myhat[0.95pt]{H}}

\newcommand{\callminusstar}{\rlap{\ensuremath{\call_\subminus}}{\phantom{\call}^*}}

\newcommand{\GW}{\mathtt{GW}}
\newcommand{\GP}{\mathtt{GP}}
\newcommand{\PP}{\mathtt{PP}}
\newcommand{\TS}{\mathtt{TS}}
\newcommand{\PV}{\mathtt{PV}}
\newcommand{\AP}{\mathtt{2\mbox{-}APPROX}}

\let\oldnl\nl% Store \nl in \oldnl
\newcommand{\nonl}{\renewcommand{\nl}{\let\nl\oldnl}}% Remove line number for one line

\usepackage{todonotes}
\usepackage{setspace}
\newcommand{\putstmttarget}[1]{%
	\ifcsname used target #1\endcsname\else%
		\global\expandafter\def\csname used target #1\endcsname{}\relax%
		\hypertarget{proofof#1}{}%
	\fi%
}

\begin{document}

%!TEX root = main.tex

% símbolos do rodapé dos emails
\renewcommand{\thefootnote}{\fnsymbol{footnote}}

\title{
	A $2$-approximation for the \\
	$k$-prize-collecting Steiner tree problem
}
\author{
	Lehilton Lelis Chaves Pedrosa \\
	Hugo Kooki Kasuya Rosado\footnote{Corresponding author.}
}
\date{
	\normalsize
	Institute of Computing -- UNICAMP -- Brazil\\
	\texttt{\{lehilton,hugo.rosado\}@ic.unicamp.br}
}

\maketitle

\begin{abstract}
We consider the $k$-prize-collecting Steiner tree problem.
An instance is composed of an integer~$k$ and a graph~$G$ with costs on edges and penalties on vertices.
The objective is to find a tree spanning at least $k$ vertices which minimizes the cost of the edges in the tree plus the penalties of vertices not in the tree.
This is one of the most fundamental network design problems and is a common generalization of the prize-collecting Steiner tree and the $k$-minimum spanning tree problems.
Our main result is a $2$-approximation algorithm, which improves on the currently best known approximation factor of $3.96$ and has a faster running time.
The algorithm builds on a modification of the primal-dual framework of Goemans and Williamson, and reveals interesting properties that can be applied to other similar problems.
\end{abstract}

\clearpage

% resetar numeração do rodapé e mudá-la para romano
\renewcommand{\thefootnote}{\arabic{footnote}}
\setcounter{footnote}{0}

{
	\small
	\renewcommand{\baselinestretch}{0.7}\normalsize
	\tableofcontents
	\clearpage
}

%!TEX root = main.tex

\section{Introduction}
\label{sec:intro}

In many network design problems, the input consists of an edge-weighted graph, and the output is a minimum-cost tree connecting a certain subset of vertices.
Two of the most fundamental \nphard variants are the \emph{prize-collecting Steiner tree} (\pcst) and the \emph{k-minimum spanning tree} (\kmst).
For \pcst, a solution may contain any subset of vertices, but any not spanned vertex incurs a penalty which is added to the objective function.
For \kmst, the output tree is required to contain at least $k$ vertices.

We consider the \emph{$k$-prize-collecting Steiner tree problem} (\kpcst), which is a common generalization of \pcst and \kmst.
An instance consists of a connected undirected graph ${G=(V,E)}$, a special vertex $r$, called the root, and a non-negative integer ${k \leq |V|}$.
Each edge ${e \in E}$ has a non-negative cost~$\ecost_e$, and each vertex ${v \in V}$ has a non-negative penalty~$\vpen_v$.
A solution is a tree spanning at least $k$ vertices, including the root~$r$, and minimizing the cost of edges of the tree plus the penalties of vertices not spanned by the tree.
Without loss generality, we assume that ${\vpen_r = \infty}$.

	%!TEX root = main.tex

\subsection{Related works}
\label{subsec:related_works}

\pcst is the special case of \kpcst for which $k = 0$.
For this problem, Bienstock~\etal presented a $3$-approximation based on an \lp rounding algorithm~\cite{Bienstock1993}, and Goemans showed that this factor could be improved to $2.54$, by randomizing the rounding algorithm~\cite{WilliamsonS2011}.
Goemans and Williamson presented a $2$-approximation~\cite{Goemans1995} based on a new general primal-dual scheme for \pcst as well many other constrained forest problems.
To bound the value of an optimal solution, they used weak duality, and, as a consequence, their analysis implies that the integrality gap of the usual \lp formulation is asymptotically tight.
Currently, the best approximation for \pcst is due to Archer \etal~\cite{ArcherBHK2011} and it has factor $2-(\frac{2-\rho}{2+\rho})^2$, where $\rho$ is an approximation factor for the Steiner tree problem.
Using the best-known value for $\rho$, which is $\ln(4)+\epsilon$~\cite{ByrkaGRS2013},  yields a factor of $1.9672+\epsilon$ for \pcst.

\kmst is the special case of \kpcst for which $\vpen_v = 0$ for each vertex~$v$.
One may assume that a solution spans exactly $k$ vertices, since a tree with more than~$k$ vertices can be pruned without increasing its cost.
Several approximation algorithms were devised for \kmst~\cite{Awerbuch1998,Chudak2004,Ravi1996}, and Blum \etal~\cite{Blum1999} gave the first constant-factor approximation, with factor $17$, by using the primal-dual scheme.
Garg improved this factor to $5$ and $3$, subsequently~\cite{Garg1996}, and Arya and Ramesh showed how Garg's algorithm could be used to obtain a factor of $2.5$~\cite{Arya1998}. Later, Arora and Karakostas gave a $(2+\epsilon)$-approximation by modifying Garg's algorithm~\cite{Arora2006}.
Currently, the best approximation for \kmst is a $2$-approximation due to Garg~\cite{Garg2005} and is based on a sophisticated use of the primal-dual scheme.

% tá em português, melhor não botar pq só ia confundir: \cite{Oshiro2010}

To our knowledge, the first constant-factor approximation for \kpcst was given by Han~\etal~\cite{Han2017} and has factor~$5$.
They presented a primal-dual algorithm based on the Lagrangean relaxation
of a linear program.
Later, Matsuda and Takahashi~\cite{Matsuda2019} derived a $4$-approximation by combining the solutions for the underlying instances of \pcst and \kmst.
The algorithm's running time is $\calo(|V|^4 |E| \log |V|)$
and is bottlenecked by Garg's $2$-approximation, which is used to solve \kmst.
By using the $1.9672+\epsilon$-approximation for \pcst, the approximation factor for \kpcst can be improved to $3.9672+\epsilon$, with a significant increase in the running time.

	%!TEX root = main.tex

\subsection{Our results}
\label{subsec:our_results}

Our main contribution is a $2$-approximation for \kpcst.
More precisely, we present an algorithm with running time $\calo(|V|^2|E|^2 + |V|^4 \log^2|V|)$  that finds a tree $T$ such that ${\sum_{e \in E(T)}\ecost_e + 2 \cdot \sum_{v \in V \setminus V(T)} \vpen_v \leq 2 \cdot \opt}$, where $\opt$ is the optimal value.
This improves on both the approximation factor and the time complexity of the previously best-known algorithms.
Our $2$-approximation is based on a modified version of the Goemans and Williamson's algorithm, and our analysis reveals many interesting properties of the primal-dual scheme, which might give insights to other problems with similar constraints.
We note that the inequalities considered by our algorithm do not correspond to a dual \lp formulation for \kpcst, and our analysis does not rely on weak duality.
Moreover, a small modification of the algorithm results in a $2$-approximation for the quota variant~\cite{Johnson2000} of \kpcst, for which an instance includes vertex weights, and a solution is any tree whose weight is at least the given quota.

	%!TEX root = main.tex

\subsection{Algorithm's overview}
\label{subsec:alg_verview}

Our algorithm successively executes a modified version of the $2$-approximation primal-dual scheme for \pcst due to Goemans and Williamson~\cite{Goemans1995}.
Their algorithm is divided into a \emph{growth-phase} and a \emph{pruning-phase}.
In the growth-phase, it computes a feasible dual solution $y$ such that for each subset of vertices~$S$, $y_S$ is a non-negative value.
It also outputs a tree~$T$ and a collection~$\calb$ of subsets of $V$, whose edges and subsets correspond to tight dual inequalities of an \lp formulation.
In the pruning-phase, the algorithm deletes from $T$ the subsets in~$\calb$ which do not disconnect the graph, resulting in a pruned tree~$\hatt$.
To derive a $2$-approximation, they bound the value of $\hatt$ by a factor of the dual objective function; in our algorithm, we compare with an optimal solution directly.

In our modification, the growth-phase receives two new arguments, a potential $\pot$ and a tie-breaking list $\tblist$.
The potential is a uniform increase on the penalties of each vertex, such that for larger values of $\pot$, the output tree~$\hatt$ spans more vertices.
During the growth-phase, there might be concurrent events, thus there are multiple choices for the execution path.
Usually, these choices are determined by some fixed lexicography order.
Our algorithm, on the contrary, relies on a tie-breaking list $\tblist$ to control the priority among concurrent events.
The $i$-th element in this list dictates which event gets the highest priority in the $i$-th iteration of the algorithm. This allows us to control the execution path of the algorithm.

The use of potential $\pot$ is built on Garg's arguments for the $2$-approximation for \kmst~\cite{Garg2005}, which can be described as follows.
If, for some $\pot$, the pruned tree $\hatt$ spans exactly $k$ vertices, this leads to a $2$-approximation by using the Lagrangean relaxation strategy (see, e.g.,~\cite{Chudak2004}).
However, it might be the case that no such $\pot$ exists; thus, the idea is to find a particular value of $\pot$ such that, for sufficiently small $\epsilon$,
using potential $\pot-\epsilon$ leads to a pruned tree $\hatt_\subminus$ spanning less than $k$ vertices, and using potential $\pot+\epsilon$ leads to a pruned tree~$\hatt_\subplus$ spanning at least than $k$ vertices.
The tree~$\hatt_\subminus$ is constructed by pruning a tree~$T_\subminus$ using a collection~$\calb_\subminus$. Similarly, $\hatt_\subplus$ is constructed by pruning a tree~$T_\subplus$ using a collection~$\calb_\subplus$.
On the one hand, $\hatt_\subplus$ is a feasible solution, but its cost cannot be bounded in terms of vector $y$. On the other hand, the value of $\hatt_\subminus$ can be bounded, but it is not a feasible solution.

In Garg's algorithm, the trees $T_\subminus$ and $T_\subplus$ and the collections $\calb_\subminus$ and $\calb_\subplus$ might be very different. Thus, to obtain a tree with $k$ vertices and whose cost can be bounded, his algorithm iteratively transforms $T_\subminus$ into $T_\subplus$ and $\calb_\subminus$ into $\calb_\subplus$ by replacing one edge of $T_\subminus$ or one subset in $\calb_\subminus$ at a time.
At some iteration, pruning the current tree using the current collection must result in a tree spanning at least $k$ vertices. Before this step, instead of performing the operation, one augments the current pruned tree by adding a sequence of edges whose corresponding dual restrictions are tight, picking up to $k$ vertices.

Our algorithm also considers similar trees $T_\subminus$ and $T_\subplus$ and corresponding collections $\calb_\subminus$ and $\calb_\subplus$.
However, both trees are constructed by executing the growth-phase using a single potential $\pot$.
To differentiate between the cases, we take into account a tie-braking list $\tblist$
and its maximal proper prefix $\sublist$.
We show how to compute a special tuple $(\pot,\tblist)$, called the \emph{threshold-tuple}, such that executing the growth-phase using $\tblist$ results in a pruned tree with less than $k$ vertices, while using~$\sublist$ results in a pruned tree with at least $k$ vertices (or vice-versa).

The trees and collections output by the two executions of the growth-phase are only slightly different.
Indeed, a key ingredient of our analysis (presented in Lemma~\ref{subsec:threshold-tuple:lemma:main_lemma}) shows that one of the two following scenarios hold:
\begin{enumerate}[label=(\roman*),itemsep=-2pt,topsep=-6pt]
	\item collections $\calb_\subminus$ and $\calb_\subplus$ are equal, and trees $T_\subminus$ and $T_\subplus$ differ in exactly one edge, or
	\item trees $T_\subminus$ and $T_\subplus$ are equal, and collections $\calb_\subminus$ and $\calb_\subplus$ differ in exactly one subset.
\end{enumerate}
Moreover, we show that the vector $y$ output in both executions of the growth-phase are the same.
This leads to a straightforward way of augmenting the pruned tree $\hatt_\subminus$, by picking a sequence of edges of $T_\subminus$ or $T_\subplus$ whose corresponding inequalities are tight, without the need for a step-by-step transformation.

Although the computed vector~$y$ satisfy a set of inequalities, these inequalities do not correspond to an \lp dual formulation for \kpcst, hence cannot use weak duality to bound the value of an optimal solution.
Instead, we show (in Lemma~\ref{subsec:approx:lemma:2-approx_or_reduce}) that either our algorithm returns a $2$-approximate solution, or it identifies a non-empty subset of the vertices which are not spanned by any optimal solution.
Thus, we either find the desired solution, or can safely reduce the size of the instance.
Therefore, by running the algorithm at most $|V|-k$ times, we find a $2$-approximate solution.

	%!TEX root = main.tex

\subsection{Text organization}
\label{subsec:text_org}

The remaining sections are organized as follows.
In Section~\ref{sec:preliminaries}, we give a series of definitions and introduce the terminology used in the text.
In Section~\ref{sec:GW}, we describe our modification of the primal-dual scheme, which accounts for the potential $\pot$ and the tie-breaking list $\tblist$.
In Section~\ref{sec:threshold-tuple}, we formally define the threshold-tuple, and show how it can be computed.
In Section~\ref{sec:find_solution}, we show that, given a threshold-tuple, one can construct a tree spanning exactly $k$ vertices.
In Section~\ref{sec:two-approx}, we bound the cost of the computed tree and give a $2$-approximation for \kpcst.
In Section~\ref{sec:final_remarks}, we give some concluding remarks.

%!TEX root = main.tex

\section{Definitions and preliminaries}
\label{sec:preliminaries}

We say that a nonempty collection $\call \subseteq 2^V$ is \emph{laminar} if, for any two subsets ${L_1, L_2 \in \call}$, either $L_1 \cap L_2 = \emptyset$, or $L_1 \subseteq L_2$, or $L_2 \subseteq L_1$.
A laminar collection~$\call$ is \emph{binary} if for every $L \in \call$ with $|L| \geq 2$, there are distinct non-empty subsets ${L_1,L_2 \in \call}$ such that ${L = L_1 \cup L_2}$.
We denote the collection of inclusion-wise maximal subsets of a collection $\call$ by $\call^*$.
Observe that, if $\call$ is laminar, then the subsets in $\call^*$ are disjoint.

Let $\calp$ be a partition of $V$, and consider an edge $e$ with extremes on $V$.
If a set in $\calp$ contains an extreme of $e$, then we call this set an \emph{endpoint} of~$e$.
We say that an edge $e$ is \emph{internal} in $\calp$ if $e$ has only one endpoint,
and we say that $e$ is \emph{external} in $\calp$ if $e$ has two distinct endpoints.
Also, two external edges are said to be \emph{parallel} in $\calp$ if they have the same pair of endpoints.

Given a graph $H$ and a subset $L \subseteq V$, we say that $H$ is \emph{$L$-connected} if $V(H) \cap L = \emptyset$ or if the induced subgraph $H[V(H) \cap L]$ is connected.
For a collection $\call$ of subsets of the vertices, we say that $H$ is \emph{$\call$-connected} if $H$ is $L$-connected for every $L \in \call$.

Let $L \subseteq V$ and $H \subseteq G$, then $\delta_H(L)$ denotes the set of edges of $H$ with exactly one extreme in $L$. We say that $L$ has \emph{degree} $|\delta_H(L)|$ on $H$.
In the case where $H = G$, we drop the subscript and write just $\delta(L)$.

For a subset $L \subseteq V$, we define its new penalty as ${\vpen^\pot_L = \sum_{v \in L} \vpen_v + \pot|L|}$, where $\pot$ is a non-negative value which we call \emph{potential}.
This implies that, for any subset $L$ containing the root $r$, we have $\vpen^\pot_L = \infty$.

Consider a vector~$y$ such that, for each ${L \subseteq 2^V}$, the entry $y_L$ is a non-negative variable.
We say that $y$ \emph{respects} edge cost $\ecost$ if
\begin{align}
	\sum_{L: e \in \delta(L)} y_L \leq \ecost_e \quad\quad\text{for every edge $e \in E$}, \label{sec:preliminaries:restriction_1}
\end{align}
and we say that $y$ \emph{respects} vertex penalty $\vpen^\pot$ if
\begin{align}
	\sum_{S:S \subseteq L} y_S \leq \vpen^\pot_L \quad\quad\text{for every subset $L \subseteq V$}. \label{sec:preliminaries:restriction_2}
\end{align}

We say that an edge~$e$ is \emph{tight} for $(y,\pot)$ if the inequality corresponding to~$e$ in~(\ref{sec:preliminaries:restriction_1}) holds with equality.
Analogously, a subset~$L \subseteq V$ is tight for~$(y,\pot)$ if the inequality corresponding to~$L$ in (\ref{sec:preliminaries:restriction_2}) is satisfied with equality.
If the pair $(y,\pot)$ is clear from context, we simply say that $e$ and $L$ are tight.

Inequalities (\ref{sec:preliminaries:restriction_1}) and (\ref{sec:preliminaries:restriction_2}) are similar to the inequalities in the dual formulation for \pcst~\cite{Goemans1995}, with the difference that we include inequalities for subsets containing $r$.
Also, we note that these inequalities may not correspond to the dual of an LP formulation for \kpcst.

For a collection of subsets $\call$, denote by $\call[S]$ the collection of subsets in~$\call$ which are subsets of $S$.
Also, denote by $\call(S)$ the collection of subsets in $\call$ which contain some, but not all, vertices of $S$.
Moreover, let $\ecost_{E'} = \sum_{e \in E'} \ecost_e$ for $E'\subseteq E$, and $\vpen_L = \sum_{v \in L} \vpen_e$ for $L \subseteq V$.
To bound the value of an optimal solution, we use the following lemma.

\begin{mylemma}
	\label{subsec:approx:lemma:opt_lower_bound}
	Let $T^*$ be an optimal solution. Suppose that $\call$ is a laminar collection and that $L^*$ is the minimal subset in $\call$ containing $V(T^*)$.
	If $y$ respects~$\ecost$~and~$\vpen^\pot$, then
	\[
	% \textstyle
	\sum_{L \in \call(L^*)} y_L - \pot|L^* \setminus V(T^*)| \leq \ecost_{E(T^*)} + \vpen_{L^* \setminus V(T^*)}.
	\]
\end{mylemma}

\begin{proof}
	Since $\call$ is laminar, each subset in $\call(L^*)$ contains some, but not all, vertices of $V(T^*)$, or it is a subset of $L^* \setminus V(T^*)$.
	Also, one subset in $\call(L^*)$  cannot be~$V(T^*)$, because otherwise $L^*$ would not be minimal.
	Thus,
	\begin{align*}
	\textstyle{\sum_{L \in \call(L^*)}} y_L
	&=    \textstyle {\sum_{L \in \call(V(T^*))}} y_L
	               + {\sum_{L \in \call[L^* \setminus V(T^*)]}} y_L\\
	&\le  \textstyle {\sum_{e \in E(T^*)} {\sum_{L : e \in \delta(L)}} y_L}
	               + {\sum_{L \in \call[L^* \setminus V(T^*)]}} y_L\\
	&\le  \textstyle {\sum_{e \in E(T^*)} \ecost_e}
		             + {\vpen^\pot_{L^* \setminus V(T^*)}}\\
	&=    \textstyle {\ecost_{E(T^*)}}
	               + \vpen_{L^* \setminus V(T^*)}
	               + \pot |L^* \setminus V(T^*)|.
	\end{align*}
	The first inequality holds as each subset in $\call(V(T^*))$ is crossed by at least one edge of $T^*$, and the second inequality holds because~$y$ respects~$\ecost$ and~$\vpen^\pot$.
\end{proof}

%!TEX root = main.tex

\section{Modified growth and pruning phases}
\label{sec:GW}

In the following, we detail our modification of the primal-dual scheme due to Goemans and Williamson for the prize-collecting Steiner tree problem~\cite{Goemans1995}. The algorithm is composed of two main routines: a clustering algorithm, also known as the growth-phase, and a cleanup algorithm, known as the pruning-phase.

	%!TEX root = main.tex

\subsection{Modified clustering algorithm}

The modified \emph{growth-phase} is described in the following and is denoted by $\GP(\pot,\tblist)$. A listing of all steps is given afterwards, in Algorithm~\ref{subsec:GP:alg:GP}.
The algorithm maintains a binary laminar collection ${\call \subseteq 2^V}$, such that $\call^*$ partitions the set of vertices~$V$, and a vector $y$ which respects~$\ecost$ and~$\vpen^\pot$.
It iteratively constructs a forest ${F \subseteq G}$ and a subcollection of \emph{processed} subsets $\calb \subseteq \call$, such that the edges of $F$ and the subsets of $\calb$ are tight for~$(y, \pot)$. In each iteration, either a new edge is added to $F$, or a new subset is included in $\calb$.

The algorithm begins by defining $\call = \{\{v\} : v \in V\}$ and ${y_S = 0}$, for each ${S \subseteq V}$ (implicitly), and by letting $F = (V, \emptyset)$, and $\calb = \emptyset$. Once initialized, it starts the iteration process.
At a given moment, a maximal subset ${L \in \call^*}$ is said to be \emph{active} if it has not been processed yet, i.e., if~${L \in \call^* \setminus \calb}$.
In each iteration, we increase the value $y_L$ of every active subset~$L$ uniformly until one of the following events occur:
\begin{enumerate}[itemsep=0pt,topsep=0pt]
	\item[$\triangleright$] an external edge~$e$ with endpoints $L_1, L_2 \in \call^*$ becomes tight, in which case $e$ is added to $F$, and the union $L_1 \cup L_2$ is included in $\call$; or

	\item[$\triangleright$] an active subset~$L$ becomes tight, in which case $L$ is included in $\calb$.
\end{enumerate}

We note that multiple edges and subsets might become tight simultaneously.
In our modified algorithm, we use the \emph{tie-breaking list} $\tblist$ to decide the order in which the events are processed.
A tie-breaking list~$\tblist$ with size $|\tblist|$ is a (possibly empty) sequence of edges and subsets. For each $i = 1, 2, \dots, |\tblist|$, the $i$-th element of the list is denoted by $\tblist_i$.
In iteration $i$, the event to be processed is determined according to the following order:
\begin{enumerate}[itemsep=0pt,topsep=0pt, label=(\roman*),ref={\mbox{\rm $(\roman*)$}}]
	\item \label{sec:GW:priority1} if $i \leq |\tblist|$, then the event corresponding to $\tblist_i$ has the highest priority;
	\item \label{sec:GW:priority2} followed by events corresponding to edges;
	\item \label{sec:GW:priority3} and finally by events corresponding to subsets.
\end{enumerate}
The priority order between events of the same type are determined by a fixed lexicographic order.

The algorithm stops when $V$ is the only active subset in~$\call^*$, at which point $F$ is a tree.
The algorithm defines ${T = F}$, and outputs the pair~$(T, \calb)$.

\begin{algorithm}[p]
	\DontPrintSemicolon
	\caption{$\GP(\pot,\tblist)$}
	\label{subsec:GP:alg:GP}

	Initialize $F \gets (V, \emptyset)$, $\call \gets \{\{v\}:v \in V\}$, $\calb \gets \emptyset$, and $y \gets 0$\;

	\For{\textnormal{$i \gets 1, 2, \dots$}}
	{
		Let $\cala \gets \call^* \setminus \calb$\;
		\nonl\;

		\If{\textnormal{$\cala = \{V\}$}}
		{
			% Let $R \in \cala$ and $T = F[R]$\;
			Let $T \gets F$\;
			\Return{$(T,\calb)$}\;
		}
		\nonl\;

		Let $\Delta_1 \gets \Delta_2 \gets 0$\;
		\If{\textnormal{there are no tight edges external in $\call^*$}}
		{
			Let $\Delta_1 \gets \min \, \left\{ \frac{\ecost_e - \sum_{L: e \in \delta(L)} y_L }{|\{L \in \cala\,:\, e \in \delta(L)\}|} : e \in \delta(S), \, S \in \cala \right\}$\;
		}
		\If{\textnormal{there are no tight active subsets}}
		{
			Let $\Delta_2 \gets \min \, \left\{ \vpen^\pot_L - \sum_{S:S \subseteq L} y_S : L \in \cala \right\}$\;
		}

		Let $\Delta \gets \min\{\Delta_1, \Delta_2\}$ and increase $y_L$ by $\Delta$ for every $L \in \cala$\label{alg:mGP_delta_icrease}\;

		\nonl\;

		Let $\caly$ be the set of tight edges external in $\call^*$\;
		Let $\calz$ be the collection of tight active subsets\;

		\uIf{\textnormal{$|\tblist| \ge i$ and $\tblist_i \in \caly \cup \calz$}}
		{
			Let $\obj \gets \tblist_i$
		}
		\uElseIf{$\caly \ne \emptyset$}
		{
			Let $\obj$ be the first edge of $\caly$
		}
		\Else
		{
			Let $\obj$ be the first subset of $\calz$
		}
		\nonl\;

		\uIf(\label{alg:mGP_process_edge}){\textnormal{$\obj$ is a tight external edge in $\call^*$}}
		{
			Let $L_1$ and $L_2$ be the endpoints of $\obj$ in $\call^*$\;
			Add $\obj$ to $F$  and include $L_1 \cup L_2$ in $\call$
		}
		\ElseIf(\label{alg:mGP_process_subset}){\textnormal{$\obj$ is a tight active subset}}
		{
			Include $\obj$ in $\calb$\;
		}
	}
\end{algorithm}

Next lemma collects basic invariants of the growth-phase.

\begin{mylemma}
    \label{subsec:GP:lemma:GP_inv}
		At the beginning of any iteration of $\GP(\pot,\tblist)$, the following holds:
		\begin{enumerate}[label=(gp\arabic*),ref={\mbox{\rm $(gp\arabic*)$}},topsep=0pt,itemsep=0pt]
			\item \label{subsec:GP:lemma:GP_inv:L_is_laminar} $\call$ is a binary laminar collection, $\bigcup_{L \in \call^*} L = V$, and $\emptyset \notin \call$;
			\item \label{subsec:GP:lemma:GP_inv:y_respects} $y$ respects~$\ecost$ and~$\vpen^\pot$;
			\item \label{subsec:GP:lemma:GP_inv:edge_in_F_are_tight} $F$ is an $\call$-connected forest and every edge $e \in E(F)$ is tight for $(y,{\pot})$;
			\item \label{subsec:GP:lemma:GP_inv:subsets_in_B_are_tight} $\calb \subseteq \call$ is laminar and every subset in $\calb$ is tight for $(y,{\pot})$;
			\item \label{subsec:GP:lemma:GP_inv:root_never_tight} no subset in $\calb$ contains the root $r$.
		\end{enumerate}
\end{mylemma}

\begin{proof}
	Before the first iteration, $\call = \{\{v\} : v \in V\}$, hence~\ref{subsec:GP:lemma:GP_inv:L_is_laminar} is valid.
	Assume that the invariant is valid at the beginning of an iteration and notice that a new subset $L$ is included into $\call$ only if $L = L_1 \cup L_2$, where $L_1, L_2 \in \call^*$. Therefore,~\ref{subsec:GP:lemma:GP_inv:L_is_laminar} is also valid at the end of the iteration.

	Before the first iteration, we have $y = 0$, thus~\ref{subsec:GP:lemma:GP_inv:y_respects} is valid since~$\ecost_e$ and~$\vpen^\pot_v$ are non-negative for every edge $e$ and vertex $v$.
	Assume that the invariant is valid at the beginning of an iteration.
	In this iteration, each variable $y_L$ corresponding to an active subset~$L$ is increased by the minimum value of~$\Delta$ such that an edge external in $\call^*$ or an active subset becomes tight.
	Observe that the only edges that can become tight must be external.
	Moreover, if a subset $L\subseteq V$ becomes tight, so does some subset $L \in \call^*$, because $\call$ is laminar.
	Therefore, after modifying $y$, no inequality is violated, and the invariant remains valid.

	Before the first iteration, \ref{subsec:GP:lemma:GP_inv:edge_in_F_are_tight} is valid because $\call$ is composed of singletons, and $F$ contains no edges.
	Assume that the invariant is valid at the beginning of an iteration, and notice that an edge $e$ is included into $F$ only if it is tight and the union $L_1 \cup L_2$ of its endpoints $L_1,L_2 \in \call^*$ is included into $\call$.
	Since $F$ is $L_1$-connected and $L_2$-connected, it follows that $F$ is also $(L_1 \cup L_2)$-connected because $e$ was added to $F$. Therefore,~\ref{subsec:GP:lemma:GP_inv:edge_in_F_are_tight} remains valid at the end of the iteration.

	Invariant~\ref{subsec:GP:lemma:GP_inv:subsets_in_B_are_tight} holds because $\call$ is laminar, and only tight active subsets are included in $\calb$.
	For invariant~\ref{subsec:GP:lemma:GP_inv:root_never_tight}, observe that in any iteration, if $L \in \call$ is the active component containing the root $r$, then $\vpen^\pot_L = \infty$. It follows that $L$ is not processed, and thus it is not included in $\calb$.
\end{proof}

Observe that, since $G$ is connected and the algorithm only ends when $V$ is the only active component, by invariant~\ref{subsec:GP:lemma:GP_inv:edge_in_F_are_tight} $T$ spans $V$.

\begin{mycorollary}
	\label{subsec:GP:corollary:GP_output}
	Let $(T,\calb)$ be the pair output by $\GP(\pot,\tblist)$.
	Then $V(T) = V$.
\end{mycorollary}

\begin{mylemma}
	\label{subsec:GP:lemma:GP_iterates_at_most_3V-3}
	The number of iterations executed by $\GP(\pot,\tblist)$ is at most $3|V|-3$.
\end{mylemma}

\begin{proof}
	Since $F$ is a forest (invariant~\ref{subsec:GP:lemma:GP_inv:edge_in_F_are_tight}), the number of processed edges is at most $|V|-1$.
	Observe that the size of $\call$ is initially $|V|$ and increases by~$1$ for each processed edge.
	Thus, the final size of $\call$ is at most ${2 |V| - 1}$.
	Each subset of $\call$ which does not contain the root~$r$ can be processed.
	Therefore, the total number of processed events is at most $3|V| - 3$.
\end{proof}

A consequence is that the growth-phase executes in polynomial time. Also, this implies that the size of a tie-breaking list need not exceed~${3|V|-3}$.

	%!TEX root = main.tex

\subsection{Modified pruning algorithm}

The modified \emph{pruning-phase} is denoted by $\PP(H,\calb)$. A listing is given in Algorithm~\ref{subsec:PP:alg:PP}.
The algorithm receives a graph $H$ and iteratively deletes from it any processed subset $B \in \calb$ such that the degree of $B$ on $H$ is one.
We assume that the input $H$ is a connected subgraph of $G$ containing~$r$, and that $\calb$ is a laminar collection of subsets of $V$ which do not contain~$r$.

\begin{algorithm}[H]
	\DontPrintSemicolon
	\caption{$\PP(H,\calb)$}
	\label{subsec:PP:alg:PP}

	\While{\label{pp-while}\textnormal{there is a subset $B$ in $\calb$
			such that $|\delta_H(B)| = 1$}}
	{
		Delete $B$ from $H$\;
	}
	\Return{$H$}\;
\end{algorithm}

In the following, we might say that we \emph{prune}~$H$~using~$\calb$ to mean that we execute algorithm $\PP(H,\calb)$, and say that a graph $H$~is~\emph{pruned}~with~$\calb$ if $|\delta_H(B)| \neq 1$ for every ${B \in \calb}$.

Note that we allow input graphs $H$ with cycles, whereas the standard pruning-phase only considers trees.
This will be useful in Section~\ref{subsec:picking}, when, to find a tree with $k$ vertices from distinct trees $T_\subminus$ and $T_\subplus$, we will first prune $T_\subminus \cup T_\subplus$, which might contain a cycle.

Next lemma collects invariants of the pruning-phase.

\begin{mylemma}
	\label{subsec:PP:lemma:PP_inv}
	At the beginning of any iteration of $\PP(H,\calb)$, the following holds:
	\begin{enumerate}[label=(pp\arabic*),ref={\mbox{\rm $(pp\arabic*)$}},itemsep=-2pt,topsep=-6pt]
		\setlength{\itemsep}{0pt}
		\item \label{subsec:PP:lemma:PP_inv:1} $H$ is connected and $r \in V(H)$;
		\item \label{subsec:PP:lemma:PP_inv:2} $H[V(H) \cap B]$ is connected.
	\end{enumerate}
\end{mylemma}

\begin{proof}
	Observe that $H$ is connected and $r \in V(H)$ when the algorithm starts. Consider the subset $B$ chosen at the beginning of an iteration, and assume that~\ref{subsec:PP:lemma:PP_inv:1} is valid. Since ${|\delta_H(B)| = 1}$ and $r \notin B$, deleting $B$ does not disconnect $H$ nor removes the root~$r$ from~$H$. Thus, \ref{subsec:PP:lemma:PP_inv:1} holds at the end of the iteration.
	For~\ref{subsec:PP:lemma:PP_inv:2}, observe that, since $H$ is connected and $|\delta_H(B)| = 1$, $V(H) \cap B$ must induce a connected subgraph.
\end{proof}

\begin{mycorollary}
	\label{subsec:PP:coro:PP_works}
	Let $H'$ be the graph output by $\PP(H,\calb)$.
	Then, ${H' \subseteq H}$, $H'$~is~pruned with~$\calb$, and $r \in V(H)$.
\end{mycorollary}

\begin{proof}
	By construction,  $H' \subseteq H$ and $H'$ is pruned with~$\calb$. Invariant~\ref{subsec:PP:lemma:PP_inv:1} implies $r \in H'$.
\end{proof}

Next, we give a structural result about the pruning-phase.
It is an auxiliary lemma which implies a series of monotonic properties of the pruning algorithm.

\begin{mylemma}
	\label{subsec:PP:lemma:structure}
	Consider connected graphs $D$ and $H$. Assume that $D$ is pruned with~$\calb$ and let $H'$ be the graph output by $\PP(H,\calb)$.
	If $D \subseteq H$, then $D \subseteq H'$.
\end{mylemma}

\begin{proof}
	Assume that in the execution that output~$H'$, the algorithm executed~$\ell$ iterations, and let $B_i$ and $H_i$ be the values assigned to variables $B$~and~$H$ at the beginning of iteration~$i$. Also, let $H_{\ell+1} = H'$.

	We show that $D \subseteq H_i$ by induction on $i$. Clearly $D  \subseteq H = H_1$. Then, assume that $D \subseteq H_i$ for some $i \ge 1$.
	It follows that $|\delta_{D}(B_i)| \le |\delta_{H_i}(B_i)|$. But $|\delta_{H_i}(B_i)| = 1$ by the choice of $B_i$, and so $|\delta_{D}(B_i)| \le 1$.
	Since $D$ is pruned with $\calb$, we know that $|\delta_{D}(B_i)| \ne 1$, and thus $|\delta_{D}(B_i)| = 0$.
	This implies that $B_i \cap V(D) = \emptyset$, because $D$ is connected.
	To complete the induction, observe that $D = D  - B_i \subseteq H_i - B_i = H_{i+1}$.
\end{proof}

The first property states that pruning using~$\calb$ determines a unique pruned graph. An implication of this lemma is that the output of $\PP(H,\calb)$ is invariant to the order in which subsets of $\calb$ are considered.

\begin{mycorollary}
	\label{subsec:PP:coro:PP_order}
	Let $H_1$ and $H_2$ be the graphs output by two executions of $\PP(H,\calb)$. Then $H_1 = H_2$.
\end{mycorollary}

\begin{proof}
	Observe that $H_1$ is connected and pruned with~$\calb$, and that $H_2$ is the output of $\PP(H,\calb)$. Since $H_1 \subseteq H$, by Lemma~\ref{subsec:PP:lemma:structure}, $H_1 \subseteq H_2$.
	Symmetrically, we have $H_2 \subseteq H_1$, and thus $H_1 = H_2$.
\end{proof}

Next corollary states that pruning a fixed graph $H$ using distinct collections reverses the subcollection relation.

\begin{mycorollary}
	\label{subsec:PP:coro:PP_subcollection_pruning}
	Let $H_1$ and $H_2$ be the graphs output by $\PP(H,\calb_1)$ and $\PP(H,\calb_2)$, respectively. If $\calb_1 \subseteq \calb_2$, then $H_2 \subseteq H_1$.
\end{mycorollary}

\begin{proof}
	Observe that $H_2$ is connected and pruned with~$\calb_2$, and thus it is pruned with~$\calb_1$ as well.
	Since ${H_2 \subseteq H}$, by Lemma~\ref{subsec:PP:lemma:structure}, $H_2 \subseteq H_1$.
\end{proof}

The last corollary considers graphs $D$ which are ``sandwiched'' between a graph $H$ and the pruned subgraph $H'$.
Intuitively, one may interpret this corollary as stating that pruning a partially pruned graph $D$ leads to the fully pruned graph $H'$.

\begin{mycorollary}
	\label{subsec:PP:coro:PP_subgraph_pruning}
	Consider connected graphs $D$ and $H$. Let $D'$ and $H'$ be the graphs output by $\PP(D,\calb)$ and $\PP(H,\calb)$, respectively.
	If $H' \subseteq D \subseteq H$, then $H' = D'$.
\end{mycorollary}

\begin{proof}
	Observe that $H'$ is connected and pruned with~$\calb$.
	Since ${H' \subseteq D}$, by Lemma~\ref{subsec:PP:lemma:structure}, $H' \subseteq D'$.
	Analogously, observe that $D'$ is connected and pruned with~$\calb$.
	Since ${D' \subseteq D \subseteq H}$, by Lemma~\ref{subsec:PP:lemma:structure}, $D' \subseteq H'$.
	Therefore, $H = D$.
\end{proof}
	%!TEX root = main.tex

\subsection{Modified Goemans-Williamson algorithm}

The modified \emph{Goemans-Williamson} algorithm wraps up the growth and the pruning-phases and is denoted by $\GW(\pot,\tblist)$. A listing is given in Algorithm~\ref{subsec:GW:alg:GW}.
First, the algorithm executes $\GP(\pot,\tblist)$ to obtain a pair $(T,\calb)$. Then, it executes $\PP(T,\calb)$ and returns the pruned tree $\hatt$.

\begin{algorithm}[H]
	\DontPrintSemicolon
	\caption{$\GW(\pot,\tblist)$}
	\label{subsec:GW:alg:GW}

	Let $(T,\calb)$ be the tuple returned by $\GP(\pot,\tblist)$\;
	Let $\hatt$ be the tree returned by $\PP(T,\calb)$\;

	\Return{$\hatt$}\;
\end{algorithm}

One interesting property of $\GW(\pot,\tblist)$ is that if $\pot$ it too large, then no subset is ever processed, and the returned tree spans the whole set of vertices.
Recall that $\ecost_{E} = \sum_{e \in E} \ecost_e$. Next lemma states that $\ecost_E$ is sufficiently large.

\begin{mylemma}
	\label{subsec:GW:lemma:high_lambda_spans_V_G_r}
	Let $\hatt$ be the output of $\GW(\pot, \tblist)$. If $\pot > \ecost_{E}$, then ${V(\hatt) = V}$.
\end{mylemma}

\begin{proof}
	We claim that $\calb = \emptyset$.
	If this is the case, by Corollary~\ref{subsec:GP:corollary:GP_output}, we have that $\GP(\pot, \tblist)$ returns the tuple $(T,\emptyset)$, thus no subset is deleted in the pruning-phase, and we get $\hatt = T$ and $V(\hatt) = V$.

	Now, we show that $\calb = \emptyset$.
	Suppose, for a contradiction, that the first processed subset is $L$ and it is processed in iteration $\ell$.
	For a vertex $v$, denote by~$\Gamma_v$ the value of $\sum_{S: v \in S} y_S$ at the end of iteration $\ell$. Since no subset was processed before iteration $\ell$, each vertex is contained in exactly one active subset in each iteration $i \le \ell$.
	%, and each such an active subset is increased by the same value.
	%
	It follows that, for any two vertices $u$~and~$v$, we have $\Gamma_u = \Gamma_v$.

	Since no subset containing the root~$r$ is ever processed (invariant~\ref{subsec:GP:lemma:GP_inv:root_never_tight}), $r \notin L$, and since $G$ is connected, there must exist an external edge $e$ with extremes $u \in L$ and $w \notin L$.
	Since $y$ respects~$\ecost$ (invariant~\ref{subsec:GP:lemma:GP_inv:y_respects}), and no $S \in \call$ contains both $u$ and $w$ at the end of iteration $\ell$, we have
	\[
	\Gamma_u + \Gamma_w
	=
	\sum_{S: e \in \delta(S)} y_S
	\le
	\ecost_e.
	\]
	Then, since $L$ becomes tight at iteration~$\ell$,
	\[
	\vpen_L^\pot
	=
	\sum_{S:S \subseteq L} y_S
	\le
	\sum_{v \in L} \Gamma_v
	=
	|L| \,  \Gamma_u
	\le
	|L| \ecost_e
	\le \ecost_E,
	\]
	which is a contradiction because $\vpen_L^\pot \ge \pot > \ecost_E$.
\end{proof}

%!TEX root = main.tex

\section{The threshold-tuple}
\label{sec:threshold-tuple}

We execute the modified Goemans-Williamson algorithm using potential zero and passing an empty tie-breaking list, i.e., we execute $\GW(0,\emptyset)$.
If the returned tree spans at least $k$ vertices, then this tree is a $2$-approximate solution, as stated in Lemma~\ref{appendix:bound_simple:lemma:2-approx}.
The proof is adapted from~Feofiloff~\etal~\cite{Feofiloff2010} and is given in Appendix~\ref{appendix:bound_simple}.
\begin{restatable}{mylemma}{appendixboundsimple}
	\label{appendix:bound_simple:lemma:2-approx}
	Let $\hatt$ be the tree returned by $\GW(0,\emptyset)$, and~$T^*$ be an optimal solution. Then,
	\begin{align*}
		\ecost_{E(\hatt)} + 2 \vpen_{V \setminus V(\hatt)} \leq 2 \left( \ecost_{E({T^*})} + \vpen_{V \setminus V(T^*)} \right).
	\end{align*}
\end{restatable}

In the remaining of this section, we assume that executing $\GW(0,\emptyset)$ returns a tree spanning less than $k$ vertices.
Observe that Lemma~\ref{subsec:GW:lemma:high_lambda_spans_V_G_r} implies that using potential greater than $\ecost_E$ leads to a tree with at least $k$ vertices for any tie-breaking list $\tblist$.
We would like to find $\pot$ and associated $\tblist$ such that the returned tree spans exactly $k$ vertices, but it might be the case that no such pair exists.
Instead, our goal will be finding a special tuple $(\pot,\tblist)$, called the \emph{threshold-tuple}, which will be defined below.

First, we need to introduce some notation.
Note that, in the $i$-th iteration of $\GP(\pot,\tblist)$, the edge or subset corresponding to $\tblist_i$ is not necessarily tight.
We say that a tie-breaking list $\tblist$ is \emph{respected} by potential $\pot$ if the sequence of edges and subsets of $V$ processed in the first $|\tblist|$ iterations of $\GP(\pot,\tblist)$  corresponds to $\tblist$.
Also, we denote by $\sublist$ the prefix of~$\tblist$ with size~${|\tblist| - 1}$.

\begin{mydefinition}
	Let $\hatt$ be the tree returned by $\GW(\pot,\tblist)$, and $\hatt'$ be the tree returned by $\GW(\pot,\sublist)$.
	We say that $(\pot,\tblist)$ is a \emph{threshold-tuple} if
	\begin{enumerate}[itemsep=0pt,topsep=-3pt, label=(\roman*),ref={\mbox{\rm $(\roman*)$}}]
		\item \label{subsec:threshold-tuple:theshold-tuple1} $\tblist$ is respected by~$\pot$; and
		\item \label{subsec:threshold-tuple:theshold-tuple2} $|V(\hatt)| \geq k > |V(\hatt')|$ or $|V(\hatt)| < k \leq |V(\hatt')|$.
	\end{enumerate}
\end{mydefinition}

Given a threshold-tuple, one may obtain a pair of trees, one with less than~$k$ vertices, and the other with at least~$k$ vertices.
These trees share many of their structures, and this will be used in Section~\ref{sec:find_solution} to construct a tree which spans exactly $k$ vertices and is a $2$-approximation.
In Subsection~\ref{subsec:threshold-tuple}, we study the properties of a threshold-tuple and, in Subsection~\ref{subsec:TS}, we show how it can be computed in polynomial time.

	%!TEX root = main.tex

\subsection{Properties of a threshold-tuple}
\label{subsec:threshold-tuple}

We start by noticing that the executions of $\GP(\pot,\tblist)$ and $\GP(\pot,\sublist)$ are identical up to the beginning of iteration $|\tblist|$.

\begin{mylemma}
	\label{subsec:threshold-tuple:lemma:if_nu_is_respected_then_nu'_is_respected}
	If $\tblist$ is respected by $\pot$, then $\sublist$ is respected by $\pot$.
	Also, at the beginning of the iteration~$|\tblist|$ of $\GP(\pot,\tblist)$ and $\GP(\pot,\sublist)$, the variables \mbox{$F$, $\call$, $\calb$ and $y$} are identical.
\end{mylemma}

\begin{proof}
	The first statement is clear by definition.
	For the second statement, observe that up to the beginning of the iteration~$|\tblist|$, the execution depends only on the $|\tblist|-1$ items of the tie-breaking list.
\end{proof}

The event processed in the iteration $|\tblist|$ of the growth-phase plays an important role in distinguishing the outputs returned by executing $\GW(\pot,\tblist)$ and $\GW(\pot,\sublist)$.
Each such an event corresponds to an edge or a subset.
The next auxiliary result lists the possibilities when $(\pot,\tblist)$ is a threshold-tuple.

\begin{mylemma}
	\label{subsec:threshold-tuple:lemma:threshold_subobject_is_an_edge}
	Assume that $(\pot,\tblist)$ is a threshold-tuple. Let $\obj$ be the edge or subset processed in the $|\tblist|$-th iteration of $\GP(\pot,\tblist)$, and $\obj'$ be the edge or subset processed in $|\tblist|$-th iteration of $\GP(\pot,\sublist)$.
	Then $\obj \neq \obj'$, and $\obj'$ is an edge.
\end{mylemma}
\begin{proof}
	If we had $\obj = \obj'$, then the output of $\GP(\pot,\sublist)$ would be identical to the output of $\GP(\pot,\tblist)$, and $\GW(\pot,\sublist)$ and $\GW(\pot,\tblist)$ would return identical trees. This is not possible, as $(\pot,\tblist)$ is a threshold-tuple, thus indeed $\obj \neq \obj'$.

	First, assume that $\obj$ is an edge.
	Since $\tblist$ is respected by $\pot$, $\obj = \tblist_{|\tblist|}$. Thus, at the beginning of the $|\tblist|$-th iteration of $\GP(\pot,\tblist)$, $\obj$ is a tight edge external in $\call^*$.
	By Lemma~\ref{subsec:threshold-tuple:lemma:if_nu_is_respected_then_nu'_is_respected}, at the beginning of the $|\tblist|$-th iteration of $\GP(\pot,\sublist)$, $\obj$ is also a tight edge external in $\call^*$.
	Observe that in this iteration of $\GP(\pot,\sublist)$, edges have the highest priorities, since the size of the considered tie-breaking list is $|\sublist| < |\tblist|$.
	Because there is at least one tight edge to be processed, this iteration processes an edge, thus $\obj'$ is an edge.

	Now, assume that $\obj$ is a subset. Suppose, for a contradiction, that $\obj'$ is a subset.
	Since a subset is processed in the $|\tblist|$-iteration of $\GP(\pot,\sublist)$, there are no tight edges external in $\call^*$ in the beginning of this iteration.
	Let $L_0,\dots,L_m$ be the collection of tight active subsets in the beginning of the $|\tblist|$-th iteration of $\GP(\pot,\sublist)$, which is the same for $\GP(\pot,\tblist)$ by  Lemma~\ref{subsec:threshold-tuple:lemma:if_nu_is_respected_then_nu'_is_respected}.
	Notice that processing a subset does not modify variables \mbox{$F$, $\call$ and $y$}, thus each subset~$L_i$ must be processed in both executions (although in different order) before any other edge or subset becomes tight.
	It follows that in the beginning of the $(|\tblist|+m)$-th iteration of both $\GP(\pot,\tblist)$ and $\GP(\pot,\sublist)$, variables \mbox{$F$, $\call$, $\calb$ and $y$} are identical.
	This implies that both executions have the same output. But, again, this is a contradiction because $(\pot,\tblist)$ is a threshold-tuple.
\end{proof}

In the the following, we use the notion of rounds of iterations.
Recall that during the execution of $\GP({\pot},\tblist)$, there might be iterations for which the increment $\Delta$ to variables of $y$ is set to zero, and thus vector $y$ remains unaltered.
We say that iterations $i$~and~$j$ are in the same \emph{round} if the value of vector~$y$ at the end of iteration $i$ equals the value of vector~$y$ at the end of iteration $j$.

The output of the growth-phase corresponds to a tree and a collection of processed subsets. Suppose that executing $\GP(\pot,\tblist)$ returns pair $(T, \calb)$, and executing $\GP(\pot,\sublist)$ returns pair $(T', \calb')$. While, for a threshold-tuple, these pairs must be different, we show that the difference is restricted to adding and removing an edge, or adding a subset.
The proof's arguments rely on the fact that both executions of the growth-phase are almost identical, and differ only in the round of iteration~$|\tblist|$.

Whether the trees or the collections will be different depends on the event processed at iteration~$|\tblist|$.
If $\GP(\pot,\tblist)$ processes an edge~$\obj$, then $\GP(\pot,\sublist)$ processes an edge~$\obj'$, and $T = T' + \obj - \obj'$.
We give an illustration in Figure~\ref{subsec:picking:fig:edge_case}.
The state at the beginning of iteration $|\tblist|$ is depicted in Subfigure~\ref{subsec:picking:fig:edge_case_round_state}, which is the same for both executions.
Each maximal subset is represented by a circle, and processed subsets correspond to filled circles.
Only edges which are tight and external are drawn, and only maximal subsets which are tight are labelled.
The priority order of each edge and subset is given by the corresponding label index.
Subfigure~\ref{subsec:picking:fig:edge_case_round_end_1} represents the end of the round for $\GP(\pot,\sublist)$, where the contours surrounding circles are the new formed subsets. In this execution, $\obj' = e_1$ is the first edge, and the processing order is $e_1, e_2, e_3, e_4, e_5, L_1$.
Subfigure~\ref{subsec:picking:fig:edge_case_round_end_2} represents the end of the round for $\GP(\pot,\tblist)$, for which $\obj = e_6$ is the last item of $\tblist$, and the processing order is $e_6,e_1,e_2,e_3,e_4, L_1$.

\begin{figure}
	\begin{figure}[H]
		\centering
		\begin{subfigure}[t]{0.5\textwidth}
			\center
			\scalebox{1}{%!TEX root = ../../main.tex

\tikzset{every picture/.style={line width=0.75pt}} %set default line width to 0.75pt        

\begin{tikzpicture}[x=0.75pt,y=0.75pt,yscale=-1,xscale=1]
%uncomment if require: \path (0,310); %set diagram left start at 0, and has height of 310

%Straight Lines [id:da5408933639643501] 
\draw    (155.21,115) -- (155.21,165) ;

%Straight Lines [id:da40541231793554] 
\draw    (155.21,115) -- (205.21,115) ;

%Straight Lines [id:da18493015647163413] 
\draw    (105.21,115) -- (155.21,115) ;

%Straight Lines [id:da282781521593489] 
\draw    (205.21,115) -- (205.21,165) ;

%Straight Lines [id:da8521716845726137] 
\draw    (155.21,165) -- (205.21,165) ;

%Straight Lines [id:da5229067382016857] 
\draw    (275.21,105) -- (275.21,175) ;

%Flowchart: Connector [id:dp18987164902597464] 
\draw  [fill={rgb, 255:red, 255; green, 255; blue, 255 }  ,fill opacity=1 ] (150,165) .. controls (150,162.24) and (152.33,160) .. (155.21,160) .. controls (158.08,160) and (160.42,162.24) .. (160.42,165) .. controls (160.42,167.76) and (158.08,170) .. (155.21,170) .. controls (152.33,170) and (150,167.76) .. (150,165) -- cycle ;
%Flowchart: Connector [id:dp5648490767218977] 
\draw  [fill={rgb, 255:red, 255; green, 255; blue, 255 }  ,fill opacity=1 ] (150,115) .. controls (150,112.24) and (152.33,110) .. (155.21,110) .. controls (158.08,110) and (160.42,112.24) .. (160.42,115) .. controls (160.42,117.76) and (158.08,120) .. (155.21,120) .. controls (152.33,120) and (150,117.76) .. (150,115) -- cycle ;
%Flowchart: Connector [id:dp6906887514856612] 
\draw  [fill={rgb, 255:red, 255; green, 255; blue, 255 }  ,fill opacity=1 ] (200,165) .. controls (200,162.24) and (202.33,160) .. (205.21,160) .. controls (208.08,160) and (210.42,162.24) .. (210.42,165) .. controls (210.42,167.76) and (208.08,170) .. (205.21,170) .. controls (202.33,170) and (200,167.76) .. (200,165) -- cycle ;
%Flowchart: Connector [id:dp9254350467657315] 
\draw  [fill={rgb, 255:red, 0; green, 0; blue, 0 }  ,fill opacity=1 ] (200,115) .. controls (200,112.24) and (202.33,110) .. (205.21,110) .. controls (208.08,110) and (210.42,112.24) .. (210.42,115) .. controls (210.42,117.76) and (208.08,120) .. (205.21,120) .. controls (202.33,120) and (200,117.76) .. (200,115) -- cycle ;
%Flowchart: Connector [id:dp6537259142087386] 
\draw  [fill={rgb, 255:red, 255; green, 255; blue, 255 }  ,fill opacity=1 ] (100,115) .. controls (100,112.24) and (102.33,110) .. (105.21,110) .. controls (108.08,110) and (110.42,112.24) .. (110.42,115) .. controls (110.42,117.76) and (108.08,120) .. (105.21,120) .. controls (102.33,120) and (100,117.76) .. (100,115) -- cycle ;
%Flowchart: Connector [id:dp9487700281635219] 
\draw  [fill={rgb, 255:red, 0; green, 0; blue, 0 }  ,fill opacity=1 ] (270,105) .. controls (270,102.24) and (272.33,100) .. (275.21,100) .. controls (278.08,100) and (280.42,102.24) .. (280.42,105) .. controls (280.42,107.76) and (278.08,110) .. (275.21,110) .. controls (272.33,110) and (270,107.76) .. (270,105) -- cycle ;
%Flowchart: Connector [id:dp4315654553902033] 
\draw  [fill={rgb, 255:red, 255; green, 255; blue, 255 }  ,fill opacity=1 ] (270,175) .. controls (270,172.24) and (272.33,170) .. (275.21,170) .. controls (278.08,170) and (280.42,172.24) .. (280.42,175) .. controls (280.42,177.76) and (278.08,180) .. (275.21,180) .. controls (272.33,180) and (270,177.76) .. (270,175) -- cycle ;
%Flowchart: Connector [id:dp29309822307272904] 
\draw   (80,175) .. controls (80,172.24) and (82.24,170) .. (85,170) .. controls (87.76,170) and (90,172.24) .. (90,175) .. controls (90,177.76) and (87.76,180) .. (85,180) .. controls (82.24,180) and (80,177.76) .. (80,175) -- cycle ;

% Text Node
\draw (148,139.5) node [scale=0.9] [align=left] {$\displaystyle e_{1}$};
% Text Node
\draw (182,100.5) node [scale=0.9] [align=left] {$\displaystyle e_{2}$};
% Text Node
\draw (288,139.5) node [scale=0.9] [align=left] {$\displaystyle e_{3}$};
% Text Node
\draw (132,100.5) node [scale=0.9] [align=left] {$\displaystyle e_{4}$};
% Text Node
\draw (218,139.5) node [scale=0.9] [align=left] {$\displaystyle e_{5}$};
% Text Node
\draw (182,175.5) node [scale=0.9] [align=left] {$\displaystyle e_{6}$};
% Text Node
\draw (88,159.5) node [scale=0.9] [align=left] {$\displaystyle L_{1}$};

\end{tikzpicture}}
			\caption{State in the $|\tblist|$-th iteration.}
			\label{subsec:picking:fig:edge_case_round_state}
		\end{subfigure}\\
		\begin{subfigure}[t]{0.5\textwidth}
			\center
			\scalebox{1}{%!TEX root = ../../main.tex

\tikzset{every picture/.style={line width=0.75pt}} %set default line width to 0.75pt        

\begin{tikzpicture}[x=0.75pt,y=0.75pt,yscale=-1,xscale=1]
%uncomment if require: \path (0,310); %set diagram left start at 0, and has height of 310

%Straight Lines [id:da40541231793554] 
\draw    (167.21,115) -- (205.21,115) ;

%Flowchart: Connector [id:dp9254350467657315] 
\draw  [fill={rgb, 255:red, 0; green, 0; blue, 0 }  ,fill opacity=1 ] (200,115) .. controls (200,112.24) and (202.33,110) .. (205.21,110) .. controls (208.08,110) and (210.42,112.24) .. (210.42,115) .. controls (210.42,117.76) and (208.08,120) .. (205.21,120) .. controls (202.33,120) and (200,117.76) .. (200,115) -- cycle ;
%Flowchart: Connector [id:dp6285788834487531] 
\draw   (140,140) .. controls (140,117.91) and (146.72,100) .. (155,100) .. controls (163.28,100) and (170,117.91) .. (170,140) .. controls (170,162.09) and (163.28,180) .. (155,180) .. controls (146.72,180) and (140,162.09) .. (140,140) -- cycle ;
%Straight Lines [id:da5408933639643501] 
\draw    (155.21,115) -- (155.21,165) ;

%Flowchart: Connector [id:dp18987164902597464] 
\draw  [fill={rgb, 255:red, 255; green, 255; blue, 255 }  ,fill opacity=1 ] (150,165) .. controls (150,162.24) and (152.33,160) .. (155.21,160) .. controls (158.08,160) and (160.42,162.24) .. (160.42,165) .. controls (160.42,167.76) and (158.08,170) .. (155.21,170) .. controls (152.33,170) and (150,167.76) .. (150,165) -- cycle ;
%Flowchart: Connector [id:dp5648490767218977] 
\draw  [fill={rgb, 255:red, 255; green, 255; blue, 255 }  ,fill opacity=1 ] (150,115) .. controls (150,112.24) and (152.33,110) .. (155.21,110) .. controls (158.08,110) and (160.42,112.24) .. (160.42,115) .. controls (160.42,117.76) and (158.08,120) .. (155.21,120) .. controls (152.33,120) and (150,117.76) .. (150,115) -- cycle ;
%Curve Lines [id:da9233670743447082] 
\draw    (80,120) .. controls (89.5,57) and (233.5,60) .. (240,110) .. controls (246.5,160) and (194.5,204) .. (150,197) .. controls (105.5,190) and (73.5,151) .. (80,120) -- cycle ;

%Straight Lines [id:da18493015647163413] 
\draw    (105.21,115) -- (131.21,115) ;

%Straight Lines [id:da282781521593489] 
\draw    (205.21,132) -- (205.21,165) ;

%Straight Lines [id:da5229067382016857] 
\draw    (275.21,105) -- (275.21,175) ;

%Flowchart: Connector [id:dp6906887514856612] 
\draw  [fill={rgb, 255:red, 255; green, 255; blue, 255 }  ,fill opacity=1 ] (200,165) .. controls (200,162.24) and (202.33,160) .. (205.21,160) .. controls (208.08,160) and (210.42,162.24) .. (210.42,165) .. controls (210.42,167.76) and (208.08,170) .. (205.21,170) .. controls (202.33,170) and (200,167.76) .. (200,165) -- cycle ;
%Flowchart: Connector [id:dp6537259142087386] 
\draw  [fill={rgb, 255:red, 255; green, 255; blue, 255 }  ,fill opacity=1 ] (100,115) .. controls (100,112.24) and (102.33,110) .. (105.21,110) .. controls (108.08,110) and (110.42,112.24) .. (110.42,115) .. controls (110.42,117.76) and (108.08,120) .. (105.21,120) .. controls (102.33,120) and (100,117.76) .. (100,115) -- cycle ;
%Flowchart: Connector [id:dp9487700281635219] 
\draw  [fill={rgb, 255:red, 0; green, 0; blue, 0 }  ,fill opacity=1 ] (270,105) .. controls (270,102.24) and (272.33,100) .. (275.21,100) .. controls (278.08,100) and (280.42,102.24) .. (280.42,105) .. controls (280.42,107.76) and (278.08,110) .. (275.21,110) .. controls (272.33,110) and (270,107.76) .. (270,105) -- cycle ;
%Flowchart: Connector [id:dp4315654553902033] 
\draw  [fill={rgb, 255:red, 255; green, 255; blue, 255 }  ,fill opacity=1 ] (270,175) .. controls (270,172.24) and (272.33,170) .. (275.21,170) .. controls (278.08,170) and (280.42,172.24) .. (280.42,175) .. controls (280.42,177.76) and (278.08,180) .. (275.21,180) .. controls (272.33,180) and (270,177.76) .. (270,175) -- cycle ;
%Flowchart: Connector [id:dp282034471814228] 
\draw   (260,140) .. controls (260,112.39) and (266.72,90) .. (275,90) .. controls (283.28,90) and (290,112.39) .. (290,140) .. controls (290,167.61) and (283.28,190) .. (275,190) .. controls (266.72,190) and (260,167.61) .. (260,140) -- cycle ;
%Curve Lines [id:da03567926667121912] 
\draw    (155,185) .. controls (118.5,186.33) and (121.5,93.33) .. (150,90) .. controls (178.5,86.67) and (220.5,94.67) .. (220,115) .. controls (219.5,135.33) and (188.52,115.18) .. (180,130) .. controls (171.48,144.82) and (186.5,181.33) .. (155,185) -- cycle ;

%Curve Lines [id:da05039348013283873] 
\draw    (90,110) .. controls (101.5,91) and (134.29,81.36) .. (165.21,80) .. controls (196.13,78.64) and (240.92,99) .. (235.21,120) .. controls (229.5,141) and (204.5,123) .. (190,140) .. controls (175.5,157) and (199.92,190) .. (155.21,192) .. controls (110.5,194) and (77.29,128) .. (90,110) -- cycle ;

%Flowchart: Connector [id:dp9158886609502273] 
\draw  [fill={rgb, 255:red, 0; green, 0; blue, 0 }  ,fill opacity=1 ] (80,175) .. controls (80,172.24) and (82.33,170) .. (85.21,170) .. controls (88.08,170) and (90.42,172.24) .. (90.42,175) .. controls (90.42,177.76) and (88.08,180) .. (85.21,180) .. controls (82.33,180) and (80,177.76) .. (80,175) -- cycle ;

\end{tikzpicture}}
			\caption{State at end of round for $\obj' = e_1$.}
			\label{subsec:picking:fig:edge_case_round_end_1}
		\end{subfigure}%
		~
		\begin{subfigure}[t]{0.5\textwidth}
			\center
			\scalebox{1}{%!TEX root = ../../main.tex

\tikzset{every picture/.style={line width=0.75pt}} %set default line width to 0.75pt        

\begin{tikzpicture}[x=0.75pt,y=0.75pt,yscale=-1,xscale=1]
%uncomment if require: \path (0,310); %set diagram left start at 0, and has height of 310

%Straight Lines [id:da707388211491798] 
\draw [fill={rgb, 255:red, 255; green, 255; blue, 255 }  ,fill opacity=1 ]   (105.21,115) -- (127.21,115) ;

%Straight Lines [id:da30936104659752406] 
\draw    (155.21,115) -- (155.21,154) ;

%Flowchart: Connector [id:dp9254350467657315] 
\draw  [fill={rgb, 255:red, 0; green, 0; blue, 0 }  ,fill opacity=1 ] (200,115) .. controls (200,112.24) and (202.33,110) .. (205.21,110) .. controls (208.08,110) and (210.42,112.24) .. (210.42,115) .. controls (210.42,117.76) and (208.08,120) .. (205.21,120) .. controls (202.33,120) and (200,117.76) .. (200,115) -- cycle ;
%Straight Lines [id:da5408933639643501] 
\draw    (205.21,165) -- (155.21,165) ;

%Flowchart: Connector [id:dp18987164902597464] 
\draw  [fill={rgb, 255:red, 255; green, 255; blue, 255 }  ,fill opacity=1 ] (150,165) .. controls (150,162.24) and (152.33,160) .. (155.21,160) .. controls (158.08,160) and (160.42,162.24) .. (160.42,165) .. controls (160.42,167.76) and (158.08,170) .. (155.21,170) .. controls (152.33,170) and (150,167.76) .. (150,165) -- cycle ;
%Flowchart: Connector [id:dp5648490767218977] 
\draw  [fill={rgb, 255:red, 255; green, 255; blue, 255 }  ,fill opacity=1 ] (150,115) .. controls (150,112.24) and (152.33,110) .. (155.21,110) .. controls (158.08,110) and (160.42,112.24) .. (160.42,115) .. controls (160.42,117.76) and (158.08,120) .. (155.21,120) .. controls (152.33,120) and (150,117.76) .. (150,115) -- cycle ;
%Straight Lines [id:da5229067382016857] 
\draw    (275.21,105) -- (275.21,175) ;

%Flowchart: Connector [id:dp6906887514856612] 
\draw  [fill={rgb, 255:red, 255; green, 255; blue, 255 }  ,fill opacity=1 ] (200,165) .. controls (200,162.24) and (202.33,160) .. (205.21,160) .. controls (208.08,160) and (210.42,162.24) .. (210.42,165) .. controls (210.42,167.76) and (208.08,170) .. (205.21,170) .. controls (202.33,170) and (200,167.76) .. (200,165) -- cycle ;
%Flowchart: Connector [id:dp6537259142087386] 
\draw  [fill={rgb, 255:red, 255; green, 255; blue, 255 }  ,fill opacity=1 ] (100,115) .. controls (100,112.24) and (102.33,110) .. (105.21,110) .. controls (108.08,110) and (110.42,112.24) .. (110.42,115) .. controls (110.42,117.76) and (108.08,120) .. (105.21,120) .. controls (102.33,120) and (100,117.76) .. (100,115) -- cycle ;
%Flowchart: Connector [id:dp9487700281635219] 
\draw  [fill={rgb, 255:red, 0; green, 0; blue, 0 }  ,fill opacity=1 ] (270,105) .. controls (270,102.24) and (272.33,100) .. (275.21,100) .. controls (278.08,100) and (280.42,102.24) .. (280.42,105) .. controls (280.42,107.76) and (278.08,110) .. (275.21,110) .. controls (272.33,110) and (270,107.76) .. (270,105) -- cycle ;
%Flowchart: Connector [id:dp4315654553902033] 
\draw  [fill={rgb, 255:red, 255; green, 255; blue, 255 }  ,fill opacity=1 ] (270,175) .. controls (270,172.24) and (272.33,170) .. (275.21,170) .. controls (278.08,170) and (280.42,172.24) .. (280.42,175) .. controls (280.42,177.76) and (278.08,180) .. (275.21,180) .. controls (272.33,180) and (270,177.76) .. (270,175) -- cycle ;
%Flowchart: Connector [id:dp282034471814228] 
\draw   (260,140) .. controls (260,112.39) and (266.72,90) .. (275,90) .. controls (283.28,90) and (290,112.39) .. (290,140) .. controls (290,167.61) and (283.28,190) .. (275,190) .. controls (266.72,190) and (260,167.61) .. (260,140) -- cycle ;
%Flowchart: Connector [id:dp19109153033674287] 
\draw   (140,165) .. controls (140,156.72) and (157.91,150) .. (180,150) .. controls (202.09,150) and (220,156.72) .. (220,165) .. controls (220,173.28) and (202.09,180) .. (180,180) .. controls (157.91,180) and (140,173.28) .. (140,165) -- cycle ;
%Curve Lines [id:da8995668081008774] 
\draw    (155,105) .. controls (170.5,104.33) and (167.5,125.33) .. (172,136) .. controls (176.5,146.67) and (231.5,144.33) .. (228,167) .. controls (224.5,189.67) and (132.5,195.67) .. (136,161) .. controls (139.5,126.33) and (140.5,106.33) .. (155,105) -- cycle ;

%Straight Lines [id:da06526396160280501] 
\draw    (205.21,115) -- (168.21,115) ;

%Curve Lines [id:da6836055876361216] 
\draw    (90,120) .. controls (103.3,78.2) and (208.3,74.2) .. (230,110) .. controls (251.7,145.8) and (250.5,176.67) .. (230,190) .. controls (209.5,203.33) and (141.3,202.2) .. (124,190) .. controls (106.7,177.8) and (82.3,146.2) .. (90,120) -- cycle ;

%Flowchart: Connector [id:dp15264728300247166] 
\draw  [fill={rgb, 255:red, 0; green, 0; blue, 0 }  ,fill opacity=1 ] (80,175) .. controls (80,172.24) and (82.24,170) .. (85,170) .. controls (87.76,170) and (90,172.24) .. (90,175) .. controls (90,177.76) and (87.76,180) .. (85,180) .. controls (82.24,180) and (80,177.76) .. (80,175) -- cycle ;
%Curve Lines [id:da03532012918440941] 
\draw    (130,110) .. controls (146.3,90.2) and (202.3,94.2) .. (220,110) .. controls (237.7,125.8) and (241.3,158.2) .. (230,180) .. controls (218.7,201.8) and (141.3,194.2) .. (130,180) .. controls (118.7,165.8) and (120.3,120.2) .. (130,110) -- cycle ;

\end{tikzpicture}}
			\caption{State at end of round for $\obj = e_6$.}
			\label{subsec:picking:fig:edge_case_round_end_2}
		\end{subfigure}
		\caption{Executions of $\GP(\pot,\sublist)$ and $\GP(\pot,\tblist)$ in the edge case.}
		\label{subsec:picking:fig:edge_case}
	\end{figure}

	\medskip

	\begin{figure}[H]
		\centering
		\begin{subfigure}[t]{0.5\textwidth}
			\center
			\scalebox{1}{%!TEX root = ../../main.tex

\tikzset{every picture/.style={line width=0.75pt}} %set default line width to 0.75pt        

\begin{tikzpicture}[x=0.75pt,y=0.75pt,yscale=-1,xscale=1]
%uncomment if require: \path (0,310); %set diagram left start at 0, and has height of 310

%Straight Lines [id:da014650406411072847] 
\draw    (155.21,175) -- (205.42,175) ;

%Straight Lines [id:da8385320461622527] 
\draw    (105.21,175) -- (154.79,175) ;

%Straight Lines [id:da40541231793554] 
\draw    (155.21,115) -- (205.21,115) ;

%Straight Lines [id:da18493015647163413] 
\draw    (105.21,115) -- (155.21,115) ;

%Flowchart: Connector [id:dp5648490767218977] 
\draw  [fill={rgb, 255:red, 255; green, 255; blue, 255 }  ,fill opacity=1 ] (150,115) .. controls (150,112.24) and (152.33,110) .. (155.21,110) .. controls (158.08,110) and (160.42,112.24) .. (160.42,115) .. controls (160.42,117.76) and (158.08,120) .. (155.21,120) .. controls (152.33,120) and (150,117.76) .. (150,115) -- cycle ;
%Flowchart: Connector [id:dp9254350467657315] 
\draw  [fill={rgb, 255:red, 0; green, 0; blue, 0 }  ,fill opacity=1 ] (200,115) .. controls (200,112.24) and (202.33,110) .. (205.21,110) .. controls (208.08,110) and (210.42,112.24) .. (210.42,115) .. controls (210.42,117.76) and (208.08,120) .. (205.21,120) .. controls (202.33,120) and (200,117.76) .. (200,115) -- cycle ;
%Flowchart: Connector [id:dp6537259142087386] 
\draw  [fill={rgb, 255:red, 255; green, 255; blue, 255 }  ,fill opacity=1 ] (100,115) .. controls (100,112.24) and (102.33,110) .. (105.21,110) .. controls (108.08,110) and (110.42,112.24) .. (110.42,115) .. controls (110.42,117.76) and (108.08,120) .. (105.21,120) .. controls (102.33,120) and (100,117.76) .. (100,115) -- cycle ;
%Flowchart: Connector [id:dp9487700281635219] 
\draw  [fill={rgb, 255:red, 0; green, 0; blue, 0 }  ,fill opacity=1 ] (150,175) .. controls (150,172.24) and (152.33,170) .. (155.21,170) .. controls (158.08,170) and (160.42,172.24) .. (160.42,175) .. controls (160.42,177.76) and (158.08,180) .. (155.21,180) .. controls (152.33,180) and (150,177.76) .. (150,175) -- cycle ;
%Flowchart: Connector [id:dp4315654553902033] 
\draw  [fill={rgb, 255:red, 255; green, 255; blue, 255 }  ,fill opacity=1 ] (100,175) .. controls (100,172.24) and (102.33,170) .. (105.21,170) .. controls (108.08,170) and (110.42,172.24) .. (110.42,175) .. controls (110.42,177.76) and (108.08,180) .. (105.21,180) .. controls (102.33,180) and (100,177.76) .. (100,175) -- cycle ;
%Flowchart: Connector [id:dp29309822307272904] 
\draw  [fill={rgb, 255:red, 255; green, 255; blue, 255 }  ,fill opacity=1 ] (200,175) .. controls (200,172.24) and (202.24,170) .. (205,170) .. controls (207.76,170) and (210,172.24) .. (210,175) .. controls (210,177.76) and (207.76,180) .. (205,180) .. controls (202.24,180) and (200,177.76) .. (200,175) -- cycle ;
%Flowchart: Connector [id:dp974419747384212] 
\draw  [fill={rgb, 255:red, 255; green, 255; blue, 255 }  ,fill opacity=1 ] (249.58,115) .. controls (249.58,112.24) and (251.92,110) .. (254.79,110) .. controls (257.67,110) and (260,112.24) .. (260,115) .. controls (260,117.76) and (257.67,120) .. (254.79,120) .. controls (251.92,120) and (249.58,117.76) .. (249.58,115) -- cycle ;
%Flowchart: Connector [id:dp9794100447164118] 
\draw  [fill={rgb, 255:red, 255; green, 255; blue, 255 }  ,fill opacity=1 ] (250,175) .. controls (250,172.24) and (252.33,170) .. (255.21,170) .. controls (258.08,170) and (260.42,172.24) .. (260.42,175) .. controls (260.42,177.76) and (258.08,180) .. (255.21,180) .. controls (252.33,180) and (250,177.76) .. (250,175) -- cycle ;

% Text Node
\draw (182,100.5) node [scale=0.9] [align=left] {$\displaystyle e_{1}$};
% Text Node
\draw (132,100.5) node [scale=0.9] [align=left] {$\displaystyle e_{4}$};
% Text Node
\draw (105.5,190.5) node [scale=0.9] [align=left] {$\displaystyle L_{1}$};
% Text Node
\draw (182,160.5) node [scale=0.9] [align=left] {$\displaystyle e_{3}$};
% Text Node
\draw (132,160.5) node [scale=0.9] [align=left] {$\displaystyle e_{2}$};
% Text Node
\draw (105.5,130.5) node [scale=0.9] [align=left] {$\displaystyle L_{3}$};
% Text Node
\draw (250.5,128.5) node [scale=0.9] [align=left] {$\displaystyle L_{2}$};

\end{tikzpicture}}
			\caption{State in the $|\tblist|$-th iteration.}
			\label{subsec:picking:fig:subset_case_round_state}
		\end{subfigure}\\
		\vspace*{10pt}
		\begin{subfigure}[t]{0.5\textwidth}
			\center
			\scalebox{1}{%!TEX root = ../../main.tex

\tikzset{every picture/.style={line width=0.75pt}} %set default line width to 0.75pt        

\begin{tikzpicture}[x=0.75pt,y=0.75pt,yscale=-1,xscale=1]
%uncomment if require: \path (0,310); %set diagram left start at 0, and has height of 310

%Straight Lines [id:da014650406411072847] 
\draw    (165,175) -- (205.42,175) ;

%Straight Lines [id:da8385320461622527] 
\draw    (105.21,175) -- (154.79,175) ;

%Straight Lines [id:da40541231793554] 
\draw    (155.21,115) -- (205.21,115) ;

%Straight Lines [id:da18493015647163413] 
\draw    (105.21,115) -- (145.21,115) ;

%Flowchart: Connector [id:dp5648490767218977] 
\draw  [fill={rgb, 255:red, 255; green, 255; blue, 255 }  ,fill opacity=1 ] (150,115) .. controls (150,112.24) and (152.33,110) .. (155.21,110) .. controls (158.08,110) and (160.42,112.24) .. (160.42,115) .. controls (160.42,117.76) and (158.08,120) .. (155.21,120) .. controls (152.33,120) and (150,117.76) .. (150,115) -- cycle ;
%Flowchart: Connector [id:dp9254350467657315] 
\draw  [fill={rgb, 255:red, 0; green, 0; blue, 0 }  ,fill opacity=1 ] (200,115) .. controls (200,112.24) and (202.33,110) .. (205.21,110) .. controls (208.08,110) and (210.42,112.24) .. (210.42,115) .. controls (210.42,117.76) and (208.08,120) .. (205.21,120) .. controls (202.33,120) and (200,117.76) .. (200,115) -- cycle ;
%Flowchart: Connector [id:dp6537259142087386] 
\draw  [fill={rgb, 255:red, 255; green, 255; blue, 255 }  ,fill opacity=1 ] (100,115) .. controls (100,112.24) and (102.33,110) .. (105.21,110) .. controls (108.08,110) and (110.42,112.24) .. (110.42,115) .. controls (110.42,117.76) and (108.08,120) .. (105.21,120) .. controls (102.33,120) and (100,117.76) .. (100,115) -- cycle ;
%Flowchart: Connector [id:dp9487700281635219] 
\draw  [fill={rgb, 255:red, 0; green, 0; blue, 0 }  ,fill opacity=1 ] (150,175) .. controls (150,172.24) and (152.33,170) .. (155.21,170) .. controls (158.08,170) and (160.42,172.24) .. (160.42,175) .. controls (160.42,177.76) and (158.08,180) .. (155.21,180) .. controls (152.33,180) and (150,177.76) .. (150,175) -- cycle ;
%Flowchart: Connector [id:dp4315654553902033] 
\draw  [fill={rgb, 255:red, 255; green, 255; blue, 255 }  ,fill opacity=1 ] (100,175) .. controls (100,172.24) and (102.33,170) .. (105.21,170) .. controls (108.08,170) and (110.42,172.24) .. (110.42,175) .. controls (110.42,177.76) and (108.08,180) .. (105.21,180) .. controls (102.33,180) and (100,177.76) .. (100,175) -- cycle ;
%Flowchart: Connector [id:dp29309822307272904] 
\draw  [fill={rgb, 255:red, 255; green, 255; blue, 255 }  ,fill opacity=1 ] (200,175) .. controls (200,172.24) and (202.24,170) .. (205,170) .. controls (207.76,170) and (210,172.24) .. (210,175) .. controls (210,177.76) and (207.76,180) .. (205,180) .. controls (202.24,180) and (200,177.76) .. (200,175) -- cycle ;
%Flowchart: Connector [id:dp974419747384212] 
\draw  [fill={rgb, 255:red, 0; green, 0; blue, 0 }  ,fill opacity=1 ] (249.58,115) .. controls (249.58,112.24) and (251.92,110) .. (254.79,110) .. controls (257.67,110) and (260,112.24) .. (260,115) .. controls (260,117.76) and (257.67,120) .. (254.79,120) .. controls (251.92,120) and (249.58,117.76) .. (249.58,115) -- cycle ;
%Flowchart: Connector [id:dp9794100447164118] 
\draw  [fill={rgb, 255:red, 255; green, 255; blue, 255 }  ,fill opacity=1 ] (250,175) .. controls (250,172.24) and (252.33,170) .. (255.21,170) .. controls (258.08,170) and (260.42,172.24) .. (260.42,175) .. controls (260.42,177.76) and (258.08,180) .. (255.21,180) .. controls (252.33,180) and (250,177.76) .. (250,175) -- cycle ;
%Flowchart: Connector [id:dp137575761480913] 
\draw   (145.21,115) .. controls (145.21,106.72) and (160.88,100) .. (180.21,100) .. controls (199.54,100) and (215.21,106.72) .. (215.21,115) .. controls (215.21,123.28) and (199.54,130) .. (180.21,130) .. controls (160.88,130) and (145.21,123.28) .. (145.21,115) -- cycle ;
%Flowchart: Connector [id:dp7500852188323093] 
\draw   (98,115) .. controls (98,101.19) and (124.86,90) .. (158,90) .. controls (191.14,90) and (218,101.19) .. (218,115) .. controls (218,128.81) and (191.14,140) .. (158,140) .. controls (124.86,140) and (98,128.81) .. (98,115) -- cycle ;
%Flowchart: Connector [id:dp3743361984621618] 
\draw   (95,175) .. controls (95,166.72) and (110.67,160) .. (130,160) .. controls (149.33,160) and (165,166.72) .. (165,175) .. controls (165,183.28) and (149.33,190) .. (130,190) .. controls (110.67,190) and (95,183.28) .. (95,175) -- cycle ;
%Flowchart: Connector [id:dp9338346101367194] 
\draw   (92,175) .. controls (92,161.19) and (118.86,150) .. (152,150) .. controls (185.14,150) and (212,161.19) .. (212,175) .. controls (212,188.81) and (185.14,200) .. (152,200) .. controls (118.86,200) and (92,188.81) .. (92,175) -- cycle ;

\end{tikzpicture}}
			\caption{State at end of round for $\obj' = e_1$.}
			\label{subsec:picking:fig:subset_case_round_end_1}
		\end{subfigure}%
		~
		\begin{subfigure}[t]{0.5\textwidth}
			\center
			\scalebox{1}{%!TEX root = ../../main.tex

\tikzset{every picture/.style={line width=0.75pt}} %set default line width to 0.75pt        

\begin{tikzpicture}[x=0.75pt,y=0.75pt,yscale=-1,xscale=1]
%uncomment if require: \path (0,310); %set diagram left start at 0, and has height of 310

%Straight Lines [id:da014650406411072847] 
\draw    (165,175) -- (205.42,175) ;

%Straight Lines [id:da8385320461622527] 
\draw    (105.21,175) -- (154.79,175) ;

%Straight Lines [id:da40541231793554] 
\draw    (155.21,115) -- (205.21,115) ;

%Straight Lines [id:da18493015647163413] 
\draw    (105.21,115) -- (145.21,115) ;

%Flowchart: Connector [id:dp5648490767218977] 
\draw  [fill={rgb, 255:red, 255; green, 255; blue, 255 }  ,fill opacity=1 ] (150,115) .. controls (150,112.24) and (152.33,110) .. (155.21,110) .. controls (158.08,110) and (160.42,112.24) .. (160.42,115) .. controls (160.42,117.76) and (158.08,120) .. (155.21,120) .. controls (152.33,120) and (150,117.76) .. (150,115) -- cycle ;
%Flowchart: Connector [id:dp9254350467657315] 
\draw  [fill={rgb, 255:red, 0; green, 0; blue, 0 }  ,fill opacity=1 ] (200,115) .. controls (200,112.24) and (202.33,110) .. (205.21,110) .. controls (208.08,110) and (210.42,112.24) .. (210.42,115) .. controls (210.42,117.76) and (208.08,120) .. (205.21,120) .. controls (202.33,120) and (200,117.76) .. (200,115) -- cycle ;
%Flowchart: Connector [id:dp6537259142087386] 
\draw  [fill={rgb, 255:red, 0; green, 0; blue, 0 }  ,fill opacity=1 ] (100,115) .. controls (100,112.24) and (102.33,110) .. (105.21,110) .. controls (108.08,110) and (110.42,112.24) .. (110.42,115) .. controls (110.42,117.76) and (108.08,120) .. (105.21,120) .. controls (102.33,120) and (100,117.76) .. (100,115) -- cycle ;
%Flowchart: Connector [id:dp9487700281635219] 
\draw  [fill={rgb, 255:red, 0; green, 0; blue, 0 }  ,fill opacity=1 ] (150,175) .. controls (150,172.24) and (152.33,170) .. (155.21,170) .. controls (158.08,170) and (160.42,172.24) .. (160.42,175) .. controls (160.42,177.76) and (158.08,180) .. (155.21,180) .. controls (152.33,180) and (150,177.76) .. (150,175) -- cycle ;
%Flowchart: Connector [id:dp4315654553902033] 
\draw  [fill={rgb, 255:red, 255; green, 255; blue, 255 }  ,fill opacity=1 ] (100,175) .. controls (100,172.24) and (102.33,170) .. (105.21,170) .. controls (108.08,170) and (110.42,172.24) .. (110.42,175) .. controls (110.42,177.76) and (108.08,180) .. (105.21,180) .. controls (102.33,180) and (100,177.76) .. (100,175) -- cycle ;
%Flowchart: Connector [id:dp29309822307272904] 
\draw  [fill={rgb, 255:red, 255; green, 255; blue, 255 }  ,fill opacity=1 ] (200,175) .. controls (200,172.24) and (202.24,170) .. (205,170) .. controls (207.76,170) and (210,172.24) .. (210,175) .. controls (210,177.76) and (207.76,180) .. (205,180) .. controls (202.24,180) and (200,177.76) .. (200,175) -- cycle ;
%Flowchart: Connector [id:dp974419747384212] 
\draw  [fill={rgb, 255:red, 0; green, 0; blue, 0 }  ,fill opacity=1 ] (249.58,115) .. controls (249.58,112.24) and (251.92,110) .. (254.79,110) .. controls (257.67,110) and (260,112.24) .. (260,115) .. controls (260,117.76) and (257.67,120) .. (254.79,120) .. controls (251.92,120) and (249.58,117.76) .. (249.58,115) -- cycle ;
%Flowchart: Connector [id:dp9794100447164118] 
\draw  [fill={rgb, 255:red, 255; green, 255; blue, 255 }  ,fill opacity=1 ] (250,175) .. controls (250,172.24) and (252.33,170) .. (255.21,170) .. controls (258.08,170) and (260.42,172.24) .. (260.42,175) .. controls (260.42,177.76) and (258.08,180) .. (255.21,180) .. controls (252.33,180) and (250,177.76) .. (250,175) -- cycle ;
%Flowchart: Connector [id:dp137575761480913] 
\draw   (145.21,115) .. controls (145.21,106.72) and (160.88,100) .. (180.21,100) .. controls (199.54,100) and (215.21,106.72) .. (215.21,115) .. controls (215.21,123.28) and (199.54,130) .. (180.21,130) .. controls (160.88,130) and (145.21,123.28) .. (145.21,115) -- cycle ;
%Flowchart: Connector [id:dp7500852188323093] 
\draw   (98,115) .. controls (98,101.19) and (124.86,90) .. (158,90) .. controls (191.14,90) and (218,101.19) .. (218,115) .. controls (218,128.81) and (191.14,140) .. (158,140) .. controls (124.86,140) and (98,128.81) .. (98,115) -- cycle ;
%Flowchart: Connector [id:dp3743361984621618] 
\draw   (95,175) .. controls (95,166.72) and (110.67,160) .. (130,160) .. controls (149.33,160) and (165,166.72) .. (165,175) .. controls (165,183.28) and (149.33,190) .. (130,190) .. controls (110.67,190) and (95,183.28) .. (95,175) -- cycle ;
%Flowchart: Connector [id:dp9338346101367194] 
\draw   (92,175) .. controls (92,161.19) and (118.86,150) .. (152,150) .. controls (185.14,150) and (212,161.19) .. (212,175) .. controls (212,188.81) and (185.14,200) .. (152,200) .. controls (118.86,200) and (92,188.81) .. (92,175) -- cycle ;

\end{tikzpicture}}
			\caption{State at end of round for  $\obj = L_3$.}
			\label{subsec:picking:fig:subset_case_round_end_2}
		\end{subfigure}
		\caption{Executions of $\GP(\pot,\sublist)$ and $\GP(\pot,\tblist)$ in the subset case.}
		\label{subsec:picking:fig:subset_case}
	\end{figure}
\end{figure}

When $\GP(\pot,\tblist)$ processes a subset~$\obj$ at iteration $|\tblist|$,  we have  ${\calb = \calb' \cup \{\obj\}}$.
We give an illustration in Figure~\ref{subsec:picking:fig:subset_case}, where we use the same convention as before.
Subfigure~\ref{subsec:picking:fig:subset_case_round_state} depicts
the state at the beginning of iteration $|\tblist|$.
Subfigure~\ref{subsec:picking:fig:subset_case_round_end_1} represents the end of the round for $\GP(\pot,\sublist)$, for which $\obj' = e_1$ is the first edge, and the processing order is $e_1, e_2, e_3, e_4, L_2$.
Subfigure~\ref{subsec:picking:fig:subset_case_round_end_2} represents the end of the round for $\GP(\pot,\tblist)$, for which $\obj = L_3$ is the last item of $\tblist$, and the processing order is $L_3,e_1,e_2,e_3,e_4, L_2$.
Notice that $L_1$ was not processed in neither execution because, after $e_2$ had been processed, $L_1$ stopped being maximal.
The subset $L_3$ was processed by $\GP(\pot,\tblist)$, but not by $\GP(\pot,\sublist)$, because $L_3$ is processed first in the execution of $\GP(\pot,\tblist)$, whereas it stopped being maximal in the execution of $\GP(\pot,\sublist)$ after $e_4$ had been processed.

The following lemma summarizes the main properties of threshold-tuples.

\begin{mylemma}
	\label{subsec:threshold-tuple:lemma:main_lemma}
	Assume that $(\pot,\tblist)$ is a threshold-tuple, and let $\obj = \tblist_{|\tblist|}$. Also, let $T$, $\calb$ and $y$ be the output computed by $\GP(\pot,\tblist)$, and let $T'$, $\calb'$ and $y'$ be the output computed by $\GP(\pot,\sublist)$.
	Then, $y = y'$ and
	\begin{enumerate}[itemsep=0pt,topsep=0pt, label=(\roman*),ref={\mbox{\rm $(\roman*)$}}]
		\setlength{\itemsep}{0pt}

		\item \label{subsec:threshold-tuple:lemma:main_lemma:2} if $\obj$ is an edge, then $\calb = \calb'$, $\obj \notin E(T')$, and $T \subseteq T' + \obj$;

		\item \label{subsec:threshold-tuple:lemma:main_lemma:3} if $\obj$ is a subset, then $T = T'$, $\obj \notin \calb'$, and $\calb = \calb' \cup \{\obj\}$.
	\end{enumerate}
\end{mylemma}

\begin{proof}
	Note that $\obj$ is the edge or subset processed in the $|\tblist|$-th iteration of $\GP(\pot,\tblist)$, because $\tblist$ is respected by $\pot$. Let $\obj'$ be the edge or subset processed in the $|\tblist|$-th iteration of $\GP(\pot,\sublist)$.
	By Lemma~\ref{subsec:threshold-tuple:lemma:threshold_subobject_is_an_edge}, $\obj \neq \obj'$, and $\obj'$ is an edge.
	We consider two cases, depending on the type of $\obj$.

	\medskip

	\textbf{Case 1:} Assume that both $\obj$ and $\obj'$ are edges.

	Let $\caly$ be the set of tight external edges considered to be processed in the $|\tblist|$-th iteration of $\GP(\pot,\tblist)$. Similarly, define $\caly'$ as the analogous set corresponding to $\GP(\pot,\sublist)$.
	Since $\obj$ and $\obj'$ are processed edges, $\obj \in \caly$, and $\obj' \in \caly'$. By Lemma~\ref{subsec:threshold-tuple:lemma:if_nu_is_respected_then_nu'_is_respected}, we have $\caly = \caly'$, and thus $\obj, \obj' \in \caly$.
	Because edges are processed with priority higher than subsets, each edge of $\caly$ is processed or becomes internal in the same round of iteration $|\tblist|$, and before any active tight subset is processed.

	Since $\obj \in \caly$, at some iteration $|\tblist| + \ell$, for some~$\ell \ge 0$, the execution $\GP(\pot,\sublist)$ must process edge $\obj$ or some other edge of $\caly$ which is parallel to~$\obj$.
	For some $i \ge 0$, let~$e_i'$ be the edge of $\caly$ processed by $\GP(\pot,\sublist)$ at iteration $|\tblist| + i$. Therefore, $e_0' = \obj'$, and $e_{\ell}'$ is parallel to $\obj$ at iteration $|\tblist| + \ell$.
	Similarly, let $e_i$ be the edge of $\caly$ processed by $\GP(\pot,\tblist)$ at iteration $|\tblist| + i$, and observe that~${e_0 = \obj}$.

	We claim that for each $1 \le i \le \ell$, $e_i = e_{i-1}'$. To see this, note that at the beginning of iteration $|\tblist|+i$ of $\GP(\pot,\tblist)$, edge $e_{i-1}'$ is tight and external.
	Since $e_{i-1}'$ was processed in iteration $|\tblist|+i-1$ of $\GP(\pot,\sublist)$, it has priority higher than any other tight edge of $\caly$ which is still external.
	Thus, $e_{i-1}'$ is selected to be processed, and indeed $e_i = e_{i-1}'$.

	At the end of iteration $|\tblist|+\ell$, the forest~$F$ computed by $\GP(\pot,\tblist)$ and the forest~$F'$ computed by $\GP(\pot,\sublist)$ are such that $F - \obj' = F' - \obj$. Moreover, $\obj$ and $\obj'$ are in the same connected components of both $F$~and~$F'$.
	Thus, the value of $y$, the set of external edges, and the collection of active subsets at the end of iteration are the same in both executions. This implies that the sequence of processed edges and subsets in succeeding iterations are the same.
	Therefore, the output of $\GP(\pot,\tblist)$ and $\GP(\pot,\sublist)$ are such that $y = y'$, $\calb = \calb'$, and $T = T' + \obj - \obj'$.
	This shows the lemma when $\obj$ is an edge.

	\textbf{Case 2:} Assume that $\obj$ is a subset and $\obj'$ is an edge.

	Let $\calx$ be the collection of active tight subsets considered to be processed in the $|\tblist|$-th iteration of $\GP(\pot,\tblist)$. Similarly, define $\calx'$ as the analogous collection corresponding to $\GP(\pot,\sublist)$.
	Since $\obj$ is a processed subset, $\obj \in \calx$. By Lemma~\ref{subsec:threshold-tuple:lemma:if_nu_is_respected_then_nu'_is_respected}, we have $\calx = \calx'$.

	Suppose that, after processing edge $\obj'$, $\GP(\pot, \sublist)$ processes other $\ell$ edges in the same round of iterations.
	Since edges have higher priority, no subset is processed before every tight external edge is processed or becomes internal.
	Let $e_0', e_1', \dots, e_{\ell}'$ be the edges processed in iteration $|\tblist|, |\tblist|+1, \dots, |\tblist|+\ell$ of $\GP(\pot,\sublist)$, respectively.
	Note that $e_0' = \obj'$.

	By Lemma~\ref{subsec:threshold-tuple:lemma:if_nu_is_respected_then_nu'_is_respected}, each $e_i'$ is also a tight external edge in iteration $|\tblist|$ of $\GP(\pot, \tblist)$.
	Processing a subset does not change the set of tight external edges, thus, after processing subset $\obj$, $\GP(\pot, \tblist)$ must process the sequence of edges $e_0', e_1', \dots, e_{\ell}'$, such that $e_i'$ is processed at iteration $|\tblist| + i + 1$ of $\GP(\pot, \tblist)$. This implies that the set of edges processed by $\GP(\pot, \sublist)$  up to iteration $|\tblist| + \ell$ equals the set of edges processed by $\GP(\pot, \tblist)$ at up to iteration $|\tblist| + \ell + 1$.

	As a consequence, any tight subset of $\calx$ that remains maximal at the end of iteration $|\tblist| + \ell$ of $\GP(\pot, \sublist)$ is also tight and maximal at the end of iteration $|\tblist| + \ell+1$ of $\GP(\pot, \tblist)$, and vice-versa. After processing edges, any such subset which is still unprocessed is processed.
	Therefore, at the end of the rounds corresponding to $\GP(\pot, \sublist)$ and $\GP(\pot, \tblist)$, the value of $y$, the set of external edges, and the collection of active subsets are the same in both executions.
	As in the previous case, this implies that the sequence of processed edges and subsets in succeeding iterations are the same.

	Now observe that both executions process the same set of edges, thus $T = T'$.
	Also, any subset processed by $\GP(\pot, \sublist)$ is processed by $\GP(\pot, \tblist)$.
	Conversely, if a subset processed by $\GP(\pot, \tblist)$ is not processed by $\GP(\pot, \sublist)$, then this subset must be~$\obj$.
	Thus, $\calb = \calb' \cup \{\obj\}$.
	To show that $\obj \notin \calb'$, note that if $\calb = \calb'$, then $\GW(\pot, \tblist)$ and $\GW(\pot, \sublist)$ would return identical trees, which is not possible since $(\pot, \tblist)$ is a threshold-tuple.
	This completes the lemma.
\end{proof}

	%!TEX root = main.tex

\subsection{Computing a threshold-tuple}
\label{subsec:TS}

Before describing the algorithm to compute a threshold-tuple $(\pot,\tblist)$, we need some definitions and corresponding auxiliary lemmas.

\subsubsection{Increase-function}

In the following, suppose that we are given a fixed tie-breaking list $\tblist$ and a real interval $[a,b]$ such that, for every $\pot \in [a,b]$, $\tblist$~is respected by~$\pot$.
Then, there is a fixed pair $(\call, \calb)$ such that, for every $\pot \in [a,b]$, the laminar collection and the collection of processed subsets computed by $\GP(\pot, \tblist)$ at the end of iteration~$|\tblist|$ correspond to $\call$~and~$\calb$, respectively.
Let ${\cala = \call^* \setminus \calb}$, and note that $\cala$ is the collection of active subsets at the beginning of iteration~${|\tblist|+1}$.
Define the \emph{object-collection} $\actObj$ corresponding to $\tblist$ and $[a,b]$ as the collection that contains the active subsets, the external edges with at least one active endpoint, and the tight external edges with no active endpoints at the beginning of iteration $|\tblist|+1$.
Observe that, if $\GP(\pot, \tblist)$ processes an edge or a subset~$\obj$ at iteration $|\tblist|+1$, then $\obj \in \actObj$.
Thus, $\actObj$ contains all subsets and edges that might be processed at iteration $|\tblist|+1$ for some $\pot \in [a,b]$.

At the beginning of iteration $|\tblist|+1$ of $\GP(\pot, \tblist)$, the variable $y_S$ of every active subset $S \in \cala$ is increased by the largest value $\Delta \ge 0$, for which no inequality corresponding to an edge or a subset is violated.
To compute $\Delta$, one first finds, for each edge or subset $\obj \in \actObj$, the value $\Delta_\obj$ to increase each variable $y_S$ such that $\obj$ becomes tight.
Since the considered tie-breaking list~$\tblist$ and interval~$[a,b]$ are fixed, the value $\Delta_\obj$ depends only on $\pot$.
The \emph{increase-function} of $\obj$ corresponding to $\tblist$ and $[a,b]$ is the function $\dfunc_\obj : [a,b] \rightarrow \mathbb{Q}_{\geq 0}$ such that $\dfunc_\obj(\pot) = \Delta_\obj$.

\begin{mylemma}
	\label{subsec:TS:lemma:dfunc_is_poly}
	For every $\obj \in \actObj$, the increase-function of $\obj$ is linear in $\pot$.
\end{mylemma}

\begin{proof}
	For some integer $\ell$, with $0 \le \ell \le |\tblist|$, let $T_\ell$ be the sublist of $\tblist$ defined as $T_\ell = (\tblist_1, \tblist_2, \dots, \tblist_\ell)$.
	Observe that the tie-breaking list $T_\ell$ is respected by $\pot$ for every~${\pot \in [a,b]}$, then there is an object-collection $\actObj_\ell$ corresponding to $T_\ell$ and $[a,b]$.
	Also, for each index~${1 \le \ell \le |\tblist|}$, the entry~$\tblist_\ell$ of the tie-breaking list is processed at iteration~$\ell$ of $\GP(\pot,\tblist)$, and thus~$\tblist_\ell \in \actObj_{\ell-1}$.
	Note that ${\actObj_{|\tblist|} = \actObj}$, thus, by defining $\tblist_{|\tblist|+1} = \obj$, we have $\tblist_\ell \in \actObj_{\ell-1}$ for index~${\ell =|\tblist|+1}$ as well.

	Let~$\dfunc_\ell(\pot)$ be the increase-function of $\tblist_\ell$ corresponding to $T_{\ell-1}$ and $[a,b]$.
	We will show by induction that, for each $1 \le \ell \leq |\tblist| + 1$, function~$\dfunc_\ell(\pot)$ is linear in $\pot$.
	The lemma will follow by taking $\ell = |\tblist| + 1$.

	First, consider the case $\ell = 1$.
	If $\tblist_1$ is a subset, then it must be a singleton~$\{v\}$ for some vertex $v$, because each maximal subset is a singleton in the first iteration of the growth-phase. Since at this moment, $y = 0$, the increase to turn $\tblist_1$ tight is $\dfunc_1(\pot) = \vpen_v^\pot = \vpen_v + \pot$.
	Similarly, if $\tblist_1$ is an edge $e$, it must have two active endpoints at the first iteration, and thus $\dfunc_1(\pot) = \ecost_e / 2$.

	Now, assume that $\dfunc_i(\pot)$ is linear in $\pot$ for each $i < \ell$.
	Consider some subset~$S$, and let $\Lambda_S$ denote the set of indices $i$ such that $i < \ell$ and $S$ was an active subset at the beginning of iteration~$i$.
	It follows that, at the beginning of iteration~$\ell$, $y_S = \sum_{i \in \Lambda_S} \dfunc_i(\pot)$, which is linear in~$\pot$, since each $\dfunc_i$ is linear in~$\pot$, and $\Lambda_S$ does not depend on $\pot$.

	Let $\cala$ be the collection of active subsets at the beginning of iteration~$\ell$.
	If $\tblist_\ell$ is an external edge $e$ with no active endpoints, then it is tight, thus $\dfunc_\ell(\pot) = 0$.
	In the case that $\tblist_\ell$ is an edge $e$ with at least one active endpoint, the left-hand side of inequality~\eqref{sec:preliminaries:restriction_2} corresponding to $e$ is $\sum_{S: e \in \delta(S)} y_S$, that is a linear function, say~$f(\pot)$.
	It follows that the increase of $y_S$ for active subsets~$S$ which is necessary to turn~$e$ tight is
	\[
		\dfunc_\ell(\pot) = \frac{\ecost_e - f(\pot)}{|\{S \in \cala : e \in \delta(S)|}.
	\]
	In both cases, $\dfunc_\ell(\pot)$ is linear in~$\pot$.

	Similarly, if $\tblist_\ell$ is a subset $L$, then the left-hand side of inequality~\eqref{sec:preliminaries:restriction_1} corresponding to $L$ is $\sum_{S:S \subseteq L} y_S$, that is a function which is linear in $\pot$.
	Let this function be $f(\pot)$, then,
	\[
		\dfunc_\ell(\pot) = \vpen_L^\pot - f(\pot) = \vpen_L + \pot|L| - f(\pot),
	\]
	which is linear in~$\pot$.
\end{proof}

\subsubsection{Diverging potential}

For each potential $\pot \in [a,b]$, the element of the object-collection $\actObj$ corresponding to $\tblist$ and $[a,b]$ that is processed at iteration $|\tblist| + 1$ of the growth-phase is determined by the increase-function which evaluates to the smallest value.
Suppose that executing $\GP(a, \tblist)$ processes some edge or subset~$\obj_1$ at iteration $|\tblist|+1$, while executing $\GP(b, \tblist)$ processes a distinct edge or subset~$\obj_2$ at iteration $|\tblist|+1$.
In this situation, $\obj_1$ is selected in the former execution because ${\dfunc_{\obj_1}(a) \le \dfunc_{\obj'}(a)}$ for any ${\obj' \in \actObj}$,  while $\obj_2$ is selected in the latter because ${\dfunc_{\obj_2}(b) \le \dfunc_{\obj'}(b)}$ for any ${\obj' \in \actObj}$.
Since increase-functions are linear, this implies that there must be some potential $p \in (a,b)$ such that, at iteration~${|\tblist|+1}$, $\GP(\pot, \tblist)$ processes~$\obj_1$ for $\pot < p$, and processes a distinct element in $\actObj$ for $\pot > p$.

Formally, we say that $p \in (a,b)$ is a \emph{diverging potential} if there are edges or subsets ${\obj, \obj' \in \actObj}$ with distinct increase-functions $\dfunc_{\obj}, \dfunc_{\obj'}$ such that: (i)~${\dfunc_{\obj}(p) = \dfunc_{\obj'}(p)}$; and (ii)~execution $\GP(p, \tblist)$ processes $\obj$ at iteration $|\tblist|+1$.
Observe that, since increase-functions are linear, the number of diverging potentials is finite and they can be computed in polynomial time using standard line-intersection algorithms.

The definition implies that, for each pair of consecutive diverging potentials, there is an edge or a subset which is tight at iteration $|\tblist|+1$ when the growth-phase is executed with either diverging potential.
This is formalized next.

\begin{mylemma}
	\label{subsec:TS:lemma:diverging_pot}
	Let $p_1 < p_2 < \cdots < p_m$ be the sequence of diverging potentials in~$(a,b)$, and let $p_0 = a$ and $p_{m+1} = b$.
	Also, let $\pot_0 \in (p_j, p_{j+1})$ for some $0 \le j \le m$.
	If $\obj \in \actObj$ is processed at iteration $|\tblist|+1$ of $\GP(\pot_0, \tblist)$, then $\obj$ is tight at the end of iteration $|\tblist|+1$ of $\GP(\pot, \tblist)$ for every $\pot \in [p_j,p_{j+1}]$.
\end{mylemma}

\begin{proof}
	Because $\obj$ has been processed, $\dfunc_\obj(\pot_0) \le \dfunc_{\obj'}(\pot_0)$ for any ${\obj' \in \actObj}$ with distinct increase-function.
	Since $p_j$ and $p_{j+1}$ are consecutive potentials, functions $\dfunc_{\obj}$ and $\dfunc_{\obj'}$ do not intersect for any $\pot \in (p_j,p_{j+1})$. Thus, by the continuity of the increase-functions, ${\dfunc_\obj(\pot) \le \dfunc_{\obj'}(\pot)}$ for each $\pot \in [p_j,p_{j+1}]$.
	It follows that, at iteration $|\tblist|+1$ of $\GP(\pot, \tblist)$, the increase of the variables is $\dfunc_\obj(\pot)$, and~$\obj$ is tight at the end of the iteration.
\end{proof}

\subsubsection{Threshold-tuple search algorithm}

We describe the \emph{threshold-tuple search} algorithm, denoted by $\TS$.
A listing is given in Algorithm~\ref{subsec:TS:alg:TS}.
To find a threshold-tuple, the algorithm constructs a tie-breaking list $\tblist$ iteratively. It maintains the following invariants:~(i)~$\tblist$~is respected by~$\pot$ for any~${\pot \in [a,b]}$; (ii) $\GW(a, \tblist)$ returns a tree spanning less than $k$ vertices,  and $\GW(b, \tblist)$ returns a tree spanning at least $k$ vertices.

Initialize the variables by making $\tblist = \emptyset$, $a = 0$ and $b = \ecost_E + 1$,
and start the iteration process.
At each iteration~$i$, compute the sequence of diverging potentials $p_1 < \cdots < p_m$ in the range $(a,b)$, and let $p_0 = a$ and $p_{m+1} = b$.
Since $\GW(a, \tblist)$ returns less than $k$ vertices and $\GW(b, \tblist)$ returns at least $k$ vertices, there must be some index~$j$ for which $\GW(p_j, \tblist)$ returns less than $k$ vertices, and $\GW(p_{j+1}, \tblist)$ returns at least $k$ vertices.

Discover the edge or subset $\obj$ which is processed at iteration~$i$ of $\GP(\pot_0,\tblist)$ for some arbitrary $\pot_0 \in (p_j ,p_{j+1})$.
Then, extend the tie-breaking list by appending $\obj$ at the end of $\tblist$. Note that, by Lemma~\ref{subsec:TS:lemma:diverging_pot}, the extended tie-breaking list $\tblist$ is respected by~$\pot$ for every $\pot \in [p_j ,p_{j+1}]$, then the first invariant is maintained by making $a = p_j$ and $b = p_{j+1}$.

Following, check if executing $\GW(p_j, \tblist)$ using the updated tie-breaking list returns at least $k$ vertices. If this is so, since $\GW(p_j, \sublist)$ returns less than $k$ vertices, $(p_j, \tblist)$ is a threshold-tuple, and the algorithm stops.
Analogously, if $\GW(p_{j+1}, \tblist)$ returns less than $k$ vertices, $(p_{j+1}, \tblist)$ is a threshold-tuple.
If neither is the case, then the second invariant is maintained, and the process is repeated.

Next, we show that $\TS$ indeed finds a threshold-tuple.

\begin{algorithm}[htp]
	\DontPrintSemicolon
	\caption{$\TS$}
	\label{subsec:TS:alg:TS}

	Initialize $a \gets 0$, $b \gets \ecost_{E} + 1$, and $\tblist \gets \emptyset$\;
	\For{\textnormal{$i \gets 1, 2, \cdots$}}
	{
		Let $p_1 < \cdots < p_m$ be the diverging potentials w.r.t. $\tblist$ in  $[a,b]$\;
		Let $p_0 \gets a$ and $p_{m+1} \gets b$\;
		Find index $j$ such that: \newline
			$\GW(p_j, \tblist)$ returns less than $k$ vertices, and  \newline
			$\GW(p_{j+1}, \tblist)$ returns at least $k$ vertices

		\nonl\;

		Let $\pot_0 = (p_j + p_{j+1})/2$\;
		Let $\obj$ be the $i$-th processed edge or subset of $\GP(\pot_0,\tblist)$\label{sec:TS:alg:TS_i-th_object}\;
		Append $\obj$ to the end of~$\tblist$\;
		Update $a \gets p_{j}$ and $b \gets p_{j+1}$\;

		\nonl\;

		\uIf(\label{sec:TS:alg:TS_if2}){\textnormal{$\GW(p_j, \tblist)$ returns at least $k$ vertices}}
			{\Return{$(p_j,\tblist)$}}

		\ElseIf(\label{sec:TS:alg:TS_if1}){\textnormal{$\GW(p_{j+1}, \tblist)$ returns less than $k$ vertices}}
			{\Return{$(p_{j+1},\tblist)$}}
	}
\end{algorithm}

\begin{mylemma}
	\label{subsec:TS:lemma:TS_inv}
	If executing $\GW(0, \emptyset)$ returns a tree spanning less than $k$ vertices, then $\TS$ returns a threshold-tuple in polynomial time.
\end{mylemma}

\begin{proof}
	We argue that the algorithm indeed mantains the invariants: (i)~$\tblist$~is respected by~$\pot$ for any~${\pot \in [a,b]}$; (ii)~$\GW(a, \tblist)$ returns a tree spanning less than $k$ vertices,  and $\GW(b, \tblist)$ returns a tree spanning at least $k$ vertices.
	Observe that an empty tie-breaking list is respected by~$\pot$ for every potential~$\pot$. Also, $\GW(0, \emptyset)$ returns a tree spanning less than $k$ vertices by assumption, and $\GW(\ecost_E+1, \emptyset)$ returns a tree spanning at least $k$ vertices by Lemma~\ref{subsec:GW:lemma:high_lambda_spans_V_G_r}.
	Thus, at the beginning of the iteration process, both invariants hold.
	By the description of the algorithm, if the invariants hold at the beginning of an iteration, then they hold at the end of the iteration as well.

	Observe that, if the algorithm stops, then it returns a threshold-tuple. We will show that $\TS$ stops after at most $3|V|-3$ iterations.
	Suppose not. Then, at the beginning of iteration $3|V|-2$, the tie-breaking list has size $|\tblist| = 3|V|-3$.
	By Lemma~\ref{subsec:GP:lemma:GP_iterates_at_most_3V-3}, the growth-phase executes at most $3|V|-3$ iterations. It follows that the sequence of edges and subsets processed by $\GP(a,\tblist)$ and $\GP(b,\tblist)$ are the same, because we know that $\tblist$ is respected by $\pot \in [a,b]$ from invariant~(i).
	This implies that $\GP(a,\tblist)$ and $\GP(b,\tblist)$ have the same output, and thus $\GW(a,\tblist)$ and $\GW(b,\tblist)$ return identical trees, which contradicts invariant~(ii).
\end{proof}

%!TEX root = main.tex

\section{Finding a solution with a threshold-tuple}
\label{sec:find_solution}

Assume that we are given a threshold-tuple $(\pot,\tblist)$.
Then, by executing $\GW(\pot,\tblist)$ and $\GW(\pot,\sublist)$, we obtain two trees, one which spans less than $k$ vertices, and the other, at least $k$ vertices.
Let $\hatt_\subminus$ and $\hatt_\subplus$ be the trees
with less than $k$ vertices, and at least $k$ vertices, respectively,
and let $(T_\subminus,\calb_\subminus)$ and $(T_\subplus,\calb_\subplus)$ denote the corresponding pairs computed in the growth-phase.

Denote by $\call_\subminus$ and $\call_\subplus$ the laminar collections computed in the growth-phase.
While these laminar collections might be (slightly) different, Lemma~\ref{subsec:threshold-tuple:lemma:main_lemma} states that $\GP(\pot,\tblist)$ and $\GP(\pot,\sublist)$ compute identical vectors~$y$.
Moreover each edge of $T_\subminus \cup T_\subplus$ and
each subset of $\calb_\subminus \cup \calb_\subplus$ is tight.
The objective of this section is to find a tree $T$ from $T_\subminus \cup T_\subplus$ spanning at least $k$ vertices.

Recall that that $\call_\subminus(S)$ is the collection of subsets which contain some but not all vertices of $S$, and let $T^*$ be an optimal solution.
Also, suppose that~$V$ is the minimal subset in~$\call_\subminus$ containing $V(T^*)$; we will relax this assumption later.
To bound the cost of this optimal solution, one can use Lemma~\ref{subsec:approx:lemma:opt_lower_bound}, which gives the lower bound
\[
% \textstyle
\left(
\sum_{L \in \call_\subminus(V)} y_L - \pot|V \setminus V(T^*)|
\right)
\leq \ecost_{E(T^*)} + \vpen_{V \setminus V(T^*)}.
\]

To bound the cost of $T$, we need to prove an inequality analogous to the one given by Goemans and Williamson's analysis~\cite{Goemans1995}, i.e., we want that
\[
% \textstyle
\ecost_{E(T)} + 2 \vpen_{V \setminus V(T)} ~\leq~ 2
\left(
\sum_{L \in \call_\subminus(V)} y_L
-  \pot|V \setminus V(T)|
\right).
\]

Therefore, to obtain a $2$-approximation, it is sufficient to find a tree such that $|V(T)| \le  |V(T^*)|$.
But, since any feasible solution must span at least $k$ vertices, this means that our goal is to compute a tree $T$ that spans exactly $k$ vertices.
Such a tree is constructed by the \emph{picking-vertices} algorithm.
This algorithm follows some of the ideas due to Garg~\cite{Garg2005}.

We now explain our assumption that $V$ is the minimal subset containing~$V(T^*)$.
Recall that ${V \in \call_\subminus}$, thus there are two subsets $L_1, L_2 \in \call_\subminus$ such that ${V = L_1 \cup L_2}$, because $\call_\subminus$ is binary laminar.
If $V$ is not the minimal subset in $\call_\subminus$ containing $V(T^*)$, then either $L_1$ or $L_2$ contains $V(T^*)$.
But, because~$T^*$ contains the root~$r$, we can decide which one of
$L_1$ or $L_2$ contains~$V(T^*)$.
Therefore, either our assumption holds, or we can safely reduce the instance's size.

	%!TEX root = main.tex

\subsection{Picking \texorpdfstring{$k$}{} vertices}
\label{subsec:picking}

Let $\obj$ be the edge or subset processed at iteration~$|\tblist|$ of $\GP(\pot, \tblist)$.
Assume for the moment that $\obj$ is a subset.
Then, Lemma~\ref{subsec:threshold-tuple:lemma:main_lemma} implies that ${T_\subminus = T_\subplus}$ and collections $\calb_\subminus$ and $\calb_\subplus$ differ in exactly one subset.
Since~$\hatt_\subminus$ spans fewer vertices than $\hatt_\subplus$, collection~$\calb_\subminus$ contains more subsets than collection~$\calb_\subplus$, by the monotonicity of the pruning operation.
Therefore, in this case, ${\calb_\subminus = \calb_\subplus \cup \{\obj\}}$.

We would like to obtain a tree from $T_\subplus$ spanning exactly~$k$ vertices.
To use Goemans and Williamson's analysis, this tree needs to be pruned with a collection of tight subsets.
If we prune $T_\subplus$ using $\calb_\subplus$, the resulting tree $\hatt_\subplus$ would have too many vertices, but if we add subset $\obj$ to the pruning collection, then the pruning algorithm would delete a sequence of tight subsets, until finding the tree~$\hatt_\subminus$, as illustrated in Figure~\ref{subsec:PV:fig:compute-path_subset_case}.
We will soon show that these subsets form a path in~$\hatt_\subminus$.

\begin{figure}[htb!]
	\center
	\scalebox{1}{%!TEX root = ../../main.tex

\tikzset{every picture/.style={line width=0.75pt}} %set default line width to 0.75pt

\begin{tikzpicture}[x=0.75pt,y=0.75pt,yscale=-1,xscale=1]
%uncomment if require: \path (0,432.66666412353516); %set diagram left start at 0, and has height of 432.66666412353516

%Straight Lines [id:da06242691861421057]
\draw    (180,270) -- (223.75,253.75) ;

%Straight Lines [id:da9754299907955215]
\draw    (150,270) -- (114.5,257.5) ;

%Straight Lines [id:da9665265710670707]
\draw    (343.75,253.75) -- (294.5,270.5) ;

%Straight Lines [id:da5489953694172667]
\draw    (234.5,257.5) -- (283.75,273.75) ;

%Straight Lines [id:da687305837020906]
\draw    (354.5,256.5) -- (427.25,273.75) ;

%Flowchart: Connector [id:dp6046010010545533]
\draw  [fill={rgb, 255:red, 255; green, 255; blue, 255 }  ,fill opacity=1 ] (384.5,273.75) .. controls (384.5,260.63) and (403.64,250) .. (427.25,250) .. controls (450.86,250) and (470,260.63) .. (470,273.75) .. controls (470,286.87) and (450.86,297.5) .. (427.25,297.5) .. controls (403.64,297.5) and (384.5,286.87) .. (384.5,273.75) -- cycle ;
%Flowchart: Connector [id:dp18641603440801924]
\draw  [color={rgb, 255:red, 128; green, 128; blue, 128 }  ,draw opacity=1 ][fill={rgb, 255:red, 128; green, 128; blue, 128 }  ,fill opacity=1 ] (330,253.75) .. controls (330,246.16) and (336.16,240) .. (343.75,240) .. controls (351.34,240) and (357.5,246.16) .. (357.5,253.75) .. controls (357.5,261.34) and (351.34,267.5) .. (343.75,267.5) .. controls (336.16,267.5) and (330,261.34) .. (330,253.75) -- cycle ;
%Flowchart: Connector [id:dp11094405081277992]
\draw  [color={rgb, 255:red, 128; green, 128; blue, 128 }  ,draw opacity=1 ][fill={rgb, 255:red, 128; green, 128; blue, 128 }  ,fill opacity=1 ] (270,273.75) .. controls (270,266.16) and (276.16,260) .. (283.75,260) .. controls (291.34,260) and (297.5,266.16) .. (297.5,273.75) .. controls (297.5,281.34) and (291.34,287.5) .. (283.75,287.5) .. controls (276.16,287.5) and (270,281.34) .. (270,273.75) -- cycle ;
%Flowchart: Connector [id:dp9533810081867387]
\draw  [color={rgb, 255:red, 128; green, 128; blue, 128 }  ,draw opacity=1 ][fill={rgb, 255:red, 128; green, 128; blue, 128 }  ,fill opacity=1 ] (210,253.75) .. controls (210,246.16) and (216.16,240) .. (223.75,240) .. controls (231.34,240) and (237.5,246.16) .. (237.5,253.75) .. controls (237.5,261.34) and (231.34,267.5) .. (223.75,267.5) .. controls (216.16,267.5) and (210,261.34) .. (210,253.75) -- cycle ;
%Straight Lines [id:da5533696307541149]
\draw  [dash pattern={on 0.84pt off 2.51pt}]  (150,270) -- (180,270) ;

%Flowchart: Connector [id:dp5085214697324458]
\draw  [color={rgb, 255:red, 128; green, 128; blue, 128 }  ,draw opacity=1 ][fill={rgb, 255:red, 128; green, 128; blue, 128 }  ,fill opacity=1 ] (30,273.75) .. controls (30,266.16) and (36.16,260) .. (43.75,260) .. controls (51.34,260) and (57.5,266.16) .. (57.5,273.75) .. controls (57.5,281.34) and (51.34,287.5) .. (43.75,287.5) .. controls (36.16,287.5) and (30,281.34) .. (30,273.75) -- cycle ;
%Straight Lines [id:da8371817633519523]
\draw    (103.75,253.75) -- (57,269) ;

%Flowchart: Connector [id:dp18338959252268783]
\draw  [color={rgb, 255:red, 128; green, 128; blue, 128 }  ,draw opacity=1 ][fill={rgb, 255:red, 128; green, 128; blue, 128 }  ,fill opacity=1 ] (90,253.75) .. controls (90,246.16) and (96.16,240) .. (103.75,240) .. controls (111.34,240) and (117.5,246.16) .. (117.5,253.75) .. controls (117.5,261.34) and (111.34,267.5) .. (103.75,267.5) .. controls (96.16,267.5) and (90,261.34) .. (90,253.75) -- cycle ;

% Text Node
\draw (344,230) node [scale=0.9,color={rgb, 255:red, 0; green, 0; blue, 0 }  ,opacity=1 ] [align=left] {$\displaystyle D_{\ell}$};
% Text Node
\draw (286,250) node [scale=0.9,color={rgb, 255:red, 0; green, 0; blue, 0 }  ,opacity=1 ] [align=left] {$\displaystyle D_{\ell-1}$};
% Text Node
\draw (226,230) node [scale=0.9,color={rgb, 255:red, 0; green, 0; blue, 0 }  ,opacity=1 ] [align=left] {$\displaystyle D_{\ell-2}$};
% Text Node
\draw (427.25,273.75) node [scale=1.44] [align=left] {$\displaystyle \hatt_\subminus$};
% Text Node
\draw (104.5,230) node [scale=0.9,color={rgb, 255:red, 0; green, 0; blue, 0 }  ,opacity=1 ] [align=left] {$\displaystyle D_{2}$};
% Text Node
\draw (43,250) node [scale=0.9,color={rgb, 255:red, 0; green, 0; blue, 0 }  ,opacity=1 ] [align=left] {$\displaystyle D_{1}$};
% Text Node
\draw (426.5,238) node [scale=0.9,color={rgb, 255:red, 0; green, 0; blue, 0 }  ,opacity=1 ] [align=left] {$\displaystyle D_{\ell+1}$};
% Text Node
\draw (215.5,203.5) node [scale=1.44] [align=left] {$\displaystyle \hatt_\subplus$};

\end{tikzpicture}}
	\caption{Sequence of deleted subsets when $\obj$ is a subset.}
	\label{subsec:PV:fig:compute-path_subset_case}
\end{figure}
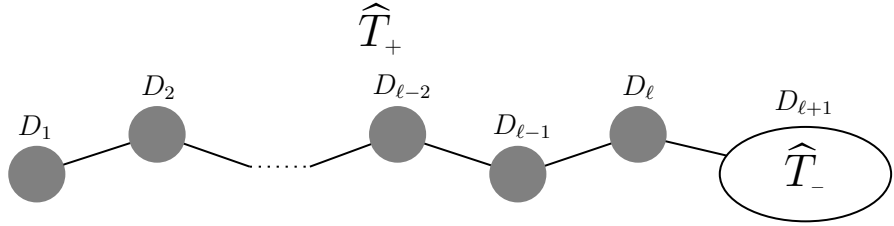

\begin{mydefinition}
Consider a laminar collection of subsets~$\calb$, and let $H$ be a connected graph. A sequence of subsets $D_1,D_2, \dots,D_{\ell+1}$ which partition~$V(H)$ is called a \emph{\mbox{subset path} of~$H$ processed with~$\calb$} if $H[D_{\ell+1}]$ is pruned with $\calb$, and, for each $1\le i \le \ell$,
\begin{enumerate}[itemsep=-2pt,topsep=-6pt,label=(\roman*)]
	\item $H$ has an edge connecting a vertex $v_i \in D_i$ to $D_{i+1}$;
	\item there is $B_i \in \calb$ such that $D_i \subseteq B_i$ and $B_i \cap D_j = \emptyset$, for every ${j > i}$;
	\item $H[D_i \cup \dots \cup D_{\ell+1}]$ is connected and pruned with $\{B \in \calb : v_i \notin B\}$.
\end{enumerate}
\end{mydefinition}

Figure~\ref{subsec:PV:fig:compute-path_subset_case2} repeats Figure~\ref{subsec:PV:fig:compute-path_subset_case}, but representing the corresponding processed subsets and vertices from the definition of a subset path.
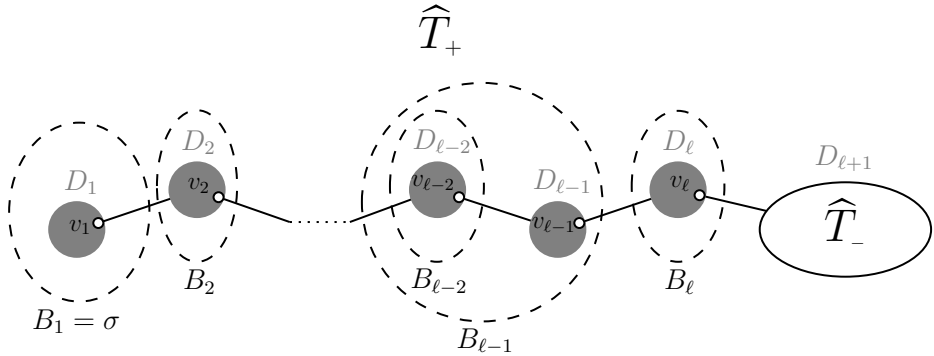
\begin{figure}[htb!]
	\center
	\scalebox{1}{%!TEX root = ../../main.tex

\tikzset{every picture/.style={line width=0.75pt}} %set default line width to 0.75pt        

\begin{tikzpicture}[x=0.75pt,y=0.75pt,yscale=-1,xscale=1]
%uncomment if require: \path (0,432.66666412353516); %set diagram left start at 0, and has height of 432.66666412353516

%Straight Lines [id:da06242691861421057] 
\draw    (180,270) -- (223.75,253.75) ;

%Straight Lines [id:da9754299907955215] 
\draw    (150,270) -- (114.5,257.5) ;

%Straight Lines [id:da9665265710670707] 
\draw    (343.75,253.75) -- (294.5,270.5) ;

%Straight Lines [id:da5489953694172667] 
\draw    (234.5,257.5) -- (283.75,273.75) ;

%Straight Lines [id:da687305837020906] 
\draw    (354.5,256.5) -- (427.25,273.75) ;

%Flowchart: Connector [id:dp6046010010545533] 
\draw  [fill={rgb, 255:red, 255; green, 255; blue, 255 }  ,fill opacity=1 ] (384.5,273.75) .. controls (384.5,260.63) and (403.64,250) .. (427.25,250) .. controls (450.86,250) and (470,260.63) .. (470,273.75) .. controls (470,286.87) and (450.86,297.5) .. (427.25,297.5) .. controls (403.64,297.5) and (384.5,286.87) .. (384.5,273.75) -- cycle ;
%Flowchart: Connector [id:dp18641603440801924] 
\draw  [color={rgb, 255:red, 128; green, 128; blue, 128 }  ,draw opacity=1 ][fill={rgb, 255:red, 128; green, 128; blue, 128 }  ,fill opacity=1 ] (330,253.75) .. controls (330,246.16) and (336.16,240) .. (343.75,240) .. controls (351.34,240) and (357.5,246.16) .. (357.5,253.75) .. controls (357.5,261.34) and (351.34,267.5) .. (343.75,267.5) .. controls (336.16,267.5) and (330,261.34) .. (330,253.75) -- cycle ;
%Flowchart: Connector [id:dp11094405081277992] 
\draw  [color={rgb, 255:red, 128; green, 128; blue, 128 }  ,draw opacity=1 ][fill={rgb, 255:red, 128; green, 128; blue, 128 }  ,fill opacity=1 ] (270,273.75) .. controls (270,266.16) and (276.16,260) .. (283.75,260) .. controls (291.34,260) and (297.5,266.16) .. (297.5,273.75) .. controls (297.5,281.34) and (291.34,287.5) .. (283.75,287.5) .. controls (276.16,287.5) and (270,281.34) .. (270,273.75) -- cycle ;
%Flowchart: Connector [id:dp9533810081867387] 
\draw  [color={rgb, 255:red, 128; green, 128; blue, 128 }  ,draw opacity=1 ][fill={rgb, 255:red, 128; green, 128; blue, 128 }  ,fill opacity=1 ] (210,253.75) .. controls (210,246.16) and (216.16,240) .. (223.75,240) .. controls (231.34,240) and (237.5,246.16) .. (237.5,253.75) .. controls (237.5,261.34) and (231.34,267.5) .. (223.75,267.5) .. controls (216.16,267.5) and (210,261.34) .. (210,253.75) -- cycle ;
%Straight Lines [id:da5533696307541149] 
\draw  [dash pattern={on 0.84pt off 2.51pt}]  (150,270) -- (180,270) ;

%Flowchart: Connector [id:dp5085214697324458] 
\draw  [color={rgb, 255:red, 128; green, 128; blue, 128 }  ,draw opacity=1 ][fill={rgb, 255:red, 128; green, 128; blue, 128 }  ,fill opacity=1 ] (30,273.75) .. controls (30,266.16) and (36.16,260) .. (43.75,260) .. controls (51.34,260) and (57.5,266.16) .. (57.5,273.75) .. controls (57.5,281.34) and (51.34,287.5) .. (43.75,287.5) .. controls (36.16,287.5) and (30,281.34) .. (30,273.75) -- cycle ;
%Flowchart: Connector [id:dp33205217744375837] 
\draw  [dash pattern={on 4.5pt off 4.5pt}] (10,265) .. controls (10,240.15) and (25.67,220) .. (45,220) .. controls (64.33,220) and (80,240.15) .. (80,265) .. controls (80,289.85) and (64.33,310) .. (45,310) .. controls (25.67,310) and (10,289.85) .. (10,265) -- cycle ;
%Flowchart: Connector [id:dp9760727169081618] 
\draw  [dash pattern={on 4.5pt off 4.5pt}] (83.75,251.88) .. controls (83.75,230.82) and (92.7,213.75) .. (103.75,213.75) .. controls (114.8,213.75) and (123.75,230.82) .. (123.75,251.88) .. controls (123.75,272.93) and (114.8,290) .. (103.75,290) .. controls (92.7,290) and (83.75,272.93) .. (83.75,251.88) -- cycle ;
%Flowchart: Connector [id:dp26752114859578] 
\draw  [dash pattern={on 4.5pt off 4.5pt}] (200,252.03) .. controls (200,231.06) and (210.49,214.06) .. (223.44,214.06) .. controls (236.38,214.06) and (246.88,231.06) .. (246.88,252.03) .. controls (246.88,273) and (236.38,290) .. (223.44,290) .. controls (210.49,290) and (200,273) .. (200,252.03) -- cycle ;
%Flowchart: Connector [id:dp9561736810573018] 
\draw  [dash pattern={on 4.5pt off 4.5pt}] (320.63,252.03) .. controls (320.63,231.06) and (330.98,214.06) .. (343.75,214.06) .. controls (356.52,214.06) and (366.88,231.06) .. (366.88,252.03) .. controls (366.88,273) and (356.52,290) .. (343.75,290) .. controls (330.98,290) and (320.63,273) .. (320.63,252.03) -- cycle ;
%Flowchart: Connector [id:dp4656214147302222] 
\draw  [dash pattern={on 4.5pt off 4.5pt}] (186,260) .. controls (186,226.86) and (212.86,200) .. (246,200) .. controls (279.14,200) and (306,226.86) .. (306,260) .. controls (306,293.14) and (279.14,320) .. (246,320) .. controls (212.86,320) and (186,293.14) .. (186,260) -- cycle ;
%Straight Lines [id:da8371817633519523] 
\draw    (103.75,253.75) -- (54.5,270.5) ;

%Flowchart: Connector [id:dp09063370482384303] 
\draw  [color={rgb, 255:red, 0; green, 0; blue, 0 }  ,draw opacity=1 ][fill={rgb, 255:red, 255; green, 255; blue, 255 }  ,fill opacity=1 ] (52,270.5) .. controls (52,269.12) and (53.12,268) .. (54.5,268) .. controls (55.88,268) and (57,269.12) .. (57,270.5) .. controls (57,271.88) and (55.88,273) .. (54.5,273) .. controls (53.12,273) and (52,271.88) .. (52,270.5) -- cycle ;
%Flowchart: Connector [id:dp18338959252268783] 
\draw  [color={rgb, 255:red, 128; green, 128; blue, 128 }  ,draw opacity=1 ][fill={rgb, 255:red, 128; green, 128; blue, 128 }  ,fill opacity=1 ] (90,253.75) .. controls (90,246.16) and (96.16,240) .. (103.75,240) .. controls (111.34,240) and (117.5,246.16) .. (117.5,253.75) .. controls (117.5,261.34) and (111.34,267.5) .. (103.75,267.5) .. controls (96.16,267.5) and (90,261.34) .. (90,253.75) -- cycle ;
%Flowchart: Connector [id:dp6722211128113778] 
\draw  [color={rgb, 255:red, 0; green, 0; blue, 0 }  ,draw opacity=1 ][fill={rgb, 255:red, 255; green, 255; blue, 255 }  ,fill opacity=1 ] (112,257.5) .. controls (112,256.12) and (113.12,255) .. (114.5,255) .. controls (115.88,255) and (117,256.12) .. (117,257.5) .. controls (117,258.88) and (115.88,260) .. (114.5,260) .. controls (113.12,260) and (112,258.88) .. (112,257.5) -- cycle ;
%Flowchart: Connector [id:dp7095053476587327] 
\draw  [color={rgb, 255:red, 0; green, 0; blue, 0 }  ,draw opacity=1 ][fill={rgb, 255:red, 255; green, 255; blue, 255 }  ,fill opacity=1 ] (232,257.5) .. controls (232,256.12) and (233.12,255) .. (234.5,255) .. controls (235.88,255) and (237,256.12) .. (237,257.5) .. controls (237,258.88) and (235.88,260) .. (234.5,260) .. controls (233.12,260) and (232,258.88) .. (232,257.5) -- cycle ;
%Flowchart: Connector [id:dp797364043820336] 
\draw  [color={rgb, 255:red, 0; green, 0; blue, 0 }  ,draw opacity=1 ][fill={rgb, 255:red, 255; green, 255; blue, 255 }  ,fill opacity=1 ] (292,270.5) .. controls (292,269.12) and (293.12,268) .. (294.5,268) .. controls (295.88,268) and (297,269.12) .. (297,270.5) .. controls (297,271.88) and (295.88,273) .. (294.5,273) .. controls (293.12,273) and (292,271.88) .. (292,270.5) -- cycle ;
%Flowchart: Connector [id:dp3220649012901804] 
\draw  [color={rgb, 255:red, 0; green, 0; blue, 0 }  ,draw opacity=1 ][fill={rgb, 255:red, 255; green, 255; blue, 255 }  ,fill opacity=1 ] (352,256.5) .. controls (352,255.12) and (353.12,254) .. (354.5,254) .. controls (355.88,254) and (357,255.12) .. (357,256.5) .. controls (357,257.88) and (355.88,259) .. (354.5,259) .. controls (353.12,259) and (352,257.88) .. (352,256.5) -- cycle ;

% Text Node
\draw (344,230) node [scale=0.9,color={rgb, 255:red, 128; green, 128; blue, 128 }  ,opacity=1 ] [align=left] {$\displaystyle D_{\ell}$};
% Text Node
\draw (286,250) node [scale=0.9,color={rgb, 255:red, 128; green, 128; blue, 128 }  ,opacity=1 ] [align=left] {$\displaystyle D_{\ell-1}$};
% Text Node
\draw (226,230) node [scale=0.9,color={rgb, 255:red, 128; green, 128; blue, 128 }  ,opacity=1 ] [align=left] {$\displaystyle D_{\ell-2}$};
% Text Node
\draw (427.25,273.75) node [scale=1.44] [align=left] {$\displaystyle \hatt_\subminus$};
% Text Node
\draw (104.5,230) node [scale=0.9,color={rgb, 255:red, 128; green, 128; blue, 128 }  ,opacity=1 ] [align=left] {$\displaystyle D_{2}$};
% Text Node
\draw (46,250) node [scale=0.9,color={rgb, 255:red, 128; green, 128; blue, 128 }  ,opacity=1 ] [align=left] {$\displaystyle D_{1}$};
% Text Node
\draw (426.5,238) node [scale=0.9,color={rgb, 255:red, 128; green, 128; blue, 128 }  ,opacity=1 ] [align=left] {$\displaystyle D_{\ell+1}$};
% Text Node
\draw (43.5,320) node [scale=0.9] [align=left] {$\displaystyle B_{1} =\obj$};
% Text Node
\draw (105,300) node [scale=0.9] [align=left] {$\displaystyle B_{2}$};
% Text Node
\draw (224.5,300) node [scale=0.9] [align=left] {$\displaystyle B_{\ell-2}$};
% Text Node
\draw (344.5,300) node [scale=0.9] [align=left] {$\displaystyle B_{\ell}$};
% Text Node
\draw (247.5,330) node [scale=0.9] [align=left] {$\displaystyle B_{\ell-1}$};
% Text Node
\draw (45.5,272.5) node [scale=0.8] [align=left] {$\displaystyle v_{1}$};
% Text Node
\draw (104.75,251.75) node [scale=0.8] [align=left] {$\displaystyle v_{2}$};
% Text Node
\draw (221.5,250.5) node [scale=0.8] [align=left] {$\displaystyle v_{\ell-2}$};
% Text Node
\draw (281.75,271.75) node [scale=0.8] [align=left] {$\displaystyle v_{\ell-1}$};
% Text Node
\draw (344.75,250.75) node [scale=0.8] [align=left] {$\displaystyle v_{\ell}$};
% Text Node
\draw (225.5,176.5) node [scale=1.44] [align=left] {$\displaystyle \hatt_\subplus$};

\end{tikzpicture}}
	\caption{A subset path $D_1,\dots,D_{l+1}$ when $\obj$ is a subset.}
	\label{subsec:PV:fig:compute-path_subset_case2}
\end{figure}
\smallskip

Next lemma states that, under certain conditions, there exists a subset path, which can be computed in polynomial time.

\begin{mylemma}
	\label{subsec:PV:lemma:compute-path}
	If $H$ is $\calb$-connected and pruned with~${\{ B \in \calb : v \notin B \}}$ for some ${v \in V(H)}$, then a subset path of $H$ processed with~$\calb$ can be computed in polynomial time.
\end{mylemma}

\begin{proof}
	If $H$ is already pruned with $\calb$, then $V(H)$ is a subset path of $H$ processed with $\calb$.
	Thus, we can assume that executing $\PP(H,\calb)$ processes $\ell \ge 1$ iterations.
	Let~$H_i$ be the graph being pruned at the beginning of the $i$-th iteration, such that $H_1 = H$ and $H_{\ell+1}$ is the graph returned by the algorithm.
	Also, let $B_i \in \calb$ be the subset chosen to be deleted at this iteration.
	In each iteration of this execution, we choose some subset~$B_i \in \calb$ with $|\delta_{H_i}(B_i)| = 1$ which is inclusion-wise minimal.
	% In each iteration of this execution, the subset~$B_i$ is selected to be some inclusion-wise minimal subset $B_i \in \calb$ with $|\delta_{H}(B_i)| = 1$.
	%
	Let ${D_{\ell+1} = V(H_{\ell+1})}$ and ${D_i = V(H_i) \cap B_i}$ for each ${1 \le i \le \ell}$.
	Note that each $D_i$ is connected by invariant~\ref{subsec:PP:lemma:PP_inv:2}.
	We will show that the sequence $D_1, D_2, \dots, D_{\ell+1}$ is a subset path of $H$ processed with $\calb$.

	For each $1 \le i \le \ell$, by the choice of $B_i$, there is an edge connecting a vertex $v_i \in D_i$ to some vertex  in $V(H_{i+1})$.
	We claim that, if ${|\delta_{H_i}(B)| = 1}$ for some $B \in \calb$, then $v_i \in B$.
	Since $H_1$ is pruned with ${\{ B \in \calb : v \notin B \}}$, any subset $B\in \calb$ with ${|\delta_{H_1}(B)| = 1}$ contains $v$. Then $B_1\subseteq B$, because $B_1$ is a inclusion-wise minimal subset with ${|\delta_{H_1}(B)| = 1}$. Thus, indeed, $v_1 \in B$.

	Now suppose that the claim is false, and let $i \ge 2$ be minimum such that ${|\delta_{H_i}(B)| = 1}$ for some ${B \in \calb}$ with $v_i \notin B$.
	If $B \cap B_i \ne \emptyset$, then either $B\subseteq B_i$ or $B_i \subseteq B$.
	But $B_i$ is an inclusion-wise minimal subset with degree one on~$H_i$, thus $B_i \subseteq B$, which contradicts $v_i \notin B$.
	Thus, assume that $B$~and~$B_i$ are disjoint.

	Let ${D = B \cap V(H_i)}$ and $R = V(H_i) \setminus (B \cup B_i)$, and observe that the three sets \mbox{$D_i, D, R$} partition the vertices of~$H_i$.
	Because both $D_i$ and $D$ have degree one on $H_i$, there is no edge between them, as otherwise $H_i$ would be disconnected.
	Thus, there is one edge from~$D_i$~to~$R$, and one edge from~$D$~to~$R$.

	If the unique edge leaving $D_{i-1}$ in $H_{i-1}$ connects $D_{i-1}$ to $R$, then~$D_i$ and~$D$ would have degree one on $H_{i-1}$,
	Because $B_i$ and $B$ are disjoint, and $\calb$ is laminar, only one of them may have a vertex of $D_{i-1}$.
	This implies that the other subset does not have $v_{i-1}$ and has degree one on $H_{i-1}$.
	This is a contradiction, since the claim holds for $i-1$.
	Thus, assume there is no edge between $D_{i-1}$ and $R$.

	Then, the unique edge leaving $D_{i-1}$ in $H_{i-1}$ has an extreme in either $D_i$ or $D$. Assume that this edge connects $D_{i-1}$ to $D_i$, as the other case is symmetrical.
	In this case, there is no edge between $D$ and $D_{i-1}$, and thus the degree of $B$ on $H_{i-1}$ is one.
	This implies that $v_{i-1} \in B$, and it follows that $B_{i-1} \cap B \ne \emptyset$. Therefore, $B_{i-1} \subseteq B$.

	Notice that $H$ is $B$-connected.
	Thus, we have two distinct edges between two connected graphs, $H[B]$ and $H[D_i \cup R]$, forming a cycle in $H$, say $C$.
	Since $C$ has vertices in both $D_{i-1}$ and $H_i$, at some iteration $j < i$, some, but not all vertices of $C$ were removed for the first time.
	But this implies that $|\delta_{H_j}(B_j)| \ge 2$, which is a contradiction to the choice of $B_j$.
	Hence, the claim~holds.

	We conclude that, for each ${1\le i \le \ell}$, $H_i$ is pruned with~${\{B \in \calb : v_i \notin B\}}$, and there is an edge from $v_i \in D_i$ to $D_{i+1}$. Then, $D_1, D_2, \dots, D_{\ell+1}$ is indeed a subset path of $H$ processed with $\calb$. This completes the lemma.
\end{proof}

\subsubsection{Finding a subset path}

In the remaining of this section, we consider ${H = T_\subminus \cup T_\subplus}$, i.e., we let $H$ be the graph with  with vertex set $V(T_\subminus) \cup V(T_\subplus)$ and edge set $E(T_\subminus) \cup E(T_\subplus)$, and let $\hath$ be the graph output by $\PP(H, \calb_\subplus)$.
By the monotonicity of the pruning operation, $\hath$ contains $\hatt_\subplus$, then $\hath$ spans at least $k$ vertices.
The goal is to find a subset path of $\hath$ leading to $\hatt_\subminus$.

Observe that, when $\obj$ is a subset, $\hath = \hatt_\subplus$, then  $\hath$ is pruned with ${\calb_\subminus \setminus \{\obj\}}$. If $v$ is a vertex in $\obj$, then $\hath$ is pruned with ${\{B \in \calb_\subminus : v \notin B \}}$. Therefore, Lemma~\ref{subsec:PV:lemma:compute-path} directly implies the following.

\begin{mylemma}
	\label{subsec:PV:lemma:compute-path_subset_case}
	If $\obj$ is a subset, then a subset path of $\hath$ processed with $\calb_\subminus$ can be computed in polynomial time.
\end{mylemma}

Now, assume that $\obj$ is an edge.
For this case, Lemma~\ref{subsec:threshold-tuple:lemma:main_lemma} implies that $\calb_\subplus = \calb_\subminus$. Moreover, $T_\subplus$ and $T_\subminus$ span the same set of vertices, and the difference between them is exactly one edge, thus there exists a unique edge $e_\subplus \in E(T_\subplus)$ such that $e_\subplus \notin E(T_\subminus)$.

Unlike before, we cannot obtain a subset path from $\hatt_\subplus$ that leads to~$\hatt_\subminus$, because removing $e_\subplus$ from $\hatt_\subplus$ disconnects the graph.
Instead, we use the fact that~$H$ has exactly one cycle, because $H = T_\subminus + e_\subplus$. Let $C$ be this cycle, and notice that $C$ contains $e_\subplus$.
We begin by showing that $\hath$ also contains the cycle~$C$.

\begin{mylemma}
	If $\obj$ is an edge, then $\hath$ contains the cycle $C$.
\end{mylemma}

\begin{proof}
	First observe that $\hatt_\subplus$ contains $e_\subplus$.
	If not, then we would have $\hatt_\subplus \subseteq T_\subminus$ and, since $\hatt_\subplus$ is pruned with $\calb_\subplus = \calb_\subminus$, this would imply $\hatt_\subplus \subseteq \hatt_\subminus$ by Lemma~\ref{subsec:PP:lemma:structure}.
	This is not possible because $|\hatt_\subminus| < |\hatt_\subplus|$. Thus, indeed $\hatt_\subplus$ contains $e_\subplus$.

	Because $T_\subplus \subseteq H$, we have $\hatt_\subplus \subseteq \hath$ by Lemma~\ref{subsec:PP:lemma:structure}.
	It follows that $\hath$ contains~$e_\subplus$, and thus $\hath$ contains at least some vertices of $C$.
	Suppose that~$\hath$ contains some, but not all vertices of $C$. Then, the pruning algorithm removed some vertices of $C$, and in the first time that this happened, the selected subset would have degree at least $2$, because the corresponding cut-set would have at least two edges from $C$.
	This is not possible, because only subsets with degree one are selected.
	Then, no vertex of $C$ was removed, and indeed $\hath$ contains $C$.
\end{proof}

To construct a subset path of $\hath$ when $\obj$ is an edge, we show that there is an edge $e$ of the cycle, such that pruning $\hath - e$ leads to $\hatt_\subminus$ by a sequence of deleted subsets that form a path.

\begin{mylemma}
	\label{subsec:PV:lemma:compute-path_edge_case}
	If $\obj$ is an edge, then there is an edge $e \in E(C)$ such that a subset path $D_1, D_2, \dots, D_{\ell+1}$ of $\hath - e$ processed with $\calb_\subminus$ can be computed in polynomial time.
	Moreover, $e_\subplus$ connects subsets $D_i$ and $D_{i+1}$ for some $i\ge 1$, and $e$ connects $D_1$ and $D_j$ for some $j \ge i+1$.
\end{mylemma}

\begin{proof}
	Define ${H_1 = \hath - e_\subplus}$, which is connected because $\hath$ contains $C$, and $e_\subplus$ is in $C$.
	Since~$\hath$ is pruned with~$\calb_\subminus$, the first processed subset must contain exactly one extreme of~$e_\subplus$.
	Let~$v$ be such an extreme, and~$v'$ be the other extreme of~$e_\subplus$.

	Now, define ${\calb_1 = \{B \in \calb_\subminus : v' \notin B\}}$, and obtain a graph $\hath_1$ by executing $\PP(H_1, \calb_1)$.
	In Figure~\ref{subsec:PV:fig:compute-path_edge_case1}, the grey circles denotes the sets deleted when pruning~$H_1$.
	Because no subset in $\calb_1$ contains $v'$, and $\hath_1$ is connected, $\hath_1$~contains a path $K_{v'}$ starting in the root~$r$ and ending in~$v'$.
	Consider the unique path $K_v$ in $H_1$ starting in the root~$r$ and ending and~$v$, and let~$u$ be the last vertex in this path which is in $V(\hath_1)$, i.e., $u$ is the last vertex of $K_v$ which was not deleted by the pruning algorithm.
	Note that $u$ must be a vertex of the cycle, since $K_v$ contains at least one vertex of $K_{v'}$ which is also in $C$.

	\begin{figure}[htb!]
		\center
		\hspace*{-10pt}
		\scalebox{1}{\input{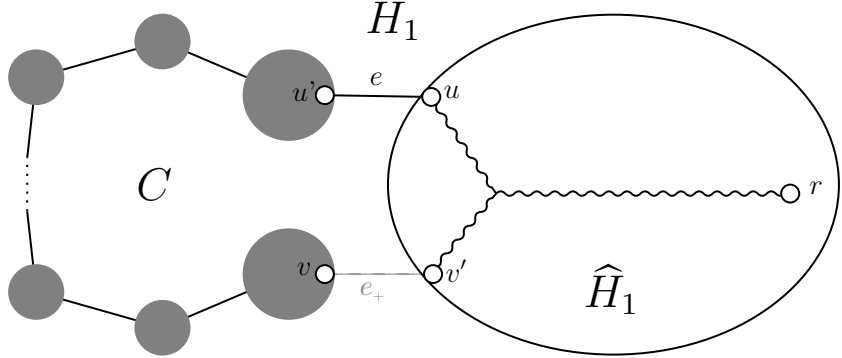}}
		\caption{Pruning $H_1$ using $\calb_1$.}
		\label{subsec:PV:fig:compute-path_edge_case1}
	\end{figure}

	Let $u'$ be the vertex that follows $u$ in $K_{v}$, and observe that $C$ has an edge~$e$ connecting $u$ to $u'$.
	Define ${H_2 = \hath - e}$, and ${\calb_2 = \{B \in \calb_\subminus : u' \notin B\}}$.
	We~claim that $H_2$ is pruned with~$\calb_2$.
	Suppose not, and let $B$ be a subset in~$\calb_2$ with ${|\delta_{H_2}(B)| = 1}$.
	Since $\hath$ is pruned with $\calb$, we know that $B$ contains exactly one extreme of $e$, thus $u \in B$, and $|\delta_{\hath_1}(B)|\ge 1$.
	Since $\hath_1$ does not contain $u'$, it does not contain~$e$ either, thus $\hath_1 \subseteq H_2$.
	This implies $|\delta_{\hath_1}(B)| \le |\delta_{H_2}(B)| = 1$, and thus $|\delta_{\hath_1}(B)| =1$.
	Because $\hath_1$ is pruned with~$\calb_1$, this implies that $B \notin \calb_1$, then $v' \in B$.
	Since $H_2$ has two edge-disjoint paths, one from $v'$ to the root $r$, and other from $v'$ to $u'$, it follows that $|\delta_{H_2}(B)|\ge 2$, which is a contradiction.
	Therefore, $H_2$ is pruned with~$\calb_2$.

	We will obtain a subset path of $H_2$ processed with $\calb_\subminus$ in two steps.
	~First, since $H_2$ is pruned with~$\{B \in \calb_\subminus : u', v' \notin B\} \subseteq \calb_2$, we can use Lemma~\ref{subsec:PV:lemma:compute-path}, and obtain a subset path $D_1, D_2, \dots, D_{i+1}$ of~$H_2$ such that $u' \in D_1$, and $H_2[D_{i+1}]$ is pruned with $\{B \in \calb_\subminus : v' \notin B\}$.
	~Second, using Lemma~\ref{subsec:PV:lemma:compute-path} again, we obtain a subset path $D'_1, D'_2, \dots, D'_{m+1}$ of~$H_2[D_{i+1}]$ such that $v' \in D'_1$, and $H_2[D'_{m+1}]$ is pruned with~$\calb_\subminus$.
	Then $D_1, \dots, D_{i}, D'_1, \dots,D'_m, D'_{m+1}$ is a subset path of $\hath-e$ processed with $\calb_\subminus$.

	Note that $H_2$ has exactly one edge between $D_i$ and $D_{i+1}$, thus this edge must be $e_\subplus$, and since $v' \in D'_1$, $e_\subplus$ must connect $D_m$ to $D'_1$. Also, observe that $e$ connects $D_1$ to some $D'_j$ for $j \ge 1$, forming a cycle with vertices in $D_1, \dots, D_i, D'_1, \dots, D'_j$.
	See Figure~\ref{subsec:PV:fig:compute-path_edge_case2}.
	\begin{figure}[htb!]
		\center
		\hspace*{-10pt}
		\scalebox{1}{\input{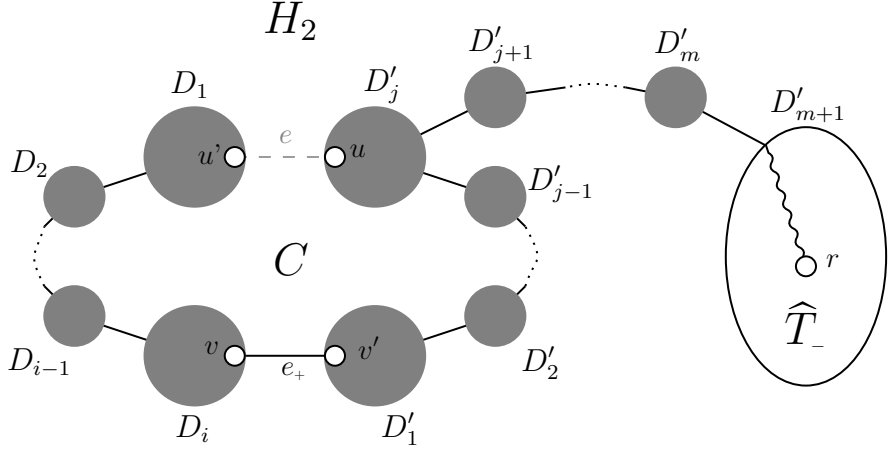}}
		\caption{A subset path $D_1, \dots, D_{i}, D'_1, \dots,D'_m, D'_{m+1}$.}
		\label{subsec:PV:fig:compute-path_edge_case2}
	\end{figure}
\end{proof}

\subsubsection{Augmenting the path}

Assume that we already computed a subset path~$D_1, D_2, \dots, D_{\ell+1}$ of~$\hath$ processed with $\calb_\subminus$.
Recall that $\hath$ contains $\hatt_\subplus$ as subgraph, and thus spans at least~$k$ vertices.
Also, observe that $D_{\ell+1}$ corresponds to the vertices of~$\hatt_\subminus$, and thus spans less than $k$ vertices.
Then, there exists an index~$t$ such that the subsets in $D_t, D_{t+1}, \dots, D_{\ell+1}$ cover at least $k$ vertices, but $D_{t+1}, \dots, D_{\ell+1}$ cover less than~$k$ vertices.

If $D_{t+1}, \dots, D_{\ell+1}$ cover exactly ${k - m}$ vertices, then we would like augment this sequence by iteratively picking subsets from $D_t$ which add up to $m$ vertices.
The goal is to find a sequence of subsets $P_1, P_2, \dots, P_s$ such that: each subset $P_i$ induces a connected subgraph in $\hath$; $P_1$ is connected to $D_{t+1}$; adjacent subsets are connected by an edge; and $|P_s| = 1$.

This can be done as follows.
Suppose that we already have computed a sequence $P_1, P_2, \dots, P_{i-1}$ for some $i \ge 1$, and want to pick $m$ vertices in $S \cap D_t$ for some subset $S \in \call_\subminus$ containing at least $m$ vertices. Also, suppose there is an edge connecting $P_{i-1}$ to some vertex~${v \in S \cap D_t}$.
To initialize the process, let $P_0 = D_{t+1}$ and $S = B_t$, where $B_t$ is the subset in $\calb_\subminus$ corresponding to~$D_t$ in the subset path.
Note that there is an edge connecting $D_{t+1}$ to some~$v \in D_t$.

If $m = 1$, then define $P_i = \{v\}$, and we are done.
Otherwise, we have $m \ge 2$, thus $S$ contains at least two vertices. Since $\call_\subminus$ is binary laminar, this implies that there are disjoint subsets $S_1$ and $S_2$ with $S = S_1 \cup S_2$, and such that $v \in S_1$.

If $|S_1 \cap D_t| \ge m$, then just make $S = S_1$, and repeat the process. This does not change the assumptions, except that it makes $S$ smaller.
Otherwise, ${|S_1 \cap D_t| < m}$, but this implies ${|S_2 \cap D_t| \ge 1}$ because we have ${|S \cap D_t| \ge m}$.
It follows that $\hath$ spans vertices in both $S_1$ and $S_2$, and, since $\hath$ is $S$-connected, there must be an edge connecting a vertex $v_1 \in S_1$ to a vertex $v_2 \in S_2$.
In this case, we define $P_i = S_1 \cap D_t$, update the variables by making $m = m - |P_i|$, $S = S_2$ and $v = v_2$, and repeat the process for $i + 1$. Note that $P_i$ induces a connected subgraph in $\hath$ because $\hath$ is $S_1$-connected.

Figure~\ref{subsec:PV:fig:second-phase} exemplifies this process.
In this figure, solid contours denote the subsets $S \in \call_\subminus$ considered in the process.

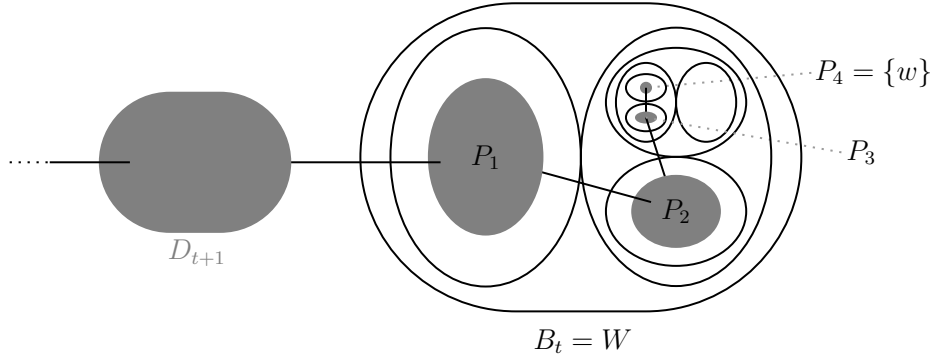
\begin{figure}[!htbp]
	\center
	\scalebox{1}{%!TEX root = ../../main.tex

\tikzset{every picture/.style={line width=0.75pt}} %set default line width to 0.75pt        

\begin{tikzpicture}[x=0.75pt,y=0.75pt,yscale=-1,xscale=1]
%uncomment if require: \path (0,310); %set diagram left start at 0, and has height of 310

%Straight Lines [id:da6779129396282977] 
\draw [color={rgb, 255:red, 155; green, 155; blue, 155 }  ,draw opacity=1 ] [dash pattern={on 0.84pt off 2.51pt}]  (525,125) -- (442.5,132.5) ;

%Flowchart: Connector [id:dp2091452114830752] 
\draw  [color={rgb, 255:red, 128; green, 128; blue, 128 }  ,draw opacity=1 ][fill={rgb, 255:red, 128; green, 128; blue, 128 }  ,fill opacity=1 ] (440,132.5) .. controls (440,131.12) and (441.12,130) .. (442.5,130) .. controls (443.88,130) and (445,131.12) .. (445,132.5) .. controls (445,133.88) and (443.88,135) .. (442.5,135) .. controls (441.12,135) and (440,133.88) .. (440,132.5) -- cycle ;
%Flowchart: Connector [id:dp9771653301772656] 
\draw   (432.5,132.5) .. controls (432.5,128.7) and (436.98,125.63) .. (442.5,125.63) .. controls (448.02,125.63) and (452.5,128.7) .. (452.5,132.5) .. controls (452.5,136.3) and (448.02,139.38) .. (442.5,139.38) .. controls (436.98,139.38) and (432.5,136.3) .. (432.5,132.5) -- cycle ;
%Flowchart: Connector [id:dp307299329982226] 
\draw   (432.5,147.5) .. controls (432.5,143.7) and (436.98,140.63) .. (442.5,140.63) .. controls (448.02,140.63) and (452.5,143.7) .. (452.5,147.5) .. controls (452.5,151.3) and (448.02,154.38) .. (442.5,154.38) .. controls (436.98,154.38) and (432.5,151.3) .. (432.5,147.5) -- cycle ;
%Straight Lines [id:da32937302746510855] 
\draw    (442.5,132.5) -- (442.5,147.5) ;

%Flowchart: Connector [id:dp07630132130755807] 
\draw  [color={rgb, 255:red, 128; green, 128; blue, 128 }  ,draw opacity=1 ][fill={rgb, 255:red, 128; green, 128; blue, 128 }  ,fill opacity=1 ] (437.5,147.5) .. controls (437.5,146.12) and (439.74,145) .. (442.5,145) .. controls (445.26,145) and (447.5,146.12) .. (447.5,147.5) .. controls (447.5,148.88) and (445.26,150) .. (442.5,150) .. controls (439.74,150) and (437.5,148.88) .. (437.5,147.5) -- cycle ;
%Flowchart: Connector [id:dp7347104145673966] 
\draw   (427.5,139.84) .. controls (427.5,128.88) and (434.22,120) .. (442.5,120) .. controls (450.78,120) and (457.5,128.88) .. (457.5,139.84) .. controls (457.5,150.8) and (450.78,159.69) .. (442.5,159.69) .. controls (434.22,159.69) and (427.5,150.8) .. (427.5,139.84) -- cycle ;
%Flowchart: Connector [id:dp6456304663382708] 
\draw   (457.5,139.84) .. controls (457.5,128.88) and (464.22,120) .. (472.5,120) .. controls (480.78,120) and (487.5,128.88) .. (487.5,139.84) .. controls (487.5,150.8) and (480.78,159.69) .. (472.5,159.69) .. controls (464.22,159.69) and (457.5,150.8) .. (457.5,139.84) -- cycle ;
%Flowchart: Connector [id:dp6313266875548216] 
\draw   (422.5,139.84) .. controls (422.5,124.7) and (438.17,112.42) .. (457.5,112.42) .. controls (476.83,112.42) and (492.5,124.7) .. (492.5,139.84) .. controls (492.5,154.99) and (476.83,167.27) .. (457.5,167.27) .. controls (438.17,167.27) and (422.5,154.99) .. (422.5,139.84) -- cycle ;
%Flowchart: Connector [id:dp047253759680033314] 
\draw   (422.5,194.69) .. controls (422.5,179.54) and (438.17,167.27) .. (457.5,167.27) .. controls (476.83,167.27) and (492.5,179.54) .. (492.5,194.69) .. controls (492.5,209.83) and (476.83,222.11) .. (457.5,222.11) .. controls (438.17,222.11) and (422.5,209.83) .. (422.5,194.69) -- cycle ;
%Straight Lines [id:da7355783660451489] 
\draw    (442.5,147.5) -- (457.5,194.69) ;

%Flowchart: Connector [id:dp4117868341037694] 
\draw  [color={rgb, 255:red, 128; green, 128; blue, 128 }  ,draw opacity=1 ][fill={rgb, 255:red, 128; green, 128; blue, 128 }  ,fill opacity=1 ] (435.63,194.69) .. controls (435.63,184.82) and (445.42,176.82) .. (457.5,176.82) .. controls (469.58,176.82) and (479.38,184.82) .. (479.38,194.69) .. controls (479.38,204.56) and (469.58,212.56) .. (457.5,212.56) .. controls (445.42,212.56) and (435.63,204.56) .. (435.63,194.69) -- cycle ;
%Straight Lines [id:da3313896120439974] 
\draw    (362.5,167.27) -- (445,190) ;

%Flowchart: Connector [id:dp19211237560536287] 
\draw  [color={rgb, 255:red, 128; green, 128; blue, 128 }  ,draw opacity=1 ][fill={rgb, 255:red, 128; green, 128; blue, 128 }  ,fill opacity=1 ] (334.17,167.27) .. controls (334.17,145.68) and (346.85,128.17) .. (362.5,128.17) .. controls (378.15,128.17) and (390.83,145.68) .. (390.83,167.27) .. controls (390.83,188.86) and (378.15,206.36) .. (362.5,206.36) .. controls (346.85,206.36) and (334.17,188.86) .. (334.17,167.27) -- cycle ;
%Rounded Rect [id:dp8689284953631191] 
\draw  [color={rgb, 255:red, 0; green, 0; blue, 0 }  ,draw opacity=1 ] (300,167.5) .. controls (300,124.7) and (334.7,90) .. (377.5,90) -- (442.5,90) .. controls (485.3,90) and (520,124.7) .. (520,167.5) -- (520,167.5) .. controls (520,210.3) and (485.3,245) .. (442.5,245) -- (377.5,245) .. controls (334.7,245) and (300,210.3) .. (300,167.5) -- cycle ;
%Straight Lines [id:da7484679869790065] 
\draw [color={rgb, 255:red, 155; green, 155; blue, 155 }  ,draw opacity=1 ] [dash pattern={on 0.84pt off 2.51pt}]  (540,165) -- (442.5,147.5) ;

%Straight Lines [id:da5331296059398574] 
\draw    (245,170) -- (340,170) ;

%Rounded Rect [id:dp6746997732955968] 
\draw  [color={rgb, 255:red, 128; green, 128; blue, 128 }  ,draw opacity=1 ][fill={rgb, 255:red, 128; green, 128; blue, 128 }  ,fill opacity=1 ] (170,170) .. controls (170,150.67) and (185.67,135) .. (205,135) -- (230,135) .. controls (249.33,135) and (265,150.67) .. (265,170) -- (265,170) .. controls (265,189.33) and (249.33,205) .. (230,205) -- (205,205) .. controls (185.67,205) and (170,189.33) .. (170,170) -- cycle ;
%Straight Lines [id:da5507696160379671] 
\draw    (145,170) -- (185,170) ;

%Straight Lines [id:da26822830616550686] 
\draw  [dash pattern={on 0.84pt off 2.51pt}]  (125,170) -- (145,170) ;

%Flowchart: Connector [id:dp4946646323779331] 
\draw   (410,167.27) .. controls (410,131.37) and (431.27,102.27) .. (457.5,102.27) .. controls (483.73,102.27) and (505,131.37) .. (505,167.27) .. controls (505,203.16) and (483.73,232.27) .. (457.5,232.27) .. controls (431.27,232.27) and (410,203.16) .. (410,167.27) -- cycle ;
%Flowchart: Connector [id:dp2532986166851847] 
\draw   (315,167.5) .. controls (315,131.6) and (336.27,102.5) .. (362.5,102.5) .. controls (388.73,102.5) and (410,131.6) .. (410,167.5) .. controls (410,203.4) and (388.73,232.5) .. (362.5,232.5) .. controls (336.27,232.5) and (315,203.4) .. (315,167.5) -- cycle ;

% Text Node
\draw (411,260) node [scale=0.9] [align=left] {$\displaystyle B_{t} =W$};
% Text Node
\draw (362.5,167.27) node [scale=0.9,color={rgb, 255:red, 0; green, 0; blue, 0 }  ,opacity=1 ] [align=left] {$\displaystyle P_{1}$};
% Text Node
\draw (457.5,194.69) node [scale=0.9,color={rgb, 255:red, 0; green, 0; blue, 0 }  ,opacity=1 ] [align=left] {$\displaystyle P_{2}$};
% Text Node
\draw (556,125) node [scale=0.9,color={rgb, 255:red, 0; green, 0; blue, 0 }  ,opacity=1 ] [align=left] {$\displaystyle P_{4} =\{w\}$};
% Text Node
\draw (550,165) node [scale=0.9,color={rgb, 255:red, 0; green, 0; blue, 0 }  ,opacity=1 ] [align=left] {$\displaystyle P_{3}$};
% Text Node
\draw (218.5,215.5) node [scale=0.9,color={rgb, 255:red, 128; green, 128; blue, 128 }  ,opacity=1 ] [align=left] {$\displaystyle D_{t+1}$};

\end{tikzpicture}}
	\caption{A sequence $\dots, D_{t+1}, P_1, P_2, P_3, P_4$.}
	\label{subsec:PV:fig:second-phase}
\end{figure}

The whole process is summarized as the \emph{picking-vertices} algorithm, denoted by $\PV(\pot,\tblist)$. A listing is given in Algorithm~\ref{subsec:PV:alg:PV}.
First, obtain a subset path $D_1, D_2, \dots, D_{\ell+1}$ using Lemmas~\ref{subsec:PV:lemma:compute-path_subset_case} or \ref{subsec:PV:lemma:compute-path_edge_case}, depending on whether $\obj$ is a subset, or an edge.
Then, find the largest $t$ such that $D_t, D_{t+1}, \dots, D_{\ell+1}$ cover at least $k$ vertices, and find a sequence $P_1, P_2, \dots, P_s$ from $D_t$ using the picking process.

Now, the sequence $P_s, \dots, P_1, D_{t+1}, \dots, D_{\ell+1}$ cover exactly~$k$ vertices, and is such that each subset induces a connected subgraph in $\hath$, and adjacent subsets are connected by an edge.
Thus, it induces a tree $T$ spanning exactly~$k$ vertices and containing the root~$r$, which is the output of the algorithm.

\begin{algorithm}[p]
	\DontPrintSemicolon
	\caption{$\PV(\pot,\tblist)$}
	\label{subsec:PV:alg:PV}

	\uIf{\textnormal{$\GW(\pot,\tblist)$ returns at least $k$ vertices}} {
		Let $(\tblist_\subplus, \tblist_\subminus) \gets (\tblist, \sublist)$
	}
	\Else{
		Let $(\tblist_\subplus, \tblist_\subminus) \gets (\sublist, \tblist)$
	}

	Let $(T_\subminus,\calb_\subminus) \gets \GP(\pot,\tblist_\subminus)$
	and $\call_\subminus$ be the computed laminar collection\;

	Let $(T_\subplus,\calb_\subplus) \gets \GP(\pot,\tblist_\subplus)$\;

	Let $H \gets T_\subminus \cup T_\subplus$ and $\hath \gets \PP(H,\calb_\subplus)$\;

	\nonl\;

	\tcp*[h]{finding a subset path}\;
	\uIf{\textnormal{$\tblist_{|\tblist|}$ is a subset}}{
		Obtain sequence $D_1,D_2,\dots,D_\ell$ using the algorithm of Lemma~\ref{subsec:PV:lemma:compute-path_subset_case}
	}
	\Else
	{
		Obtain sequence $D_1,D_2,\dots,D_\ell$ using the algorithm of Lemma~\ref{subsec:PV:lemma:compute-path_edge_case}
	}
	Find largest $t$ such that $D_t,D_{t+1}, \dots,D_{\ell+1}$ cover at least $k$ vertices

	\nonl\;
	\tcp*[h]{picking-vertices}\;
	Let $P_0 \gets D_{t+1}$, $i \gets 1$ and $m \gets k - \sum_{i=t+1}^{\ell+1} |D_i|$\;
	Find a neighbor $v$ of $P_0$ in $D_t$\;
	Find $B_t \in \calb_\subminus$ such that $D_t \subseteq B_t$ and let $S \gets B_t$\;
	\While{\textnormal{$m \ne 0$}}
	{
		\uIf{\textnormal{$m = 1$}}
		{
			Define $P_i \gets \{v\}$ and $s \gets i$\;
			Let $m \gets 0$ and stop the loop\;
		}
		\Else{
			Find $S_1, S_2 \in \call_\subminus$ with $S = S_1 \cup S_2$ and $v \in S_1$\;
			\uIf{\textnormal{$|S_1 \cap D_t| \ge m$}}
			{
				Let $S \gets S_1$\;
			}\Else{
				Find an edge with extremes $v_1 \in S_1$ and $v_2 \in S_2$\;
				Let $P_i \gets S_1 \cap D_t$, $m \gets m - |P_i|$\;
				Let $S \gets S_2$ and $v \gets v_2$\;
				Let $i \gets i + 1$\;
			}
		}
	}
	\nonl\;
	Let $T$ be the tree induced by sequence $P_s, \dots, P_1, D_{t+1}, \dots, D_{\ell+1}$\;
	\Return{T}
\end{algorithm}

In what follows, denote by $W$ the subset $B_t \in \calb_\subminus$ corresponding to $D_t$ in the subset path such that $D_t \subseteq W$ and $W \cap D_j = \emptyset$ for every $j >  t$. Recall that $\call[W]$ is the collection of subsets in $\call$ which are subsets of $W$.
Also, let~$w$ be the last vertex selected by the picking process, i.e., $P_s = \{w\}$.
Next lemmas present some properties of vertex~$w$ and laminar collections~$\call$~and~$\calb$.

\begin{mylemma}
	\label{subsec:PV:lemma:w_in_L}
	Let $L \in \call_\subminus[W]$ and suppose $L$ contains a vertex in $V(T)$ and a vertex $v \in W \setminus V(T)$.
	If no subset in $\calb_\subminus[W \setminus V(T)]$ contains $v$, then ${w \in L}$.
\end{mylemma}

\begin{proof}
	Let $D_1, \dots, D_{\ell+1}$ be the computed subset path and recall that~$T$ is a subgraph of ${H_t = \hath[D_t \cup \dots \cup D_{\ell+1}]}$.
	Since no subset in $\calb_\subminus[W \setminus V(T)]$ contains $v$, the pruning algorithm has not deleted~$v$, and thus $v \in V(H_t)$.
	Because $L\subseteq W$, we have $v \in D_t$.
	This implies that~$L$ was considered by the picking process, and because not all vertices in $L \cap D_t$ were selected, all the remaining vertices of $T$ were picked from~$L$. Therefore, $w \in L$.
\end{proof}

Next lemma shows that the only processed subsets in $\calb_\subminus$ with degree one on $T$ which are not pruned from $T$ are those containing~$w$.

\begin{mylemma}
	\label{subsec:PV:lemma:B_has_deg_1_then_B_contains_w}
	Let $B \in \calb_\subminus$.
	If $|\delta_T(B)| = 1$, then $w \in B$.
\end{mylemma}

\begin{proof}
	Remember that the algorithm computes a subset path $D_1, \dots, D_{\ell+1}$ of a tree $\hath'$ processed with $\calb_\subminus$, where $\hath' = \hath$ if $\obj$ is a subset, and $\hath' = \hath - e$ for some edge $e$ if $\obj$ is an edge.
	Define $H_t = \hath'[D_t \cup \dots \cup D_{\ell+1}]$.
	Let $K$ be the path in $T$ starting in the root~$r$ and ending in $w$, and observe that there is a vertex $v_t \in V(K)$ such that $H_t$ is pruned with ${\{B' \in \calb_\subminus : v_t \notin B' \}}$.

	Let $B \in \calb_\subminus$ with $|\delta_T(B)| = 1$, and
	suppose for a contradiction that $w \notin B$.
	If $|\delta_T(B)| = |\delta_{H_t}(B)|$, then $v_t \in B$, and there are two edge-disjoint  paths crossing $B$, one from $w$ to $v_t$ and other from $v_t$ to the root~$r$. This is not possible because $|\delta_T(B)| = 1$. Therefore,
	$|\delta_T(B)| \ne |\delta_{H_t}(B)|$, and thus
	$|\delta_{H_t}(B)|\ge 2$.

	Observe that $K$ contains all edges connecting adjacent subset $D_i$~and~$D_{i+1}$, for~$i \ge 1$. It follows that $B \subseteq D_j$ for some $j \ge t$,
	because $B$ does not contain a vertex of $K$.
	If $j \ge t + 1$, then we would have ${|\delta_{H_t}(B)| = |\delta_T(B)|}$, a~contradiction. Thus, assume that $j = t$ and $B \subseteq D_t$.

	Because $T$ spans vertices in $B$ and $B \subseteq D_t$, theses vertices were selected by the picking process.
	Let $m$ the number of vertices still needed at the iteration which considered subset~$B$.
	If $m \le |B\cap D_{t}|$, then all the $m$ remaining vertices of $T$ were picked from $B\cap D_{t}$, and thus $w \in B$.
	Otherwise, all the vertices of $B \cap D_{t}$ were picked, and
	$T$ has two edges leaving $B$, which is a contradiction since $|\delta_T(B)| = 1$.
\end{proof}

Next lemma extracts the critical property of the tree output by the algorithm, which is used to bound the solution cost.
For some collection of subsets~$\call$, denote by $\call_w$ the collection of subsets in $\call$ which contain vertex~$w$.
Thus, $\call_w[W]$ contains subsets in $\call$ which are subsets of $W$ and contain~$w$.

\begin{mylemma}
	\label{subsec:bound:lemma:average_degree}
	Consider the execution of $\GP$ which returned $(T_\subminus,\calb_\subminus)$, and let~$\cala$ be the collection of active subsets that contain some vertex of $T$ at the beginning of an iteration.
	Then,
	\begin{align*}
		\sum_{A \in \cala}|\delta_T(A)| \leq 2 (|\cala| - |\cala_w[W]|).
	\end{align*}
\end{mylemma}

\begin{proof}
	Let $D_1, D_2, \dots, D_{\ell+1}$ be the subset path computed by the algorithm, and define $H_i = \hath[D_i \cup D_{i+1} \cup \dots D_{\ell+1}]$ for each $1 \le i \le \ell$.
	Observe that $T$ is a subgraph of $H_t$, and remember that $W$ is the subset in~$\calb_\subminus$ associated with~$D_t$, such that $H_{t+1} = H_t - W$ and $|\delta_{H_t}(W)| = 1$.
	Also, let $P_1, P_2, \dots, P_s$ be the sequence of subsets picked from $W$, such that $P_s = \{w\}$.
	By Lemma~\ref{subsec:PV:lemma:B_has_deg_1_then_B_contains_w}, any subset in $\calb_\subminus$ with degree one on $T$ contains $w$.

	Consider the laminar collection $\call_\subminus$ at the beginning of an iteration of $\GP$,
	and let $\cala$ be the collection of maximal subsets in $\callminusstar$ which are active and contain some vertex of~$T$,
	and $\cali$ be the collection of maximal subsets in $\callminusstar$ which are not active and contain some vertex of~$T$.
	Note that $\cala \cup \cali$ partitions $V(T)$.
	If we contract on $T$ each subset in $\cala \cup \cali$, obtaining a graph $T'$, then each vertex of~$T'$ corresponding to a subset~$S$ will have degree~$|\delta_T(S)|$.

	Observe that an active subset is a maximal subset of the laminar collection at the beginning of the iteration, thus the subsets in $\cala$ are disjoint, and then~${|\cala_w[W]| \le 1}$.
	We are also interested in counting the vertices of $T'$ with degree one which correspond to non-active subsets, thus let $\cali_1$ be the collection of subsets ${S\in \cali}$ with~${|\delta_T(S)| = 1}$.
	Because each subset in $\cali_1$ is in~$\calb_\subminus$ and has degree one on $T$, we know that $w \in S$ for every $S \in \cali_1$, by Lemma~\ref{subsec:PV:lemma:B_has_deg_1_then_B_contains_w}.

	We break the proof into two cases, depending on whether $T'$ is a tree.

	\textbf{Case 1:} Assume that the contracted graph $T'$ is a tree.

	We claim that $|\cali_1| + |\cala_w[W]| \le 1$.
	This means that there is at most one non-active subset with degree one, and, if there is one, then there are no active subsets of $W$ containing $w$.
	Indeed, if $\cali_1$ contains some subset $S$, then
	$w \in S$, and thus no active subset contains $w$.
	Otherwise, $|\cali_1|=0$, and the claim holds because $|\cala_w[W]| \le 1$.

	Since $T'$ is a tree, it follows that
	\begin{align*}
		\sum_{A \in \cala}|\delta_T(A)| + 2|\cali| - |\cali_1|
		 & \le \sum_{A \in \cala}|\delta_T(A)| + \sum_{I \in \cali}|\delta_T(I)|
		=    2(|\cala| + |\cali| - 1)                                            \\
		 & \le 2(|\cala| + |\cali| - |\cali_1| -|\cala_w[W]|),
	\end{align*}
	which implies $\sum_{A \in \cala}|\delta_T(A)| \le  2(|\cala| -|\cala_w[W]|)$, thus the lemma holds in this case.

	\textbf{Case 2:} Assume that the contracted graph $T'$ has a cycle.

	Because $T'$ has a cycle, there exists some maximal subset $L$ in $\cala \cup \cali$ which induces a disconnected subgraph $T[L]$.
	Since $T_\subminus$ is $L$-connected by invariant~\ref{subsec:GP:lemma:GP_inv:edge_in_F_are_tight}, we know that $T_\subminus[L]$ is connected,
	thus $H[V(T)\cup L]$ has a cycle $C$ whose vertices intersect~$L$, and then $\hath[V(T) \cup L]$ contains $C$ as well.
	Observe that this case can occur only when $\obj$ is an edge, and $e_\subplus$ is in $C$.
	Because $T$ contains at most one edge that is not in $T_\subminus$, the contracted tree~$T'$ contains at most one cycle.
	Also, note that $T[L]$ has only two components, since otherwise $H$ would have two cycles.

	Recall that $C$ has the edge $e_\subplus$.
	Suppose, for a contradiction, that both extremes of $e_\subplus$ are contained in the same  maximal subset $S \in \cala \cup \cali$.
	Then, because $T_\subminus$ is $S$-connected, $H[S]$ would contain a cycle which is not $C$. Since there is only one cycle, this is not possible, and we conclude that $e_\subplus$ connects consecutive maximal subsets in the cycle.
	Therefore, $e_\subplus$ is an edge of~$T'$, and at least one extreme of $e_\subplus$ is not in $L$.

	We claim that $W$ contains some but not all vertices of~$C$.
	Observe that $D_t$ spans vertices of $C$, since otherwise $C$ would be contained in $D_1, D_2, \dots, D_{t-1}$, by Lemma~\ref{subsec:PV:lemma:compute-path_edge_case}, and thus $T$ would not span vertices of $C$.
	Since $D_t$ is a subset of $W$, it follows that $W$ spans vertices of~$C$.
	Now, also by Lemma~\ref{subsec:PV:lemma:compute-path_edge_case}, $e_\subplus$ connects consecutive subsets $D_i$ to $D_{i+1}$ for some $i \ge 1$.
	Then $T$ contains vertices of both $D_i$ and $D_{i+1}$, and it follows that $t \le i$, since $D_t$ is the first subset of the subset path containing vertices of~$T$.
	Therefore, $W$ cannot contain all vertices of~$C$, because $D_{i+1} \subseteq V(H_{i+1})$ and $H_{i+1}$ does not contain vertices of~$W$.

	It follows that the cut $\delta_{C}(W)$ contains two edges of $C$.
	Since ${|\delta_T(W)| = 1}$, both edges cannot be edges of $T$, thus $W$ contains some vertex $v_1$ in $L$.
	Let $L_1$ and $L_2$ be the two connected components of~$T[L]$. Suppose without loss of generality that $v_1 \in L_1$ and let $v_2 \in L_2$.
	Because there is no path from $v_1$ to $v_2$ in $T- e_\subplus$, it follows
	that $v_2 \in D_j$ for some $j > t$. Thus, $v_2 \notin B_t = W$.
	Now $v_1 \in L \cap W$, and $v_2 \in L \setminus W$. Since $\call_\subminus$ is laminar, $W \subseteq L$.

	Figure~\ref{subsec:picking:fig:maximal_subset_cycle} illustrates $T$ with maximal subsets in dashed lines, whose contraction resulted in a graph $T'$ with a cycle.
	\begin{figure}[htb!]
		\center
		\scalebox{1}{\input{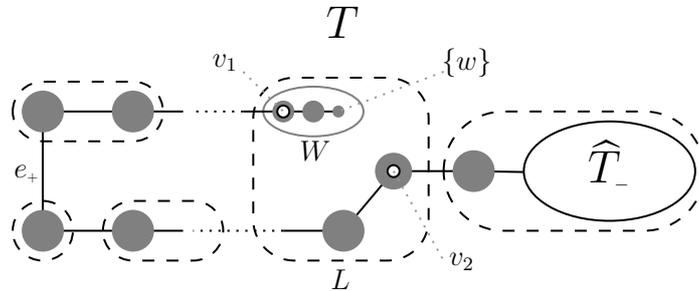}}
		\caption{A graph $T'$ with a cycle.}
		\label{subsec:picking:fig:maximal_subset_cycle}
	\end{figure}

	This implies $|\cala_w[W]| = 0$, since  $\cala_w[W]$ contains only subsets of $W$. Also, $|\cali_1| = 0$, because the maximal subset containing $w$ is $L$, which has degree at least $2$.
	Using the fact that in this case $T'$ has exactly one cycle, we get
	\begin{align*}
		\sum_{A \in \cala}|\delta_T(A)| + 2|\cali|
		 & \le \sum_{A \in \cala}|\delta_T(A)| + \sum_{I \in \cali}|\delta_T(I)| = 2(|\cala| + |\cali|) \\
		 & =  2(|\cala| + |\cali| -|\cala_w[W]|),
	\end{align*}
	which implies $\sum_{A \in \cala}|\delta_T(A)| \le  2(|\cala| -|\cala_w[W]|)$, completing the lemma.
\end{proof}

%!TEX root = main.tex

\section{The \texorpdfstring{$2$}{2}-approximation}
\label{sec:two-approx}

This section wraps up the whole algorithm.
First, we bound the cost of the tree output by $\PV$.
Then, we show how to use this bound to find a tree whose cost is within factor $2$ of an optimal solution for \kpcst.

	%!TEX root = main.tex

\subsection{Bounding the cost of the tree}
\label{subsec:bound}

We adopt the same definitions introduced in Section~\ref{sec:find_solution}, except that, to simplify the notation, for the remainder of this section, we let $\call = \call_\subminus$.
Recall that $\call(S)$ denotes the collection of subsets in $\call$ that contains some, but not all vertices of $S$, that $\call[S]$ denotes the collection of subsets in $\call$ which are subsets of $S$, and that $\call_w$ is the collection of subsets in $\call$ that contain~$w$.
Next two lemmas bound the cost of edges and vertex penalties of the tree output by $\PV$.

\begin{mylemma}
	\label{subsec:bound:lemma:edge_cost_bound}
	Let $T$ be the tree returned by $\PV(\pot,\tblist)$, then
	\begin{align*}
		\ecost_{E(T)} ~\leq~ 2 \smashoperator{\sum_{L \in \call(V(T))}} y_L ~-~ 2 \smashoperator{\sum_{L \in \call_w[W]}} y_L.
	\end{align*}
\end{mylemma}

\begin{proof}
	By Lemma~\ref{subsec:threshold-tuple:lemma:main_lemma}, we have that the vectors $y$ computed in the growth-phases $\GP(\pot,\tblist)$ and $\GP(\pot,\sublist)$ are the same.
	Since~$T \subseteq T_\subminus \cup T_\subplus$, each edge of~$T$ is tight, thus
	\begin{align*}
		\ecost_{E(T)}
			= \smashoperator[r]{\sum_{e \in E(T)}} \ecost_e
			= \sum_{e \in E(T)} \smashoperator[r]{\sum_{L: e \in \delta(L)}} y_L
			= \smashoperator[r]{\sum_{L \in \call(V(T))}} ~|\delta_T(L)| y_L.
	\end{align*}

	We prove by induction that, at the beginning of each iteration of~$\GP(\pot,\tblist)$,
	\begin{align*}
		\smashoperator[r]{\sum_{L \in \call(V(T))}}~ |\delta_T(L)| y_L ~\leq~ 2 \smashoperator{\sum_{L \in \call(V(T))}} y_L ~-~ 2 \smashoperator{\sum_{L \in \call_w[W]}} y_L.
	\end{align*}

	At the beginning of the growth-phase, $y = 0$, then the inequality holds. Thus, assume that the inequality holds at the beginning of an iteration, and let $\cala$ be the collection of active subsets that contain some vertex of $T$.

	Notice that, if some maximal subset $S$ contains $V(T)$, then $S \notin \call(V(T))$.
	Also, any such subset contains both $r$ and $w$, thus $S \notin \call_w[W]$, because~$\call$ is laminar.
	Therefore, if $V(T)$ is contained in some maximal subset, then neither the left nor the right side of the inequality changes.
	Hence, we can assume that no active subset contains $V(T)$.
	It follows that ${\cala \subseteq \call(V(T))}$ and that ${\cala_w[W] \subseteq \call_w[W]}$.
	Suppose that, in this iteration, the  variable~$y_L$ of each active subset $L$ is increased by $\Delta$.
	Then, the left side of the inequality is increased by ${\sum_{A \in \cala}|\delta_T(A)|\Delta}$, while the right side is increased by ${2 (|\cala| - |\cala_w[W]|)  \Delta}$.
	By Lemma~\ref{subsec:bound:lemma:average_degree}, the increase on the left side is smaller, thus the inequality is maintained at the end of the iteration.
\end{proof}

\begin{mylemma}
	\label{subsec:bound:lemma:penalty_cost_bound}
	Let $T$ be the tree returned by $\PV(\pot,\tblist)$, then
	\begin{align*}
		\vpen^\pot_{V \setminus V(T)} ~\leq \smashoperator[r]{\sum_{L \in \call[V \setminus V(T)]}} y_L ~~+ \smashoperator[r]{\sum_{L \in \call_w[W]}} y_L
	\end{align*}
\end{mylemma}

\begin{proof}
	We partition~$W$ into three parts, $P$, $I$, and $Q$, corresponding to the set of vertices whose penalties can be paid by some processed component which is a proper subset of $W \setminus V(T)$, the intersection with $T$, and the vertices which are not paid.
	More precisely, define
	\begin{align*}
		P &= \{v \in W:v \in B \text{ for some } B \in \calb_\subminus[W\setminus V(T)]\},\\
		I &= V(T) \cap W,
		\quad \mbox{and} \quad
		Q = W \setminus (P \cup I).
	\end{align*}

	First, we consider the vertices of $W$ which are not paid. Since $W$ is tight for~$(y,\pot)$ and $y$ respects~$\vpen^\pot$, it follows that
	\begin{align}
	\vpen^\pot_{Q}
	= \vpen^\pot_W - \vpen^\pot_{P \cup I}
	&\leq \sum_{L \in \call[W]} y_L - \sum_{L \in \call[P\cup I]} y_L.
	\label{subsec:bound:lemma:penalty_cost_bound:cost_unpaid}
	\end{align}

	Denote by $\calp$ the collection of subsets in $\call[P\cup I]$ that contain a vertex in~$P$ and a vertex in~$I$.
	Then ${\{\call[P], \call[I], \calp\}}$ partitions the collection~${\call[P\cup I]}$.
	Similarly, denote by $\calw$ the collection of subsets in $\call[W]$ that contain a vertex in ${P \cup Q}$ and a vertex in $I$.
	Then ${\{\call[P \cup Q], \call[I],\calw\}}$ partitions the collection~$\call[W]$.

	We claim that $\calw \setminus \calp \subseteq \call_w[W]$.
	To see this, consider a subset ${L \in \calw \setminus \calp}$.
	Since $L$ contains a vertex in $I$, but is not in $\calp$, it does not contain a vertex in~$P$.
	It follows that that $L$ contains a vertex of $v \in Q$. Because $v \in W \setminus V(T)$ and no subset in $\calb_\subminus[W\setminus V(T)]$ contains~$v$, Lemma~\ref{subsec:PV:lemma:w_in_L} implies that ${w \in L}$. Therefore, $L \in \call_w[W]$, and then ${\calw \setminus \calp \subseteq \call_w[W]}$.

	We can simplify the indices in the summation~\eqref{subsec:bound:lemma:penalty_cost_bound:cost_unpaid} as
	\begin{align*}
	\call[W] \setminus
	\call[P \cup I]
	&=
	(\call[P \cup Q] \cup \call[I] \cup \calw)
	\setminus
	(\call[P] \cup \call[I] \cup \calp)\\
	&\subseteq
	(\call[P \cup Q] \cup \call_w[W]) \setminus \call[P].
	\end{align*}

	Now, we consider the vertices which are paid, i.e., vertices not spanned by~$T$ which are contained in some subset $B \in \calb_\subminus$. In addition to vertices~$P$, the vertices in ${R = V \setminus (V(T) \cup W)}$ are also paid.
	Note that any vertex in~$R$ was deleted by the pruning algorithm, thus every vertex in $R$ is indeed contained in some subset in $\calb_\subminus$
	which contains no vertex in $V(T) \cup W$.
	Hence, there is a collection $\{B_1,B_2,\dots,B_m\}$ of processed subsets in $\calb_\subminus$ that partitions the vertices in~$R$.
	By similar arguments, there is also a collection $\{B_{m+1},B_{m+2},\dots,B_\ell\}$ of subsets in $\calb_\subminus$ that partitions the vertices in~$P$.
	Thus, since each such subset is tight for $(y,\pot)$,
	\begin{align}
		\vpen^\pot_{R \cup P}
		&= \sum_{i=1}^\ell \vpen^\pot_{B_i}
		= \sum_{i=1}^\ell \sum_{L \in \call[B_i]} y_L
		=\sum_{L \in \call[R]\cup \call[P]} y_L. \label{subsec:bound:lemma:penalty_cost_bound:paid}
	\end{align}

	Since the collections of summations~\eqref{subsec:bound:lemma:penalty_cost_bound:cost_unpaid} and~\eqref{subsec:bound:lemma:penalty_cost_bound:paid} are disjoint,
	\begin{align*}
		(\call[W] \setminus
		\call[P \cup I]) \cup
		(\call[R]\cup \call[P])
		&\subseteq \call[P \cup Q] \cup \call_w[W] \cup \call[R]\\
		&\subseteq \call[V \setminus V(T)] \cup \call_w[W].
	\end{align*}

	By combining~\eqref{subsec:bound:lemma:penalty_cost_bound:cost_unpaid} and~\eqref{subsec:bound:lemma:penalty_cost_bound:paid}, we conclude
	\begin{align*}
	\vpen^\pot_{V \setminus V(T)}
		&=\vpen^\pot_{Q} + \vpen^\pot_{R \cup P}
		\le \sum_{L \in \call[V \setminus V(T)]} y_L + \sum_{L \in \call_w[W]} y_L. \qedhere
	\end{align*}
\end{proof}

Combining the bounds from the previous two lemmas, one obtains the following corollary.

\begin{mycorollary}
	\label{subsec:bound:coro:total_cost_T}
	Let $T$ be the tree returned by $\PV(\pot,\tblist)$, then
	\begin{align*}
		\ecost_{E(T)} + 2 \vpen^\pot_{V \setminus V(T)} ~\leq~ 2 \smashoperator{\sum_{L \in \call(V)}} y_L.
	\end{align*}
\end{mycorollary}

	%!TEX root = main.tex

\subsection{The approximation algorithm}
\label{subsec:approx}

Next lemma is the final ingredient of our $2$-approximation.

\begin{mylemma}
	\label{subsec:approx:lemma:2-approx_or_reduce}
	Let $T$ be the tree returned by $\PV(\pot,\tblist)$, and $T^*$ be an optimal solution.
	Also, suppose that $L^*$ is the minimal subset in $\call$ containing $V(T^*)$.
	If $L^* = V$, then
	\[
	\ecost_{E(T)} + 2 \vpen_{V \setminus V(T)} \leq 2 \left( \ecost_{E(T^*)} + \vpen_{V \setminus V(T^*)} \right).
	\]
\end{mylemma}

\begin{proof}
	Assume $V = L^*$.
	Since $T$ spans exactly $k$ vertices and $T^*$ spans at least $k$ vertices, we have ${|V \setminus V(T)| \geq |V \setminus V(T^*)|}$. Therefore,
	\begin{align*}
		\ecost_{E(T)} + 2 \vpen_{V \setminus V(T)}
			&=      \ecost_{E(T)} + 2 \vpen^\pot_{V \setminus V(T)}
			        - 2 \pot|V \setminus V(T)|\\
			&\leq   \ecost_{E(T)} + 2 \vpen^\pot_{V \setminus V(T)}
			        - 2 \pot|V \setminus V(T^*)|\\
			&\leq   \textstyle 2 \sum_{L \in \call(V)} y_L
						  - 2 \pot|V \setminus V(T^*)|\\
			&\leq   2 \left( \ecost_{E(T^*)} + \vpen_{V \setminus V(T^*)} \right),
	\end{align*}
	where the penultimate inequality follows from Corollary~\ref{subsec:bound:coro:total_cost_T}, and the last inequality follows from Lemma~\ref{subsec:approx:lemma:opt_lower_bound}, because $y$ respects~$\ecost$~and~$\vpen^\pot$.
\end{proof}

We finally present our $2$-approximation, which is denoted by $\AP$.
A~listing is given in Algorithm~\ref{subsec:approx:alg:AP}.

The algorithm will compute a series of trees, the best of which
will be the output.
In each iteration, start by computing a tree executing~$\GW(0,\emptyset)$, and, if it already spans at least $k$ vertices, then store this tree and stop the iteration process.
Otherwise, there is a threshold-tuple $(\pot,\tblist)$, by Lemma~\ref{subsec:TS:lemma:TS_inv}, which can be computed by $\TS$.
Now, store the tree returned by $\PV(\pot,\tblist)$ and let $\call$ be the laminar collection used by $\PV(\pot,\tblist)$ when computing this tree.
Find the subset $L_r \in \call$ which is the inclusion-wise maximal proper subset of $V$ containing the root $r$.
If $|L_r| < k$, then stop; otherwise, remove from $G$ any vertex not in $L_r$, and repeat the iteration process with $G[L_r]$.

Observe that at least one vertex is deleted in each iteration, and the reduced graph is connected because it is $\call$-connected.
Thus, the algorithm stops, and the output is the computed tree $T$ which minimizes the cost with respect to the original graph $G$.
We argue why~$T$ is a $2$-approximation.
Let~$T^*$ be an optimal solution with respect to~$G$ and consider the last iteration which processed some subgraph $G'$ containing~$T^*$.
If the algorithm stopped in this iteration after computing a tree using $\GW(0,\emptyset)$, then Lemma~\ref{appendix:bound_simple:lemma:2-approx} implies that this tree is a $2$-approximation with respect to~$G'$.
Otherwise, the inclusion-wise minimal subset containing $T^*$ is $V(G')$, and Lemma~\ref{subsec:approx:lemma:2-approx_or_reduce} implies that the tree computed by $\PV(\pot,\tblist)$ is a $2$-approximation with respect to $G'$.
Now, note that a $2$-approximation with respect to~$G'$ is also a $2$-approximation with respect to~$G$.

\begin{mytheorem}
	\label{subsec:approx:theo:2-approx}
	If $T$ is the tree returned by $\AP$, and $T^*$ is an optimal solution, then
	\begin{align*}
		\ecost_{E(T)} + 2 \vpen_{V \setminus V(T)} \leq 2 \left( \ecost_{E({T^*})} + \vpen_{V \setminus V(T^*)} \right).
	\end{align*}
\end{mytheorem}

\begin{algorithm}[H]
	\DontPrintSemicolon
	\caption{$\AP$}
	\label{subsec:approx:alg:AP}

	Let $V_0 \gets V(G)$

	\While{$|V(G)| \ge k$}
	{
		\If{$\GW(0,\emptyset)$ returns at least $k$ vertices}
		{
			Let $T$ be the tree returned by $\GW(0,\emptyset)$\;
			Stop the loop
		}
		\nonl\;
		Compute threshold-tuple $(\pot,\tblist)$ by executing $\TS$\;
		Let $T$ be the tree returned by $\PV(\pot,\tblist)$\;
		Let $\call$ be the laminar collection computed in $\PV(\pot,\tblist)$\;

		\nonl\;
		Let $L_r \in \call$ be the maximal proper subset of $V$ containing $r$\;

		Update $G \gets G[L_r]$\;

	}

	\Return{the computed tree $T$ which minimizes
	$
		\ecost_{E(T)} + \vpen_{V_0 \setminus V(T)}
	$
	}
\end{algorithm}

%!TEX root = main.tex

\section{Final remarks}
\label{sec:final_remarks}

In this paper, we present a $2$-approximation for the $k$-prize-collecting Steiner tree problem.
This improves over the previous best known approximation factor~\cite{Matsuda2019}, and has a smaller running time.
Observe that $\AP$ iterates for at most $\mathcal{O}(|V|)$ times, and the time complexity of each iteration is dominated by the subroutine $\TS$, whose running time is $\calo(|V||E|^2 + |V|^3 \log^2|V|)$.
Therefore, our algorithm runs in time~${\calo(|V|^2|E|^2 + |V|^4 \log^2|V|)}$.

Variants of the prize-collecting Steiner tree problem are among the most classical network design problems, and have long been studied from both the practice and theory perspectives.
In this paper we considered only the version of \pcst with cardinality constraint, but similar techniques can apply to other connectivity problems, particularly those for which the primal-dual framework have been used~\cite{Goemans1995}.
Although the algorithm is based on the primal-dual framework, it does not use an \lp formulation.
We note that the standard integer linear program for \kpcst~\cite{Han2017} has integrality gap of at least~$4$.

Johnson~\etal~\cite{Johnson2000} also considered the quota version of \kmst, in which each vertex $v$ has an associated non-negative integer weight $\vweight_v$ and a solution is a minimum-cost tree with any number of vertices, such that the total weight of the vertices in the tree is at least some given quota.
Note that \kmst is the special case in which the quota is $k$ and every vertex has weight~one.
A~small modification of our algorithm also leads to a $2$-approximation for the quota variant of~\kpcst.
To do this, build an equivalent instance of $\kpcst$ by replacing each vertex $v$ with $\vweight_v$ copies at the same location, and observe that, although the size of this instance is not necessarily polynomial, the growth-phase and picking-vertices can be simulated in polynomial time.

\bibliographystyle{abbrv}
\bibliography{main}

\clearpage
\appendix

%!TEX root = main.tex

\section{Proof of Lemma~\ref{appendix:bound_simple:lemma:2-approx}}
\label{appendix:bound_simple}

In this appendix, we prove Lemma~\ref{appendix:bound_simple:lemma:2-approx}.
We show that, if $\GW(0,\emptyset)$ returns a tree $T$ with at least $k$ vertices, then $T$ is a $2$-approximation.
The proof is adapted from~\cite{Feofiloff2010}.

Let $(T,\calb)$ be the tuple returned by $\GP(0,\emptyset)$, and remember that the growth-phase computes a corresponding laminar collection $\call$ associated with vector~$y$.
Also, let $\hatt$ be the output of $\GW(0,\emptyset)$, and observe that $\hatt$ is the tree returned by $\PP(T,\calb)$.

Recall that $\call(S)$ denotes the collection of subsets in $\call$ which contain some but not all vertices of a subset~$S$, and that ${\call_v}$ denote the collection of subsets in $\call$ which contain a vertex~$v$.
Thus, $\call_r(S)$ is the collection of subsets in $\call(S)$ which contain the root~$r$.

We bound the cost of $\hatt$ into two steps.
Lemma~\ref{appendix:bound_simple:lemma:edge_cost_bound} bounds the cost of edges of $\hatt$, and Lemma~\ref{appendix:bound_simple:lemma:penalty_cost_bound} bounds the penalty of vertices not in $\hatt$.

\begin{mylemma}
	\label{appendix:bound_simple:lemma:edge_cost_bound}
	Let $\hatt$ be the tree returned by $\GW(0,\emptyset)$, then
	\begin{align*}
		\ecost_{E(\hatt)}
		\leq
		~~~
		2 \smashoperator{\sum_{L \in \call(V(\hatt))\setminus \call_r}} y_L
	\end{align*}
\end{mylemma}

\begin{proof}
	Since every edge in $\hatt$ is tight for $(y,0)$, we have
	\begin{align*}
		\ecost_{E(\hatt)}
			= \smashoperator[r]{\sum_{e \in E(\hatt)}} \ecost_e
			= \sum_{e \in E(\hatt)} \smashoperator[r]{\sum_{L: e \in \delta(L)}} y_L
			= \sum_{L \in \call} |\delta_{\hatt}(L)| y_L.
	\end{align*}

	We prove by induction that, at the beginning of each iteration of~$\GP(0,\emptyset)$,
	\begin{align*}
		\sum_{L \in \call} |\delta_{\hatt}(L)| y_L
		\leq
				~~~
				2 \smashoperator{\sum_{L \in \call(V(\hatt))\setminus \call_r}} y_L
	\end{align*}

	At the beginning of the growth-phase, $y = 0$, then the inequality holds.
	Thus, assume that the inequality holds at the beginning of an iteration.
	Consider the laminar collection $\call$ at the beginning of this iteration, let~$\cala$ be the collection of active maximal subsets in $\call^*$ which contain some vertex of~$\hatt$, and let $\cali$ be the collection of non-active maximal subsets in~$\call^*$ which contain some vertex of~$\hatt$.

	If some maximal subset contains $V(\hatt)$, then neither side of the inequality changes, thus assume that no maximal subset contains $V(\hatt)$.
	Suppose that the variable $y_L$ of each active subset~$L$ is increased by~$\Delta$ in this iteration. Hence, the left side of the inequality  increases by $\sum_{A \in \cala}|\delta_{\hatt}(A)|\Delta$, and the right side increases by $2 |\cala \setminus \cala_r|\Delta $.
	We claim that
	\[
	\sum_{A \in \cala}|\delta_{\hatt}(A)| \leq 2 |\cala \setminus \cala_r|,
	\]
	and thus the inequality is maintained at the end of the iteration.

	Note that $\cala \cup \cali$ partition $V(\hatt)$.
	Thus, we can create a graph~$T'$ from~$\hatt$ by contracting each subset in $\cala \cup \cali$.
	As $\hatt$ is a tree and a subgraph of $T$, it follows that the graph $T'$ is a tree because $T$ is $\call$-connected by invariant~\ref{subsec:GP:lemma:GP_inv:edge_in_F_are_tight}.
	Let~$S$ be a non-active subset, and observe that $S$ is in $\calb$. Since $\hatt$ is pruned with~$\calb$, Corollary~\ref{subsec:PP:coro:PP_works} implies that the degree of $S$  on $\hatt$ is not one.
	Therefore, the degree of each vertex of $T'$ corresponding a subset in $\cali$ is at least $2$.
	It follows that
	\begin{align*}
		\sum_{A \in \cala}|\delta_{\hatt}(A)| + 2 |\cali|
			\leq \sum_{A \in \cala}|\delta_{\hatt}(A)| + \sum_{I \in \cali}|\delta_{\hatt}(I)|
			= 2 (|\cala| + |\cali| - 1).
	\end{align*}

	Since the subset containing the root $r$ is always active, exactly one subset in~$\cala$ contains $r$, hence ${|\cala \setminus \cala_r| = |\cala|-1}$, and thus
		${\sum_{A \in \cala}|\delta_{\hatt}(A)| \leq 2 |\cala \setminus \cala_r|}$,
	which shows the claim.
\end{proof}

\begin{mylemma}
	\label{appendix:bound_simple:lemma:penalty_cost_bound}
	Let $\hatt$ be the tree returned by $\GW(0,\emptyset)$, then
	\begin{align*}
		\vpen_{V \setminus V(\hatt)}
		\leq ~~~~
		 \smashoperator{\sum_{L \in \call[V \setminus V(\hatt)]}} y_L.
	\end{align*}
\end{mylemma}

\begin{proof}
	By Corollary~\ref{subsec:GP:corollary:GP_output}, every vertex not spanned by $T$ is contained in a subset in $\calb$, and since $\PP(T,\calb)$ only deletes from $T$ subsets in $\calb$, we have that every vertex not spanned by $\hatt$ is contained in a subset in $\calb$.
	Since $\calb$ is laminar, there is a collection $\{B_1,\cdots,B_m\}$ of subsets in $\calb$ which partitions the vertices in $V \setminus V(\hatt)$. Because  every subset in $\calb$ is tight for $(y,0)$,
	\begin{align*}
				\vpen_{V \setminus V(\hatt)}
			&= \vpen^0_{V \setminus V(\hatt)}
			= \sum_{i=1}^m \vpen_{B_i}^0
			= \sum_{i=1}^m \smashoperator[r]{\sum_{L \in \call[B_i]}} y_L
			\leq \smashoperator[r]{\sum_{L \in \call[V \setminus V(\hatt)]}} y_L.
			\qedhere
	\end{align*}
\end{proof}

Now, we can prove Lemma~\ref{appendix:bound_simple:lemma:2-approx}.

\appendixboundsimple*

\begin{proof}
	Let $L^*$ be the minimal subset in $\call$ containing all the vertices of~$T^*$.
	Also, let $\calp$ be the collection of subsets in $\call$
	which contain all the vertices of~$L^*$.
	Observe that $\{\call(L^*),\call[V \setminus L^*],\calp\}$ partitions~$\call$.
	Also, since $L^*$ contains~$r$, it follows that ${\calp \subseteq \call_r}$.

	Combining Lemmas~\ref{appendix:bound_simple:lemma:edge_cost_bound} and \ref{appendix:bound_simple:lemma:penalty_cost_bound} we have
	\begin{align*}
		\textstyle \ecost_{E(\hatt)} + 2 \vpen_{V \setminus V(\hatt)}
			\le  2 \sum_{L \in \call(V(\hatt))\setminus \call_r} y_L + 2 \sum_{L \in \call[V \setminus V(\hatt)]} y_L.
	\end{align*}

	Observe that the subsets considered in the terms are disjoint. Thus, we can simplify the indices of the summation as
	\begin{align*}
		(\call(V(\hatt))\setminus \call_r)
		\cup (\call[V \setminus V(\hatt)])
		\subseteq
		\call \setminus \call_r
		\subseteq
		\call \setminus \calp
		=
		\call(L^*) \cup \call[V \setminus L^*].
	\end{align*}

	Now, since $y$ respects~$\vpen^0$, using Lemma~\ref{subsec:approx:lemma:opt_lower_bound} for $\pot = 0$,
	\begin{align*}
		\ecost_{E(\hatt)} + 2 \vpen_{V \setminus V(\hatt)}
			&\textstyle \le  2  \sum_{L \in \call(L^*)} y_L + 2\sum_{L \in \call[V \setminus L^*]} y_L\\
			&\le  2 \ecost_{E(T^*)} + 2 \vpen_{L^* \setminus V(T^*)} + 2 \vpen^0_{V \setminus L^*}\\
			&=    2 \ecost_{E(T^*)} + 2 \vpen_{V \setminus V(T^*)}.
			\qedhere
	\end{align*}
\end{proof}

\end{document}